%% file: scmp-tr-full.tex
\documentclass{article}

\usepackage{etex}
\usepackage[english]{babel}
\usepackage[inference]{semantic}
\usepackage{amsmath}
\usepackage{amsthm}
\usepackage{mathpartir}
\usepackage{amssymb}
\usepackage{stmaryrd}
\usepackage{xspace}
\usepackage{latexsym}
\usepackage{hyperref}
\usepackage{ifthen}
\usepackage{mathtools}
\usepackage{color}
\usepackage[usenames,dvipsnames]{xcolor}
\usepackage{listings}
\usepackage{verbatim}
\usepackage[colorinlistoftodos]{todonotes}
\usepackage[capitalize]{cleveref}
\usepackage{etoolbox}
\usepackage{nameref}
\usepackage{multirow}
\usepackage{tabularx}

\input{./cmds}
\renewcommand{\longcol}[0]{0.7}
\renewcommand{\shortcol}[0]{0.15}
\renewcommand{\colleng}[0]{0.9}
\renewcommand{\fcol}[0]{0.4}

\hypersetup{ pdfpagemode=UseOutlines, colorlinks=true, linkcolor=red, citecolor=blue}

\author{ 
		 	Marco Patrignani$^1$
	\and	Dominique Devriese$^2$
	\and 	Frank Piessens$^2$
	\\
			$^1$ Max Planck Institute for Software Systems, Germany \\
	\and 	$^2$ iMinds-Distrinet, Dept.\ Computer Science, KU Leuven, Belgium \\ 
	\texttt{\{ first.last \}@\{ mpi-sws.org | cs.kuleuven.be \}}
}
\title{On Modular and Fully-Abstract Compilation -- Technical Appendix}

\begin{document}
\maketitle

\input{./bodies/abs}
\input{./bodies/intro} 	
\input{./bodies/langs}		
\input{./bodies/threat}
\input{./bodies/probs}	
\input{./bodies/comp}	
\input{./bodies/disc}	
\input{./bodies/proof}	
\input{./bodies/rw-conc}

\newpage

\appendix
\input{./bodies2/apps}

\input{./bodies2/jemlang}

\input{./bodies2/aimlang}

\input{./bodies2/trsem}
\input{./bodies2/algoinform}
\input{./bodies2/algoform}

\input{./bodies2/actualproofs}

\newpage
\bibliographystyle{plain}
\bibliography{./biblio}

\end{document}

%% file: cmds.tex
\newcommand{\mi}[1]{\ensuremath{\mathit{#1}}}

\newcommand{\mtt}[1]{\ensuremath{\mathtt{#1}}}

\newcommand{\mc}[1]{\ensuremath{\mathcal{#1}}}
\newcommand{\ms}[1]{\ensuremath{\mathsf{#1}}}
\newcommand{\mb}[1]{\ensuremath{\mathbb{#1}}}

\newcommand{\OB}[1]{\ensuremath{\overline{#1}}}

\newcommand{\divrjem}[0]{\ensuremath{\src{\Uparrow}}\xspace}
\newcommand{\divraim}[0]{\ensuremath{\trg{\Uparrow}}\xspace}
\newcommand{\termjem}[0]{\ensuremath{\src{\Downarrow}}\xspace}

\newcommand*{\QEDA}{\hfill\ensuremath{\blacksquare}}%

\AtEndEnvironment{problem}{\null\hfill\QEDA}
\AtEndEnvironment{example}{\null\hfill$\boxdot$}
\AtEndEnvironment{proof}{}

\Crefname{lstlisting}{Listing}{Listings}
\Crefname{problem}{Problem}{Problems}

\newenvironment{proofsketch}{\trivlist\item[]\emph{Proof Sketch}.\xspace}{\unskip\nobreak\hskip 1em plus 1fil\nobreak$\Box$\parfillskip=0pt\endtrivlist}
\newcommand{\compskel}[3]{\ensuremath{\bl{\llbracket} #1 \bl{\rrbracket}^{\bl{#2}}_{\bl{#3}}}}
\newcommand{\comp}[1]{\ensuremath{ \bl{(\!|}#1\bl{|\!)} }}

\newcommand{\compaim}[1]{\compskel{#1}{\jem}{\aim}}

\newcommand{\funname}[1]{\mtt{#1}}
\newcommand{\fun}[2]{\ensuremath{{\bl{\funname{#1}(}#2{\bl{)}}}}\xspace}
\newcommand{\dom}[1]{\fun{dom}{#1}}

\newcommand{\protname}[0]{protect}
\newcommand{\prot}[1]{\fun{\protname}{#1}}

\newcommand{\jem}[0]{\ms{J}{\scriptsize\ms{EM}}\xspace}
\newcommand{\aim}[0]{\ms{A}{\scriptsize\ms{IL}}\xspace}

\newcommand{\sys}[0]{\compaim{\src{Sys}}\xspace}

\newcommand{\PMA}[0]{\mi{PMA}\xspace}

\newcommand{\ASLR}[0]{ASLR\xspace}

\newcommand{\stlccol}[0]{ForestGreen}
\newcommand{\ulccol}[0]{WildStrawberry}
\newcommand{\neutcol}[0]{black}

\newcommand{\col}[2]{\ensuremath{{\color{#1}{#2}}}}
\newcommand{\src}[1]{\col{\stlccol}{{#1}}}		
\newcommand{\trg}[1]{\col{\ulccol }{{#1}}}		

\newcommand{\bl}[1]{\ensuremath{{\col{\neutcol}{#1}}}}

\newcommand{\fail}[0]{\mi{fail}}

\newcounter{typerule}
\crefname{typerule}{rule}{rules}

\newcommand{\typeruleInt}[5]{
	\def\thetyperule{#1}%
	\refstepcounter{typerule}%
	\label{tr:#4}%
  \ensuremath{\begin{array}{c}#5 \inference{#2}{#3}\end{array}} 
}
\newcommand{\typerule}[4]{
  \typeruleInt{#1}{#2}{#3}{#4}{\textsf{\scriptsize ({#1})} \\      }
}

\pgfdeclarelayer{background}
\pgfdeclarelayer{veryback}
\pgfdeclarelayer{veryback2}
\pgfdeclarelayer{veryback3}
\pgfdeclarelayer{back2}
\pgfdeclarelayer{foreground}
\pgfsetlayers{veryback3,veryback2,veryback,background,back2,main,foreground}

\newcommand{\myfig}[3]{\begin{figure} [!h]
#1
\vspace{-2em}
\caption{\label{fig:#2}#3}
\end{figure}}

\newcommand{\mytab}[3]{\begin{table} [!h]
\centering \small
#1

\caption{\label{tab:#2}#3}
\end{table}}

\newcommand{\mytabnocap}[1]{\begin{table}[!h]
\centering \small
#1
\end{table}}

\newcommand{\etal}[0]{\textit{et al.}\xspace}

\newcommand{\comm}[1]{}

\newcommand{\longcol}[0]{0.3}
\newcommand{\shortcol}[0]{0.11}
\newcommand{\colleng}[0]{0.49}
\newcommand{\fcol}[0]{0.2}

\newcommand{\algo}[1]{\ensuremath{\langle\!\langle#1\rangle\!\rangle}}
\newcommand{\skel}[1]{\fun{skeleton}{#1}}
\newcommand{\emul}[1]{\fun{emulate}{#1}}
\newcommand{\diff}[1]{\fun{diff}{#1}}

\newcommand{\traces}[2]{\ensuremath{\ms{Traces}_{#2}(#1)}}
\newcommand{\taim}[1]{\traces{#1}{\aim}}

\newcommand{\hl}[1]{\color{red!60}{#1}} 
\newcounter{line}

\newcommand{\assline}[1]{\trg{\arabic{line}}~~ \trg{\mtt{#1}}\stepcounter{line}}
\newcommand{\asm}[1]{\trg{\mtt{#1}}}

\newcommand{\xto}[1]{\ensuremath{\xrightarrow{~#1~}}}

\newcommand{\xtol}[1]{\ensuremath{\xrightarrow{~#1~}\low}}
\newcommand{\Xtol}[1]{\ensuremath{\xRightarrow{~#1~}\Low}}
\newcommand{\tol}[0]{\ensuremath{\to\low}}
\newcommand{\low}[0]{\ensuremath{\!\!\!\!\to} }
\newcommand{\Low}[0]{\ensuremath{\!\!\!\!\Rightarrow} }

\newcommand{\acp}[1]{\ms{#1}}
\newcommand{\blk}[0]{\ensuremath{(\unk,\trg{m},\trg{\OB{s}},\trg{h},\trg{t})}\xspace}
\newcommand{\unk}[0]{\ensuremath{\ms{unk}}}

\newcommand{\id}[0]{\mi{id}}

\newcommand{\toid}[0]{\ensuremath{\xto{\trg{\id}}}}
\newcommand{\tolid}[0]{\ensuremath{\xtol{\trg{\id}}}}

\definecolor{mygreen}{rgb}{0,0.6,0}
\definecolor{mygray}{rgb}{0.5,0.5,0.5}
\definecolor{mymauve}{rgb}{0.58,0,0.82}

\newcommand{\lstb}[1]{\ensuremath{\text{\texttt{\textbf{#1}}}}}
\newcommand{\lst}[1]{\ensuremath{\text{\texttt{{#1}}}}}

\lstdefinelanguage{Java} 
{morekeywords={abstract, all, and, as, assert, but, check, disj, else, exactly, extends, fact, for, fun, iden, if, iff, implies, in, Int, int, let, lone, module, no, none, not, one, open, or, part, pred, run, seq, set, sig, some, sum, then, univ, package, class, public, private, null, return, new, interface, extern, object, implements, System, static, super, try , catch, throw, throws, Unit, var, val, of, principal, trust},
sensitive=true,
keywordstyle=\bfseries\color{green!40!black},
commentstyle=\itshape\color{purple!40!black},
morecomment=[l][\small\itshape\color{purple!40!black}]{//},
identifierstyle=\color{blue},
stringstyle=\color{orange},
basicstyle=\small,
basicstyle={\small\ttfamily},
numbers=left,
numberstyle=\tiny\color{mygray},
tabsize=2,
numbersep=3pt,
breaklines=true,
lineskip=-2pt,
stepnumber=1,
captionpos=b,
breaklines=true,
breakatwhitespace=false,
showspaces=false,
showtabs=false,
float=!h,
columns=fullflexible,escapeinside={(*@}{@*)},
moredelim=**[is][\color{red!60}]{@}{@},
literate={->}{{$\to$}}1 {^}{{$\mspace{-3mu}\widehat{\quad}\mspace{-3mu}$}}1
{<}{$<$ }2 {>}{$>$ }2 {>=}{$\geq$ }2 {=<}{$\leq$ }2
{<:}{{$<\mspace{-3mu}:$}}2 {:>}{{$:\mspace{-3mu}>$}}2
{=>}{{$\Rightarrow$ }}2 {+}{$+$ }2 {++}{{$+\mspace{-8mu}+$ }}2
{<=>}{{$\Leftrightarrow$ }}2 {+}{$+$ }2 {++}{{$+\mspace{-8mu}+$ }}2
{\~}{{$\mspace{-3mu}\widetilde{\quad}\mspace{-3mu}$}}1
{!=}{$\neq$ }2 {*}{${}^{\ast}$}1 
{\#}{$\#$}1
}
\DeclareMathOperator\niff{\ensuremath{\mathrel{\iff\!\!\!\!\!\!\!\!\!\!/\ \ }}}
\DeclareMathOperator\nsimeq{\ensuremath{\mathrel{\not\simeq}}}
\DeclareMathOperator\nequiv{\ensuremath{\mathrel{\not\equiv}}}

\DeclareMathOperator\ceq{\ensuremath{\mathrel{\simeq_{\mi{ctx}}}}}
\DeclareMathOperator\nceq{\mathrel{\nsimeq_{\mi{ctx}}}}

\DeclareMathOperator\ceqjem{\src{\ceq}}                                      
\DeclareMathOperator\nceqjem{\src{\nceq}}                                    

\DeclareMathOperator\ceqaim{\trg{\ceq}}

\def\teqaux#1{\vcenter{\hbox{\ooalign{\hfil
       \raise6pt \hbox{\scriptsize{T}}\hfil\cr\hfil
       $=$}}}}
\def\teqa{\mathrel{\mathpalette\teqaux{}}}
\def\nteqa{\mathrel{\mathpalette\teqaux{}\!\!\!/\ }}

\DeclareMathOperator\teqaim{\ensuremath{\trg{\teqa}}}       
\DeclareMathOperator\nteqaim{\ensuremath{\trg{\nteqa}}}     

\def\ceqwaux#1{\vcenter{\hbox{\ooalign{\hfil
       \raise6pt \hbox{\scriptsize{w-b}}\hfil\cr\hfil
       $\ceq$}}}}
\def\ceqwf{\mathrel{\mathpalette\ceqwaux{}}}
\DeclareMathOperator\ceqw{\ceqwf}
\DeclareMathOperator\ceqwt{\trg{\ceqw}^{\aim}_{\jem}}

\def\praux#1{\vcenter{\hbox{\ooalign{\hfil
       \raise6pt \hbox{$\sqsubset$}\hfil\cr\hfil
       $\sim$}}}}
\def\pr{\mathrel{\mathpalette\praux{}}}

\DeclareMathOperator\preqaim{\trg{\pr}}

\newcommand{\labelfont}[1]{\ensuremath{\asm{#1}}}
\newcommand{\cl}[1]{\ensuremath{\labelfont{call}\ #1\trg{?}}}
\newcommand{\cb}[1]{\ensuremath{\labelfont{call}\ #1\trg{!}}}
\newcommand{\rt}[1]{\ensuremath{\labelfont{ret}\ #1\trg{!}}}
\newcommand{\rb}[1]{\ensuremath{\labelfont{ret}\ #1\trg{?}}}

\newcommand{\clgen}[0]{\ensuremath{\labelfont{call}}}
\newcommand{\rtgen}[0]{\ensuremath{\labelfont{ret}}}

\newcommand{\myparagraph}[1]{ \smallskip \noindent\noindent\textit{#1}~}
\DeclareMathOperator\compat{\ensuremath{\raisebox{1mm}{$\frown$}}}

\theoremstyle{plain}
\newtheorem{theorem}{Theorem}

\newtheorem{lemma}{Lemma}
\newtheorem{corollary}{Corollary}

\newtheorem{example}{Example}

\newtheorem{assumption}{Assumption}

\newtheorem{problem}{Problem}
\Crefname{assumption}{Assumption}{Assumption}
\crefname{assumption}{Assumption}{Assumption}

\theoremstyle{definition}
\newtheorem{definition}{Definition}

\newcommand{\mylink}[1]{\ensuremath{\fun{link}{#1}}}

%% file: bodies/abs.tex

\begin{abstract}
Secure compilation studies compilers that generate target-level components that are as secure as their source-level counterparts.
\emph{Full abstraction} is the most widely-proven property when defining a secure compiler.

A compiler is \emph{modular} if it allows different components to be compiled independently and then to be linked together to form a whole program.

Unfortunately, many existing fully-abstract compilers to untyped machine code are not modular.
So, while fully-abstractly compiled components are secure from malicious attackers, if they are linked against each other the resulting component may become vulnerable to attacks.

This paper studies how to devise modular, fully-abstract compilers.
It first analyses the attacks arising when compiled programs are linked together, identifying security threats that are due to linking.
Then, it defines a compiler from an object-based language with method calls and dynamic memory allocation to untyped assembly language extended with a memory isolation mechanism.
The paper provides a proof sketch that the defined compiler is fully-abstract and modular, so its output can be linked together without introducing security violations.
\end{abstract}

%% file: bodies/intro.tex
\begin{center}\small
This paper uses colours to distinguish elements of different languages; please print this in colour.
\end{center}

\section{Introduction}\label{sec:intro} 
A compiler is a tool that, among other things, translates programs written in a source language into programs written in a target language.
A compiler is \emph{secure} when it preserves all security properties of the components (i.e., partial programs) it inputs in the components it outputs.
Secure compilation studies compilers that generate target-level components that are as secure as their source-level counterparts.
Full abstraction is the most widely-adopted property to prove when defining a secure compiler~\cite{abadiFa,abadiLayout,Jagadeesan,scoo-j,AhmedFa,ahmedCPS,fstar2js,scoojoin,bugliesi,popldomi,nonintfree}.
A fully-abstract compiler preserves and reflects observational equivalence between the source components it inputs and the target components it produces.
Observational equivalence captures security properties such confidentiality, integrity, etc, and since a fully-abstract compiler preserves all of them, it is a secure compiler.

A compiler is \emph{modular} when it operates on components and its output can be \emph{linked} together into larger components (and possibly into whole programs).

With the advent of machine-code level security architectures (e.g., protected modules architectures (\PMA)~\cite{fidesArch,sancus,intel}, capability machines~\cite{cheri,capsicum}, micro-policies enforcing architectures~\cite{crashsafemachine}, address space layout randomisation~\cite{abadiLayout}, etc.) researchers have investigated secure (fully-abstract) compilation to such architectures~\cite{scoo-j,jannis,abadiLayout,Jagadeesan}.

For secure compilation to \PMA (informally, a low-level memory isolation mechanism), only compilation of a single protected module has been considered. 
Generalizing this to multiple modules, supporting modular compilation and linking of protected modules at machine code level, would be useful for a number of reasons: (i) code is easier to develop and distribute in standalone components and (ii) merging modules in a single one is not always possible nor desirable.
Concerning (i), standard arguments for separate compilation apply.
Particularly, it is more efficient to compile only the modules that changed and link it against the previously compiled other code as opposed to recompiling the entire code base. 
Moreover, it is common in practice and useful to use libraries that are built from safe source code but only available in compiled form. 
When targeting \PMA with attestation capabilities, a benefit of  having each module in a different protected module is that this makes it possible to attest them independently.
Concerning (ii), there can be several reasons why merging more modules in a single one is not desirable; even with a secure compiler, it can still be useful to protect securely compiled components from each other.
Consider a source language that allows programmers to write both safe and unsafe code at the same time (e.g., Rust, C\# unsafe blocks or the language considered by Juglaret \etal~\cite{catalin}).
A desirable property for a compiler for such a language is that unsafe code does not affect the safe one.
In this case, one can foresee splitting the safe and the unsafe program parts in two different modules and using the secure compiler presented in this paper for the safe part.
In this way, both the safe and the unsafe code cannot be directly tampered with by an attacker and the unsafe code can not affect the safe one.
Another setting where merging multiple modules is undesirable is when a programmer provides a hand-optimised securely-compiled component that performs better than the ones obtained by a secure compiler, but with the same security guarantees.
In this case, merging such a component with other ones in a single module seems undesirable, as the hand-made assembly can be error-prone and maliciously-crafted.
However, by using a modular secure compiler, other components can still interact with the hand-optimised one and no one's security can be tampered.



Extending the compiler to support modular compilation is, however, surprisingly complicated for a number of reasons:
\begin{enumerate}
\item\label{pt:probone} with a single protected module, all run-time meta information about the execution that needs to be protected (e.g.,  the call stack, or dynamic type information) can be stored within that single protected module. 
For a modular compiler, that state must be divided over the various protected modules, or -- possibly -- stored in a centralized trusted protected module.

\item\label{pt:probtwo} with a single protected module, object references are either private to the module or public. 
With support for multiple protected modules, object references can also be shared between some modules and still be unknown to other ones. 
This requires a mechanism for protecting object references that is more elaborate than what the existing secure compiler to \PMA uses.
\end{enumerate}

There are several interesting approaches to address these challenges.
This paper investigates how these issues can be addressed constructing a modular fully-abstract compiler for \PMA in the style of Intel SGX~\cite{intel} \emph{without} additional hardware support.
Even if this turns out to be intricate, we believe there is value in this approach, as this kind of hardware support is available on commercially available systems today, whereas more advanced hardware support is only available in research prototypes.

An alternative approach to modular, secure compilation would be to investigate what kind of novel hardware security architectures can help in addressing these challenges. 
This is the approach taken by Juglaret \etal~\cite{jannis} and it can likely lead to a secure compilation scheme. 
An important downside is that it may take a while before such hardware extensions are available in mainstream systems.

The main contribution of this paper is \compaim{\cdot}, a modular and fully-abstract compiler to SGX-like \PMA.
The source language of \compaim{\cdot} is \jem, an object-based imperative language with method calls and dynamic memory allocation.
The target language of \compaim{\cdot} is \aim, untyped assembly language with an explicit linking mechanism extended with \PMA.
\jem, \aim and \PMA are formalised in \Cref{sec:langform}.
Explicitly considering linking between assembly components generates a number of previously-unobserved problems that are presented in \Cref{sec:probs}.
The main contribution of this paper is the formalisation of \compaim{\cdot}, presented in \Cref{sec:seccomp}.
Supporting additional language features is discussed in \Cref{sec:disc}.
The proof sketch of \compaim{\cdot} being fully-abstract and modular is presented in \Cref{sec:proof}. 
Finally, \Cref{sec:rw} discusses related work and \Cref{sec:conc} concludes.

%% file: bodies/langs.tex

\section{Languages Formalisation}\label{sec:langform}
This section describes \jem (\Cref{sec:informje}) and \aim (\Cref{sec:informaim}), respectively the source and the target language of the secure compiler \compaim{\cdot}.

\subsection{The Source Language \jem}\label{sec:informje}
\jem is a strongly-typed, single-threaded, object-based imperative language that has \lstb{private} fields and \lstb{public} methods; it does not allow any undefined behaviour to arise.
\jem is presented in a \src{green} font (\Cref{fig:jemsyn}). 

\myfig{
\begin{align*}
&\mi{component} 		
	& \src{\mc{C}} ::=&\ \src{\OB{C}}
	 	\\
&\mi{classes}
	&\src{C} ::=&\ \lstb{import}\ \src{\OB{I}};\src{\OB{X}};\ \lstb{class}\ \src{c} \{\src{K}\ \src{\OB{F_t}}\ \src{\OB{M}}\}; \src{\OB{O}}
		\\
&\mi{objects}
	& \src{O} ::=&\ \lstb{object}\ \src{o}:\src{t} \{\src{\OB{F}}\}	
		\\
&\begin{multlined}\mi{class}\\\mi{declarations}\end{multlined}
	& \src{I} ::=&\ \lstb{class-decl}\ \src{c} \{\src{\OB{m : M_t}}\}		
		\\
&\begin{multlined}\mi{object}\\\mi{declarations}\end{multlined}
	& \src{X} ::=&\ \lstb{obj-decl}\ \src{o}:\src{t};
	\\
&\mi{methods}		
	& \src{M} ::=&\ \lstb{public}\ \src{m}(\src{\OB{x}}):\src{M_t}\ \{\lstb{return}\ \src{E};\}		
	\\
&\mi{signatures}	
	& \src{M_t} ::=&\ \src{t}(\src{\OB{t}})\to\src{t} 
		\\
&\mi{fields}			
	& \src{F} ::=&\ \lstb{private}\ \src{f}=\src{v}		
		\\
&\mi{field\ types}
	&\src{F_t} ::=&\ \src{f}:\src{t}
		\\
&\mi{constructors}	
	& \src{K} ::=&\ \src{c}(\src{\OB{f}}:\src{\OB{t}})\ \{\lstb{this}.\src{\OB{f'}}=\src{\OB{f}}\}
	\\
&\mi{types}			
	& \src{t} ::=&\ \src{\lstb{Unit}} \mid \src{\lstb{Bool}} \mid \src{\lstb{Int}} \mid \src{c} \mid \src{\lstb{Obj}} 
		\\
&\mi{values}		
	& \src{v} ::=&\ \src{\lst{unit}} \mid \src{\lst{true}} \mid \src{\lst{false}} \mid \src{n} \mid\src{o} 
		\\
&\mi{operations}		
	& \src{\mtt{op}} ::=&\ \src{+} \mid \src{-}\mid \src{==} \mid \cdots		
		\\
&\mi{expressions}	
	& \src{E} ::=&\ \src{v} \mid \src{x} \mid \src{E}.\src{f} \mid \src{E}.\src{f}=\src{E} \mid \src{E}.\src{m}(\src{\OB{E}}) 
	\\
	&&\mid&\ \src{E\ \mtt{op}\ E}  \mid  \lstb{new}\ \src{t}(\src{\OB{E}}) \mid \src{E};\src{E}  \mid \lstb{this} 
	\\
	&&\mid&\ \lstb{if}\ (\src{E})\ \{\src{E}\}\ \lstb{else}\ \{\src{E}\} \mid \lstb{exit}\ \src{E}
	\\
	&&\mid&\ \lstb{instanceof}(\src{E}:\src{c}) \mid \lstb{var}\ \src{x} : \src{t} = \src{E}
\end{align*}
}{jemsyn}{Syntax of \jem; lists of elements $a_1\cdots a_n$ are denoted as~\OB{a}.}
 A class \src{C} declares (external) classes and objects it requires (these are called \lstb{import} requirements) then it defines its constructor, fields, methods and objects implementing that class.
Objects of a class can only be allocated by methods of that class (so cross-component memory allocation happens via factory methods).
Class declarations \src{I} define class signatures, i.e., the class name and the methods implemented by that class. 
Object declarations \src{X} are references to objects implementing a different class.
A \jem component \src{\mc{C}} is a collection of classes \src{\OB{C}}.
If all \lstb{import} requirements of \src{\OB{C}} are satisfied by some other class in \src{\OB{C}}, then \src{\OB{C}} is a whole program.
Two components $\src{\mc{C}_1}$ and $\src{\mc{C}_2}$ satisfy each other, denoted with $\src{\mc{C}_1}\compat\src{\mc{C}_2}$, if all \lstb{import} requirements of \src{\mc{C}_1} are classes and objects in \src{\mc{C}_2} and vice-versa.

The top of the \jem class hierarchy is \src{\lstb{Obj}}, a class defining no methods.
All classes implicitly extend \src{\lstb{Obj}}; \jem does not provide any other form of inheritance.
Primitive types are \src{\lstb{Unit}}, inhabited by \src{\lst{unit}}, \src{\lstb{Bool}}, inhabited by \src{\lst{true}} and \src{\lst{false}} and \src{\lstb{Int}}, inhabited by natural numbers \src{n}. 
Identifiers for classes \src{c}, objects \src{o}, methods \src{m}, fields \src{f} and variables \src{x} are taken from distinct denumerable sets.

The semantics of \jem is standard and unsurprising; it is omitted for space reasons.

The security mechanism of \jem is given by \lstb{private} fields, which can be used to define security properties such as confidentiality and integrity (as defined in \Cref{sec:threatmodel}).

\subsubsection{Contextual Equivalence for \jem}\label{sec:ceqjem}
To reason about the behaviour of \jem components, contextual equivalence is used~\cite{lcfConsidered}.
Contextual equivalence is the coarsest relation that tells when two components are behaviourally equivalent; its definition (\Cref{def:ceqjem}) is rather standard~\cite{AhmedFa,ahmedCPS,scoo-j,fstar2js,abadiLayout,Jagadeesan,nonintfree,faEHM,gcFA}.

Contexts \src{\mb{C}} are partial programs with a hole ($[\cdot]$), formally $\src{\mb{C}} ::= \src{\OB{C}} [\cdot]$.
The plugging of a component \src{\mc{C}} in a context \src{\mb{C}}, denoted as $\src{\mb{C}}[\src{\mc{C}}]$, returns a whole program $\src{\mb{C}};\src{\mc{C}}$.
There are two (common) assumptions for the plugging to succeed.
The first one is that $\src{\mc{C}}\compat\src{\mb{C}}$, the second one is that \src{\mc{C}} and \src{\mb{C}} are well-typed.
If any assumption is not upheld, the plugging returns the empty program.
\begin{definition}[Contextual equivalence for \jem]\label{def:ceqjem}
$\src{\mc{C}_1}\ceqjem\src{\mc{C}_2}\triangleq\forall\src{\mb{C}}, \src{\mb{C}}[\src{\mc{C}_1}]\divrjem\iff\src{\mb{C}}[\src{\mc{C}_2}]\divrjem$, where $\divrjem$ means divergence, i.e., the execution of an unbounded number of reduction steps.
\end{definition}

\subsection{\PMA and the Target Language \aim}\label{sec:informaim}
\aim (acronym of \ms{A}ssembly plus \ms{I}solation and \ms{L}inking) is a low-level language that models a von Neumann machine enhanced with \PMA (\Cref{sec:pma}), with an idealised form of cryptographic nonces (\Cref{sec:coreaim}) and with an explicit linking mechanism (\Cref{sec:aimmodules}).
\aim is the adaptation of an analogous language that had no linking mechanism, no idealised cryptographic nonces and a single \PMA module~\cite{scoo-j,llfatr-j}.
After presenting \PMA and \aim, this section describes the semantics (\Cref{sec:dynsem-aim}) and contextual equivalence for \aim (\Cref{sec:ceqaim}).

\subsubsection{\PMA}\label{sec:pma}
\PMA is an assembly-level isolation mechanism based on program counter-based memory access control.
\PMA can be implemented in software (e.g., via a hypervisor)~\cite{trustVisor,nizza,fidesArch} or in hardware~\cite{intelattestation,intel,sancus}; Intel is bringing it to mainstream processors with the Intel SGX instruction set~\cite{intel}.
\PMA logically divides the memory space into several protected and one unprotected section; a protected section is called \emph{protected module}, each module has a unique \emph{module id}.
All protected sections are further divided into a \emph{code} and a \emph{data} section.
Code sections contain a variable number of \emph{entry points}: the only protected addresses to which instructions in unprotected memory or in other protected sections can jump.
Entry points allow code from within the module to interoperate with external code.
Data sections are only accessible from within the protected code section of the same module.
The table below summarises the \PMA access control model.
\begin{center}
\begin{tabularx}{\colleng\textwidth}{| X | c | c | c | c |}
\hline
\multirow{2}{*}{\raisebox{-4pt}{From}~$\setminus$~\raisebox{4pt}{To} }
	& \multirow{2}{*}{Unprotected }
		& \multicolumn{3}{| c |}{Protected } 
	\\

	& 
		& Entry Point & Code & Data 
		\\
\hline
Unprotected 
	& r w x 
		& x & & 
			\\
\hline
\multirow{4}{*}{Protected}
	& \multirow{4}{*}{r w x} 
		&\multicolumn{3}{| c |}{ Same id}  
			\\

	&
		& r x & r x & r w 
			\\
			\cline{3-5}
	&
		&\multicolumn{3}{| c |}{ Different id}  
			\\

	&
		& x & & 
			\\
\hline
\end{tabularx}
\end{center}

\subsubsection{Core \aim}\label{sec:coreaim}
\aim is presented in a \trg{pink} font (\Cref{fig:aimsym}).

\myfig{
\begin{align*}
&\mi{numbers} &\trg{n} ::=&\ \trg{n}\in\mb{N} 
		\\	
&\mi{words}& \trg{w} ::=&\ \trg{n} \mid \trg{\sigma} \mid \trg{\pi}
		\\
&\mi{addresses} &\trg{a} ::=&\ (\trg{id},\trg{n})
		\\
&\mi{memories} & \trg{m}::=&\ \OB{\trg{a\mapsto \trg{w}}}
	\\
&\mi{flags}& \trg{f}::=&\ \trg{\ms{ZF}}\mapsto\trg{0 \mid 1};\trg{\ms{SF}}\mapsto\trg{0 \mid 1}
		\\
&\mi{register\ files}& \trg{r} ::=&\ \OB{\asm{r}\mapsto\trg{w}}
	\\
&\mi{module\ ids} & \trg{\mi{id}} \in &\ \trg{\mc{ID}} \subset \mb{N}
		\\
&\mi{symbolic\ nonces} &\trg{\pi} \in&\ \trg{\Pi}
	\\
& \mi{module\ descriptors} &\trg{s}::=&\ (\trg{id},\trg{n_c},\trg{n})
		\\
&\mi{nonce\ oracles} &\trg{h} ::=&\ \trg{\OB{\pi}}
	\\
&\mi{symbols} & \trg{\sigma},\trg{\iota}  \in&\ \trg{\mc{S}}
	\\
&\mi{instructions}& \trg{i}\in &\ \trg{\mc{I}} \subset \mb{N}
\end{align*}
}{aimsym}{Formalisation of \aim (part 1: language).}
Words \trg{w} are either natural numbers \trg{n} (including instruction encodings and module ids), symbols \trg{\sigma} (explained in \Cref{sec:aimmodules}) or symbolic nonces \trg{\pi}.
Addresses \trg{a} are pairs of natural numbers and module ids, so a module with id \trg{id} has an infinite memory starting from address $(\trg{\trg{id},\trg{0}})$.
Unprotected memory has id~\trg{0}.
The registers file \trg{r} contains an unbounded number of registers \asm{r}.
The flags register \trg{f} contains a sign (\trg{\ms{SF}}) and a zero flag (\trg{\ms{ZF}}) that are set by arithmetic and testing instructions.
Memories $\trg{m}$ are lists of bindings from addresses to words.
Module descriptors $\trg{s}$ are triples that define a module memory layout: \trg{id} indicates the id of the module, $\trg{n_c}$ is the length (in number of addresses) of the code section and $\trg{n}$ is the number of entry points. 
Between each entry point there are a fixed amount of addresses that are not entry points, indicate this number of addresses with $\trg{\mc{N}_w}$.
The first entry point is located at address \trg{0}, the second one at address \trg{\mc{N}_w} and so on; the last one is located at address $\trg{n}\cdot\trg{\mc{N}_w}$ such that $\trg{n}\cdot\trg{\mc{N}_w}<\trg{n_c}$.
Symbolic nonces \trg{\pi} are non-guessable, unforgeable tokens; they model what a program could create using cryptographic primitives or unguessable randomisation~\cite{dolevyao,Sumii:2004:BDS:964001.964015}. For the sake of simplicity, they are considered to be \trg{0} for arithmetic operators.\footnote{
	Having arithmetic operations affect nonces would only let us model guessing attacks on them, but they are assumed to be resilient to these attacks so we chose this option for simplicity.
}
Nonce oracles \trg{h} are used as suppliers of fresh nonces for the \asm{new} instruction (see \Cref{tr:evalnew} in \Cref{fig:aim-dyn-sem}).  
Instructions $\trg{i}$ are elements of the set \trg{\mc{I}} (\Cref{tab:instructions}) that define the programming language executed on the architecture.
\mytab{
\begin{tabular}{p{\shortcol\textwidth} | p{\longcol\textwidth}}
\hline
\hline
\asm{movl\ r_d\ r_s\ r_i}
 	& Load the word from the address in registers (\asm{r_s},\asm{r_i}) into register \asm{r_d}.
 		\\
\asm{movs\ r_d\ r_s\ r_i}
	& Store the contents of register \asm{r_s} at the address found in registers (\asm{r_d},\asm{r_i}).
	\\
\asm{movi\ r_d\ k}
	& Load the constant value \asm{k} into register \asm{r_d}.
		\\
\asm{add\ r_d\ r_s}
	& Write \asm{r_d + r_s} into register \asm{r_d} and set the \trg{\ms{ZF}} flag accordingly.
		\\  	
\asm{sub\ r_d\ r_s}
	& Write \asm{r_d - r_s} into register \asm{r_d} and set both the \trg{\ms{ZF}} and the \trg{\ms{SF}} flags accordingly.
		\\  	
\asm{cmp\ r_s\ r_d}
	& Compare \asm{r_s} and \asm{r_d} and set the flags accordingly.
		\\  	
\asm{jmp\ r_d\ r_i}
	& Jump to the address located in register \asm{r_d} in the module with id \asm{r_i}.
		\\
\asm{je\ f_i\ r_i}
	& If flag \asm{f_i} is set, jump to the address in register \asm{r_i} in the current module.
		\\
\asm{zero}
	& Set all registers to \trg{0}.
		\\
\asm{new\ r_d}
	& Initialise register \asm{r_d} with a fresh symbolic nonce.
		\\
\asm{halt}
	& Stop the execution.
\end{tabular}
}{instructions}{Instruction set \trg{\mc{I}}.}


\subsubsection{Modules, Programs and Linking (\Cref{fig:aimsym2})}\label{sec:aimmodules}

To deal with linking, \aim has symbols \trg{\sigma}, \trg{\iota}. 
Symbols in a module are placeholders for words that will be filled when the module is linked against another one; symbols can be found in memory and in the symbol table.
A symbol table \trg{t} contains exported method and object bindings as well as required method and object bindings.
An exported method binding \trg{EM} (resp. object binding \trg{EO}) maps a method name \src{m} and its type \src{M_t} (resp. an object name \src{o} and its class \src{c}) to the address \trg{a} where the method (resp. class) is located.
A required method binding \trg{RM} (resp. object binding \trg{RO}) maps a method name \src{m} and its type \src{M_t} (resp. an object name \src{o} and its class \src{c}) to the id symbol \trg{\iota} and the number symbol \trg{\sigma} used for it.
A method (resp. object) export binding fulfils a method (resp. object) requirement if both bind the same method name and type (resp. object name and class).
\myfig{
\begin{align*}
&\mi{modules} &\trg{M} ::=& (\trg{m};\trg{s};\trg{t})
	&
&\mi{programs} &\trg{P} ::=& (\trg{m};\trg{\OB{s}};\trg{t})
\end{align*}
\vspace{-2.1em}
\begin{align*}
&\mi{symbol\ tables} &\trg{t} ::=&\ \trg{\OB{EM}};\trg{\OB{EO}};\trg{\OB{RM}};\trg{\OB{RO}}
	\\
&\mi{exported\ methods} & \trg{EM} ::=&\ \src{m}:\src{M_t}\mapsto\trg{a}
	\\
&\mi{exported\ objects} & \trg{EO} ::=&\ \src{o}:\src{c}\mapsto\trg{n}
	\\
&\mi{required\ methods} & \trg{RM} ::=&\ \src{m}:\src{M_t}\mapsto\trg{\iota};\trg{\sigma}
	\\
&\mi{required\ objects} & \trg{RO} ::=&\ \src{o}:\src{c}\mapsto\trg{\sigma}
\end{align*}
}{aimsym2}{Formalisation of \aim (part 2: linking).}

A single \aim module \trg{M} is a triple listing a memory \trg{m}, its module descriptor \trg{s} and a symbol table.
An \aim program \trg{P} is a collection of modules whose module descriptors are listed in \trg{\OB{s}}.
If the required bindings of a program are empty lists, then the program is whole.
When a program is not whole it is partial, i.e., it is a collection of modules with a symbol table with unfulfilled requirements. 
Informally, two \aim programs $\trg{P_1}$ and $\trg{P_2}$ satisfy each other, denoted with $\trg{P_1}\compat\trg{P_2}$, if all the required bindings of \trg{P_1} are fulfilled by a binding in \trg{P_2} and vice-versa.
Modules and programs are well-formed if the only symbols in their memory are those captured in the symbol table.
In the following, we only consider well-formed programs and modules.

Two memories (resp. two module descriptors) agree if their domains are disjoint (resp. if their ids are distinct).
Two \aim modules $\trg{P_1}$ and $\trg{P_2}$ can be joined together, denoted with $\trg{P_1}\trg{+}\trg{P_2}$, only if the two memories and if the memory descriptors agree.
Joining of modules results in a program that is the concatenation of the memories, memory descriptors and symbol tables of $\trg{P_1}$ and $\trg{P_2}$.
Joining of symbol tables results in a new symbol table.
If a requirement (method or object) is fulfilled by a requirement in another table, that requirement is removed from the resulting table.
Moreover, in the resulting memory, all symbols of fulfilled requirements are replaced with what is bound by the fulfiller requirement.
Consider $\trg{P_1}$ to require method $\src{m}:\src{M_t}\mapsto\trg{\iota};\trg{\sigma}$ and \trg{P_2} to export method $\src{m}:\src{M_t}\mapsto(\trg{id,n})$.
When merging $\trg{P_1}$ and $\trg{P_2}$, all occurrences of \trg{\iota} and \trg{\sigma} in the memory of \trg{P_1} will be replaced with \trg{id} and \trg{n}, respectively.
The resulting symbol table will no longer have the requirement for \src{m: {M_t}} as that has been fulfilled.

\subsubsection{Dynamic Semantics of \aim}\label{sec:dynsem-aim}
The dynamic semantics of whole programs ($\tol$) relates program states $(\trg{p},\trg{r},\trg{f},\trg{m},\trg{{s}},\trg{h})$ where \trg{p} is the program counter (i.e., an address).
This semantics relies on another one for modules ($\tolid$); the latter tells when a module with id \trg{id} performs a reduction (\Cref{tr:evalutop}).
\Cref{fig:aim-dyn-sem} shows an excerpt of the reduction rules.
Memory access is indicated as: $\trg{m}(\trg{a}) = \trg{w}$; memory update is indicated as $\trg{m}[\trg{a}\mapsto\trg{w}]$.
The same notation is adopted for register files and flags register access and update.
\begin{figure}
\centering
\small
	\typerule{Eval-new}{
		\trg{p}=(\trg{id},\trg{n})
		&
		\trg{m}(\trg{p}) = \asm{new\ r_d}
		&
		\trg{r'} = \trg{r}[\asm{r_d}\mapsto\trg{\pi}]
		&
		\trg{\pi}\notin\trg{h}
	}{
		(\trg{p},\trg{r},\trg{f},\trg{m},\trg{{s}},\trg{\pi};\trg{h})\tolid((\trg{id},\trg{n+1}),\trg{r'},\trg{f},\trg{m},\trg{{s}},\trg{h})
	}{evalnew}\quad
	\typerule{Eval-module}{
		(\trg{p},\trg{r},\trg{f},\trg{m},\trg{s},\trg{h}) \tolid (\trg{p'},\trg{r'},\trg{f'},\trg{m'},\trg{s'},\trg{h'})
		\\
		\trg{p}\vdash\acp{currentModule}(\trg{\OB{s}},\trg{s})
		&
		\trg{p}=(\trg{id},\trg{n})
	}{
		(\trg{p},\trg{r},\trg{f},\trg{m},\trg{\OB{s}},\trg{h}) \tol (\trg{p'},\trg{r'},\trg{f'},\trg{m'},\trg{\OB{s}},\trg{h'})
	}{evalutop}\quad
	\typerule{Eval-movs}{
		\trg{p}=(\trg{id},\trg{n})
		&
		\trg{m}(\trg{p})=  (\asm{movs\ r_d\ r_s\ r_i})
		\\
		\trg{\OB{s}}\vdash\acp{writeAllowed}(\trg{n}, (\trg{r}(\asm{r_d}),\trg{r}(\asm{r_i})))
		\\
		\trg{m'}=\trg{m}[(\trg{r}(\asm{r_d}),\trg{r}(\asm{r_i}))\mapsto\trg{r}(\asm{r_s})]
	}{
		(\trg{p},\trg{r},\trg{f},\trg{m},\trg{\OB{s}},\trg{h}) \tol ((\trg{id},\trg{n+1}),\trg{r},\trg{f},\trg{m'},\trg{\OB{s}},\trg{h})
	}{evalmovs}\quad
	\typerule{Eval-jmp}{
		\trg{p}=(\trg{id},\trg{n})
		&
		\trg{m}(\trg{p})= (\asm{jmp\ r_d\ r_i})
		&
		\trg{n'}= \trg{r}(\asm{r_d})
		\\
		\trg{id'}= \trg{r}(\asm{r_i})
		&
		\trg{r'}=\trg{r'}[\asm{r_0} \mapsto \trg{id}]
		&
		\trg{\OB{s}}\vdash\acp{validJump}(\trg{p},(\trg{id'},\trg{n'}))
	}{
		(\trg{p},\trg{r},\trg{f},\trg{m},\trg{s},\trg{h})\tol((\trg{id'},\trg{n'}),\trg{r'},\trg{f},\trg{m},\trg{s},\trg{h})
	}{evaljmp}\quad
\caption{An excerpt of the dynamic semantics of \aim{}}
\label{fig:aim-dyn-sem}
\end{figure}
Assumptions of the form $\trg{\OB{s}}\vdash\acp{name}(\trg{a},\trg{a'})$ model the enforcement of the \PMA access control policy; their names are self-explanatory.
For example, $\trg{s}\vdash\acp{writeAllowed}(\trg{n},\trg{n'})$ tells if an address \trg{n'} can be written from another address \trg{n} and $\trg{\OB{s}}\vdash\acp{validJump}(\trg{a},\trg{a'})$ tells if address $\trg{a'}$ is executable ($x$) from address \trg{a}.
$\trg{p}\vdash\acp{currentModule}(\trg{\OB{s}},\trg{s})$ tells whether the program counter \trg{p} is in a specific module whose descriptor is \trg{s} where \trg{s} is within \trg{\OB{s}}.
Instruction \asm{jmp} is the only instruction that can perform cross-module jumps (\Cref{tr:evaljmp}).
When a cross-module jump is performed, the \PMA architecture inserts the \trg{id} of the caller in register \asm{r_0}.
This functionality is called caller-callee authentication; some \PMA implementations provide it~\cite{fidesArch,sancus} and it is achievable on Intel SGX by means of the EREPORT instruction~\cite{intelsgxisa}.
Whenever the execution of an instruction violates the \PMA access control policy (e.g., jumping to a module address that is not an entry point), the execution is terminated. 

\subsubsection{Contextual Equivalence for \aim}\label{sec:ceqaim}
Analogously to \jem, to reason about \aim components we define when they are contextually equivalent.
Contexts \trg{\mb{P}} are partial programs with a hole ($[\cdot]$) that can be filled with a component in order to create a whole program, formally $\trg{\mb{P}} ::= \trg{P}[\cdot]$. 
Plugging a component \trg{P} in a context \trg{\mb{P}}, denoted as $\trg{\mb{P}}[\trg{P}]$, returns a whole program $\trg{\mb{P}}\trg{+}\trg{P}$, given that $\trg{P}\compat\trg{\mb{P}}$.

Contextual equivalence for languages with oracles relies on a notion of contextual preorder~\cite{oracle} which is needed to ensure obvious equivalences are satisfied (\Cref{ex:preorder}).
\begin{example}[Need for preorders]\label{ex:preorder}
Consider a module \trg{P_{n2}} that generates two nonces and returns only one and another one \trg{P_{n1}} that generates and returns one nonce only.
These programs are intuitively equivalent, but a contextual equivalence definition analogous to that for \jem simply does not yield this fact because we cannot quantify over all oracles and use the same oracle for both programs.
Given that \trg{P_{n1}} runs with an oracle \trg{h_1}, \trg{P_{n2}} has the same behaviour if it runs with an oracle \trg{h_2} that has all elements of \trg{h_1} interleaved with fresh elements.
Dually, given that \trg{P_{n2}} runs with an oracle \trg{h_2}, \trg{P_{n2}} has the same behaviour if it runs with an oracle \trg{h_1} that has only the even-numbered elements of \trg{h_2}.
\end{example}

\begin{definition}[Contextual preorder for \aim]\label{def:preaim}
$\trg{P_1}\preqaim\trg{P_2} \triangleq$ $\forall \trg{\mb{P}}, \trg{h_1}, \exists \trg{h_2}.\trg{\mb{P}}[\trg{P_1}],\trg{h_1}\divraim\Rightarrow\trg{\mb{P}}[\trg{P_2}],\trg{h_2}\divraim$, where $\trg{\mb{P}}[\trg{P}],\trg{h}$ indicates the initial state of $\trg{\mb{P}}\trg{+}\trg{P}$ with oracle \trg{h}.
\end{definition}
\begin{definition}[Contextual equivalence for \aim]\label{def:ceqaim}
$\trg{P_1}\ceqaim\trg{P_2} \triangleq \trg{P_1}\preqaim\trg{P_2} $ and $ \trg{P_2}\preqaim\trg{P_1}$.
\end{definition}

%% file: bodies/threat.tex
\section{Security Pitfalls}\label{sec:probs}
This section describes the threat model considered in this paper (\Cref{sec:threatmodel}) and the security pitfalls that arise when linking is explicitly considered (\Cref{sec:generalproblems}).

\subsection{The Threat Model}\label{sec:threatmodel}
The goal of this section is to familiarise the reader with the security aspects of secure compilation that are relevant for this paper.
We do not provide an in-depth analysis of these aspects as they are not the focus of this paper, though we believe these insights can be helpful for many readers.

Secure compilation papers generally consider a threat model for an attacker with target-level code injection capabilities that operates, for example, by exploiting a bug in the system and injecting or loading arbitrary target programs.
His goal is to violate the security properties of compiled programs found in the system.
Specifically, in this paper we consider an attacker that can load arbitrary malicious code.

In this paper, the system is the von Neumann machine formalised by the \aim language, where multiple protected modules are found.
The compiled programs that the attacker wants to violate span \emph{some} of the protected modules, as explained in \Cref{sec:intro}; their security properties must not be violated.
The attacker is modelled as code (and data) that spans unprotected memory as well as \emph{other} protected modules; thus the attacker is modelled as an \aim context.
This attacker can interact with any of the compiled programs whose security is of interest in any way that the assembly language (and the \PMA access control policy) allow him.
Thus, the attacker must respect the processor-enforced \PMA access control policy and, most importantly, he cannot tamper with it.

We are interested in program security properties that can be expressed by means of contextual equivalence; some examples include confidentiality and integrity.
A value is confidential if it cannot be discerned by other components than the one declaring it.
In other words, a value $\src{v}$ in a component $\src{\mc{C}}$ is confidential if $\src{\mc{C}}$ is contextually-equivalent to $\src{\mc{C}'}$ which is $\src{\mc{C}}$ with a different value for $\src{v}$.
Integrity of a value means that it cannot be modified by other components than the one declaring it.
In other words, a value $\src{v}$ in a component $\src{\mc{C}}$ has integrity if $\src{\mc{C}}$ is contextually-equivalent to $\src{\mc{C}'}$ which is $\src{\mc{C}}$ where every interaction with other component is followed by a check that the value of $\src{v}$ is the same as before the interaction.

The goal of the paper is to provide a compiler \compaim{\cdot} that produces secure \aim modules.
Thus, \compaim{\cdot} takes \jem components that possibly has some security properties and outputs \aim modules that has the same security properties.
If this holds, then the attacker's goal is nullified.

As compiler full-abstraction is preservation and reflection of contextual equivalence, proving \compaim{\cdot} to be fully-abstract implies that \compaim{\cdot} is secure.

%% file: bodies/probs.tex

\subsection{Security Problems Related to Linking}\label{sec:generalproblems}
\Cref{prob:guessing,prob:callstack,prob:types,prob:exist} are full-abstraction violations that arise in the presence of linking.
Any full-abstraction violation can be lead back to confidentiality or integrity violations.

\begin{problem}[Object id guessing]\label{prob:guessing}
Consider an object allocated at address \trg{a} in module \trg{M_1} and that is only shared between two modules \trg{M_1} and \trg{M_2}.
As object ids are just addresses in memory, nothing prevents another module \trg{M_3} to guess address \trg{a} and call methods on it.
Note however, that this violates the intended confidentiality of \trg{a}, that is supposed to be visible only to \trg{M_1} and \trg{M_2}.

A na\"ive solution to this problem is tracking which module will receive a certain object id in order to detect and stop guesses from modules that have not received that object id.
However, this does not scale to the case where \trg{M_2} forwards \trg{a} to \trg{M_3}.
Since the forwarding is done outside of \trg{M_1}, \trg{M_1} has no way of adding \trg{M_3} to the modules that are allowed to access \trg{a}.
As this may not be known statically, a different solution is needed.

To address this concern, target-level object ids must be unguessable.
\end{problem}

\begin{problem}[Call stack shortcutting]\label{prob:callstack} 
As all assembly languages, \aim calls and returns are jumps to addresses in memory.
Consider the following sequence of cross-module function calls: 
$\trg{M_2} \to \trg{M} \to \trg{M_1} \to \trg{M'}$, where primes are used to denote subsequent calls to the same module.

$\trg{M'}$ can return to \trg{M_2}, as it learnt the address to return there when \trg{M_2} called it.
However, that instruction bypasses the rest of the call stack and particularly $\trg{M_1}$ in a way that is not possible in the source language, where control flow follows a well-bracketed sequence of calls/returns.

To address this concern, securely compiled code must ensure a well-bracketed traversal of the stack.
\end{problem}

\begin{problem}[Types of objects in other modules]\label{prob:types}
Consider an object of type \src{t'} allocated at address \trg{a} in module \trg{M_1} and a method \src{m} compiled in a different module \trg{M_2} that takes a parameter of type \src{t}.
\src{t} and \src{t'} are different.
A third module \trg{M_3} could pass \trg{a} as an argument to the compilation of \src{m}, violating the \jem assumption that only well-typed programs are executed.
A dynamic check on the type of \trg{a} should be made, but the code of \src{m} resides in \trg{M_2} and it has no way to access \trg{a}, which resides in module \trg{M_1}.

To address this concern, dynamic typechecks must be made on all objects that are passed and received via methods and the encoding of the class of an object must be accessible outside of the module containing the object.
\end{problem}

\begin{problem}[Existence of objects in other modules]\label{prob:exist}
Analogously, to \Cref{prob:types}, in place of an object of the wrong type, \trg{M_3} code can call the compilation of \src{m} by passing a non-existent object id.
A non-existent object id is a non-\src{\lst{null}} object id that does not point to an object, so that calling methods on it will fail.

To address this concern, when securely compiled code receives an object from another module, it must assess whether that object exists or not.
\end{problem}

%

%

%% file: bodies/comp.tex
\section{A Secure Compiler from \jem to \aim}\label{sec:seccomp}
\compaim{\cdot} is a modular, two-step compiler that inputs a \jem class \src{C} and returns two \aim modules.
Formally:
\begin{align*}
	\compaim{\src{C}} =&\ \prot{\comp{\src{C}}}\trg{+} \sys.
\end{align*}
\compaim{\cdot} consists of a call to the \comp{\cdot} function (\Cref{sec:comp}), followed by a call to the \prot{\cdot} one (\Cref{sec:prot}). 
The result is joined with an additional module \sys (\Cref{sec:sys}).

Intuitively, \comp{\cdot} is a compiler responsible for correctly translating classes, methods and objects to assembly.
\prot{\cdot} is a wrapper around the code generated by \comp{\cdot} that prevents direct access to that code and ensures that any interaction with that code is regulated to behave as valid \jem code.
This wrapping ensures that the output of \compaim{\cdot} is secure.
\sys is a central, trusted module that operates as a monitor, regulating the structure of function calls and keeping information of allocated objects.

\compaim{\cdot} is modular, i.e., it can be applied to \jem components
\src{\mc{C}}. In this case \compaim{\cdot} recursively calls itself on
all classes composing the \jem component. The secure linker
\mylink{\cdot} combines all compiled components to a low-level
component (\Cref{sec:mylink}). Formally:
\begin{align*}
	\compaim{\src{C_1};\cdots;\src{C_n}} =&\ \mylink{\compaim{\src{C_1}}, \cdots, \compaim{\src{C_n}}}
\end{align*}

The overall structure of the compiler proposed in this paper is similar to the one for a single-module version of \PMA~\cite{scoo-j}, but two important extensions are needed to handle the issues (\Cref{pt:probone,pt:probtwo}) mentioned in \Cref{sec:intro}:
\begin{enumerate}
\item we introduce a small centralized trusted module \sys that keeps information about the global control flow of the program, and about types of objects shared between multiple modules. 
A key challenge is to design this module such that the information it exposes does not break full abstraction (\Cref{sec:sys}).

\item we introduce unforgeable object references. 
Since we do not want to assume hardware support for this (as is done for instance in capability machines), we solve this by assuming the existence of a secure random number generator that can create unguessable random numbers (nonces) that are sufficiently large to make brute-force attacks infeasible (as achievable via the \mtt{rdrand} instruction of intel processors~\cite{rdrand}). 
We model these nonces symbolically (i.e., \trg{\pi}) and use them to represent references to objects outside of the module where the object is defined (\Cref{sec:dynmem}).
\end{enumerate}

\subsection{The First Step: \comp{\cdot} }\label{sec:comp}
\comp{\cdot} is a compiler that translates a single \jem class to \aim code, data and a symbol table.
Instead of giving a specific instance of \comp{\cdot}, we define what assumptions (\Cref{ass:comp-out,ass:comp-cc,ass:comp-corr}) such a compiler must uphold for the full-abstraction result to hold.
The full-abstraction result is achieved parametrically for any \comp{\cdot} that upholds those assumptions.

\begin{assumption}[Output of \comp{\cdot}]\label{ass:comp-out}
The compilation of a \jem class returns a memory \trg{m_c} providing code implementing methods \src{\OB{M}}, a memory \trg{m_d} providing data implementing objects \src{\OB{O}} and a symbol table \trg{t}.
\trg{t} contains exported and required methods and objects bindings sorted lexicographically by class and then method or object name~\cite{scoo-j}.

The produced code and data implement a class-local stack that is used to perform calls and returns between methods.

\comp{\cdot} is multi-entry i.e., it exports an address for each method that can be called.
\comp{\cdot} is single-exit i.e., all jumps outside the code are performed by a common piece of code located in a known part of the produced code: \trg{m_{exit}}.
\comp{\cdot} expects the address where \lstb{instanceof} is implemented to be supplied later, so a call to \lstb{instanceof} is compiled to symbols $(\trg{\iota_{inst}},\trg{\sigma_{inst}})$ which capture both parts of the address where \lstb{instanceof} is.
This is captured by the required bindings with a binding of the form $\src{instanceof}:\src{\lstb{Obj}}(\src{\lstb{Obj}},\src{\lstb{Obj}})\src{\lstb{Bool}}\mapsto (\trg{\iota_{inst}},\trg{\sigma_{inst}})$ which must be in \trg{t}.

\jem object ids \src{o} are compiled to numbers \trg{n} that point to a memory region with type information and fields.

Formally: $\comp{\lstb{import}\ \src{\OB{I}};\src{\OB{X}};\ \lstb{class}\ \src{c} \{\src{K}\ \src{\OB{F_t}}\ \src{\OB{M}}\}; \src{\OB{O}}} = (\trg{m_c}+\trg{m_{exit}}+\trg{m_{inst}}; \trg{m_d}; \trg{t})$.
\end{assumption}
We overload the \comp{\cdot} notation and use it to indicate values as are compiled by \comp{\cdot}.
Indicate the compilation of \src{\lst{unit}} as \comp{\src{\lst{unit}}} and the compilation of \src{\lst{true}} as \comp{\src{\lst{true}}}.
Analogously, indicate the encoding of \jem types \src{t} in \aim as \comp{\src{t}}.
For example, indicate the encoding of \src{\lstb{Unit}} as \comp{\src{\lstb{Unit}}} and the encoding of a class type \src{c} as a natural number $\comp{\src{c}}=\trg{n}$.

\begin{assumption}[Calling convention of \comp{\cdot}]\label{ass:comp-cc}
Registers are used according to the following calling convention.
\asm{r_0} is used by caller-callee authentication to store the caller module id.
\asm{r_1} to \asm{r_4} are used as general working registers, so they do not contain relevant information.
\asm{r_5} identifies the return address in a \asm{jmp} that models a method call; if \asm{r_5} is \trg{1}, then that \asm{jmp} is interpreted as a return. 
\asm{r_6} identifies the current object (\lst{this}) in a target-level method call or the returned value in a target-level return.
\asm{r_7} onwards are used to communicate method parameters.
\end{assumption}

\begin{assumption}[Restrictions of \comp{\cdot}]\label{ass:comp-res}
Compiled components do not use the $\lstb{exit}\ \src{E}$ expression.
Addiitonally, they do not read nor write to unprotected memory.
\end{assumption}

\begin{assumption}[Correctness of \comp{\cdot}]\label{ass:comp-corr}
\comp{\cdot} must be correct and adequate, i.e., it translates any \jem expressions and values into \aim code that behaves in the same way~\cite{KRLMLa,bistcc,leroy}.
\end{assumption}

\subsection{The System Module: \sys}\label{sec:sys}
The system module \sys is a central, trusted module that provides functionality used by all securely-compiled components.
\sys contains: a global store \mc{G} (\Cref{sec:globstore}), a global call stack \mc{S} (\Cref{sec:sys-callstack}) and functionality to interact with \mc{G} and \mc{S} (\Cref{sec:sysfuncs}); its definition is provided last (\Cref{sec:sysdef}).

In the following, when code \emph{aborts} we mean that all registers and flags are reset, then \asm{halt} is executed.

\subsubsection{The Global Store \mc{G}} \label{sec:globstore} 
\mc{G} tracks globally-known object ids (i.e., ids that are not just local to a module), their type and the module id where the object resides.
Formally $\mc{G}=\OB{\trg{w}\mapsto\trg{w'},\trg{id}}$, where \trg{w} is an object id and \trg{w'} is a class encoding.
Retrieval from \mc{G} is denoted with $\mc{G}(\trg{w})$ and addition to it is denoted with $\mc{G}+(\trg{w}\mapsto\trg{w'},\trg{id})$.

\sys provides two entry points where the following procedures are implemented: \fun{testObj}{\trg{w},\trg{w'}} and \fun{registerObj}{\trg{w},\trg{w'}}.
The former tells if an object \trg{w} exists and implements a certain class \trg{w'}.
The latter adds a new binding to \mc{G}, the new binding is added for the module whose id is in \trg{r}(\asm{r_0}), i.e., for the module that calls \fun{registerObj}{\cdot}.
Parameters \trg{w} and \trg{w'} for these procedures are expected respectively in registers \asm{r_7} and \asm{r_8}.
\mytabnocap{
\begin{tabular}{l}
	\fun{testObj}{\trg{w},\trg{w'}} 
	\\
	\hline
	if $\trg{w}\notin\dom{\mc{G}}$ then abort
	\\
	if $\mc{G}(\trg{w}) \equiv (\trg{w'},\trg{\_})$ then 
	\\
	\quad return \trg{0}
	\\
	else return \trg{1}
\end{tabular}
\begin{tabular}{l}
	\fun{registerObj}{\trg{w},\trg{w'}}
	\\
	\hline
	if $\trg{w}\in\dom{\mc{G}}$ then abort
	\\
	$\mc{G}+\trg{w}\mapsto\trg{w'},\trg{r}(\asm{r_0})$
	\\
	return \trg{0}
	\\
	\phantom{a}
\end{tabular}
}

\subsubsection{The Global Call Stack \mc{S}}\label{sec:sys-callstack}
\mc{S} tracks all \aim-level function calls.
Formally, $\mc{S}=\OB{(\trg{a})}$.

\sys provides the following procedures implemented at entry points: \fun{forwardCall}{} and \fun{forwardReturn}{}.
Both procedures use the following helper functions: \fun{resetFlags}{\cdot} sets flags to \trg{0}, \fun{resetRegisters}{\cdot} inputs the registers that need to be reset to \trg{0} and \fun{resetRegistersExcept}{\cdot}inputs the registers that need \emph{not} be reset to \trg{0}.

\fun{forwardCall}{} reads an address from \asm{r_3} and \asm{r_4} and forwards the call there.
Before forwarding, it stores in \mc{S} the addess where to return as passed via registers \asm{r_0} and \asm{r_5}.
In order to ensure a correct return, the procedure stores the address of the entry point for \fun{forwardReturn}{} (i.e., $3*\trg{\mc{N}_w}$) in \asm{r_5}.
Since the caller id is placed in \asm{r_0} by the caller-callee authentication mechanism of \PMA, it cannot be tampered with.
This procedure aborts if the module jumping here is the one that jumped here last (if any) or if the address where the call is forwarded is inside \sys.
\mytabnocap{
\begin{tabular}{l}
	\fun{forwardCall}{} 
	\\
	\hline
	If $\mc{S}\neq\emptyset$
	\\
	\quad Let $\mc{S} = (\trg{id},\trg{\_});\mc{S'}$
	\\
	\quad if $\trg{r}(\asm{r_0})==\trg{id}$ then abort
	\\
	if $\trg{r}(\asm{r_3})==\trg{1}$ then abort
	\\
	Push $(\trg{r}(\asm{r_0}),\trg{r}(\asm{r_5}))$ on \mc{S}
	\\
	Set \asm{r_5} to $3*\trg{\mc{N}_w}$
	\\
	\fun{resetFlags}{}
	\\
	\fun{resetRegisters}{\asm{r_0},\asm{r_1},\asm{r_2}}
	\\
	\asm{jmp\ r_3\ r_4}
\end{tabular}
}

\fun{forwardReturn}{} pops the head of the stack and returns there, setting \asm{r_5} to \trg{1} as expected in the case for returns.
If the call stack is empty, it aborts to prevent returning when no code was called.
If the id of the caller module stored on \mc{S} is different from the id of the module jumping to this procedure, it aborts, as it detects a non-well-bracketed execution flow.
\mytabnocap{
\begin{tabular}{l}
	\fun{forwardReturn}{} 
	\\
	\hline
	If $\mc{S} = \emptyset$ then abort
	\\
	Pop (\trg{id},\trg{n}) from \mc{S} into \asm{r_2} and \asm{r_1}
	\\
	if $\trg{id}\neq\trg{r}(\asm{r_0})$ then abort
	\\
	Set \asm{r_5} to 1;
	\\
	\fun{resetFlags}{}
	\\
	\fun{resetRegistersExcept}{\asm{r_1},\asm{r_2},\asm{r_5},\asm{r_6}} 
	\\
	\asm{jmp\ r_1\ r_2} \qquad\emph{// i.e., jump to address} (\trg{id},\trg{n})
\end{tabular}
}

The calling convention is updated as follows: calls to different modules must go via \sys, with registers \asm{r_3} and \asm{r_4} set to the method that needed to be called.
Both these procedures do not use registers \asm{r_6} onwards, so they do not alter the calling convention regarding parameters and returned values.

\subsubsection{Functionality of \sys}\label{sec:sysfuncs}
The addition of \sys may seem to violate full abstraction as it provides target-level functionality that are not available in \jem~\cite{abadiFa}.
Let us now see why this is not the case.

Procedure \fun{testObj}{\cdot} is analogous to the \lstb{instanceof} expression in \jem.
This procedure aborts when the object id parameter is not registered in \mc{G} i.e., an ill-formed execution of \lstb{instanceof}.
Aborting in this case lets \sys prevent object id guessing (\Cref{prob:guessing}).

Procedure \fun{registerObj}{\cdot} is analogous to object creation.
This procedure aborts when the object id is already registered; aborting in this case ensures that object ids are globally unique. 

Procedure \fun{forwardCall}{\cdot} is analogous to performing a function call.
This procedure aborts when the same module performs two calls in a row without there being a return in between or if it performs a call to \sys as those behaviours are not available in \jem.

Procedure \fun{forwardReturn}{\cdot} is analogous to returning from a function call.
This procedure aborts when a return is not made by the module who was called last.
This enforces well-bracketed control flow, preventing \Cref{prob:callstack} as well as external code returning when no method was called. 

All the procedures implemented in \sys add functionality in \aim that is already available in \jem, so adding \sys when linking \aim components does not violate full-abstraction.

\subsubsection{\sys Definition}\label{sec:sysdef}
Assume the functions provided by \sys span memory \trg{m}, whose size is \trg{n} addresses.
Let the data needed by these functions plus the initialisation of \mc{G} and \mc{S} span memory \trg{m'}.
\sys is always compiled to a module with \trg{id} \trg{1} (since \trg{0} is the id of unprotected code) and it is defined as follows:
	$$(\trg{m+m'};(\trg{1,n,4});\trg{EM_{to}},\trg{EM_{ro}},\trg{EM_{rc}},\trg{EM_{tr}};\trg{\emptyset};\trg{\emptyset};\trg{\emptyset})$$
\sys exports a method binding for each procedure it defines; it has no exported object bindings and no required bindings.
\begin{itemize}\small
\item $\trg{EM_{to}} = \src{testObj}:\src{\lstb{Obj}}(\src{\lstb{Obj}},\src{\lstb{Obj}})\src{\lstb{Bool}}\mapsto (\trg{1},\trg{0})$ 
\item $\trg{EM_{ro}} = \src{registerObj}:\src{\lstb{Obj}}(\src{\lstb{Obj}},\src{\lstb{Obj}})\src{\lstb{Unit}}\mapsto(\trg{1},\trg{\mc{N}_w})$
\item $\trg{EM_{rc}} = \src{forwardCall}:\src{\lstb{Obj}}(\cdot)\src{\lstb{Unit}}\mapsto(\trg{1},2*\trg{\mc{N}_w})$
\item $\trg{EM_{tr}} = \src{forwardReturn}:\src{\lstb{Obj}}(\cdot)\src{\lstb{Unit}}\mapsto(\trg{1},3*\trg{\mc{N}_w})$
\end{itemize}

\compaim{\cdot} knows the offset of all procedures defined in \sys, so it can perform calls to them.
For example a call to \fun{registerObj}{\cdot} is a \asm{jmp\ r_d\ r_i} where $\trg{r}(\asm{r_d})\mapsto\trg{\mc{N}_w}$ and $\trg{r}(\asm{r_i})\mapsto\trg{1}$.

\subsection{The Second Step: \prot{\cdot}}\label{sec:prot}
The \prot{\cdot} function is a wrapper that takes the memory generated by \comp{\cdot} and adds checks to it, making it secure. 
This section first presents the required helper functions to compile \jem features: dynamic memory allocation (\Cref{sec:dynmem}), function calls (\Cref{sec:funcalls}) and outcalls (\Cref{sec:outcalls}).
The definition of \prot{\cdot} is provided last (\Cref{sec:protdef}). 

\subsubsection{Dynamic memory allocation}\label{sec:dynmem}
Dynamic memory allocation is the creation of objects at runtime via the \lstb{new} expression in \jem.
The representation of an object id within a compiled component does not change.
To ensure security of object ids when communicated between \aim modules, object ids are no longer just an address in memory, they are symbolic nonces.
The latter are called \emph{cross-module object ids} and they are denoted with \compaim{\src{o}}.
Formally: $\compaim{\src{o}} = \trg{\pi}$. 
Symbolic nonces cannot be forged nor guessed, so the format of cross-module object ids addresses \Cref{prob:guessing}.
Cross-module object ids being a symbolic nonce instead of an address does not disrupt the functionality of compiled \jem code, which functions in the same way as before.

To relate cross-module and internal object ids within a module, masking tables are used~\cite{scoo-j}.
A masking table, indicated with \mc{T}, is a bidirectional hash map between internal object id representations and cross-module ones; each module has its own masking table.
Formally $\mc{T}::=\OB{\trg{w}\mapsto\trg{\pi}}$.
Before an (internal) object id \trg{w} is first passed to external code, it is placed in the table.
This process is called \emph{masking}, thus the passed id \trg{\pi} is called the \emph{mask}.
Denote the retrieval of a mask \trg{\pi} with $\mc{T}(\trg{\pi})$ and the retrieval of an internal object id \trg{w} with $\mc{T}(\trg{w})$.
Any retrieval causes abortion if the element to be retrieved is not in \mc{T}.
Adding an internal object id $\trg{w}$ to a table \mc{T} is denoted with $\mc{T}+\trg{w}$; the binding $\trg{w}\mapsto\trg{\pi}$, where \trg{\pi} is fresh, is added to \mc{T}.

Next are the functions used to manage masking tables.

A compiled component must store the encoding of the class it contains for helper functions to rely on.
Function \fun{classOf}{\trg{n}} returns the encoding of the class implemented by the object compiled at address \trg{n} in the current module.
Function \fun{isInternal}{\trg{w}} takes a class type encoding and returns true if the current module implements that class.

Function \fun{updateMaskingTable}{\trg{w},\trg{n}} inputs a pair of an internal object id and a class type encoding.
This function is invoked before releasing an object to external code to add the id of freshly-allocated object to the masking table.
This function retrieves the class of the object and then checks if that object id is already in the masking table; if not, it adds the id to the table and registers its cross-module id globally.
Retrieving the class object may be necessary, as some arguments may have formal type \src{\lstb{Obj}}.
A crucial part of this function is the call to \fun{registerObj}{\mc{T}(\trg{w}),\trg{n'}}, which makes the information that object \trg{w}, with cross-module id $\mc{T}(\trg{w})$, implements class \trg{n'} globally available via \sys.
\mytabnocap{
\begin{tabular}{c l}
	&
	\fun{updateMaskingTable}{\trg{w},\trg{n}}
	\\
	\hline
	&
	if $\trg{n}$ \emph{is a class type encoding} then
	\\
	&
	\quad let \trg{n'} = \fun{classOf}{\trg{w}}
	\\
	&
	\quad if \fun{isInternal}{\trg{n'}} and if $\trg{w}\notin\dom{\mc{T}}$ then
	\\
	&
	\quad\quad \mc{T}+\trg{w}
	;
	\fun{registerObj}{\mc{T}(\trg{w}),\trg{n'}}
\end{tabular}
}

Function \fun{loadObjects}{\cdot} loads the internal object ids of the masking indexes that are passed as input into the related register.
\mytabnocap{
\begin{tabular}{c l}
	&
	\fun{loadObjects}{\trg{\pi_1},\cdots,\trg{\pi_n}}
	\\
	\hline
	&
	$\forall i \in 1..k$
	\quad 
	$\trg{r}(\asm{r_i}) \mapsto \mc{T}(\trg{\pi_i})$
\end{tabular}
}

Function \fun{maskingTable}{\trg{EO}} creates the masking table for all exported objects bindings listed in \trg{EO}.
\mytabnocap{
\begin{tabular}{c l}
	&
	\fun{maskingTable}{\src{o_1}:\src{c_1}\mapsto\trg{n_1},\cdots,\src{o_k}:\src{c_k}\mapsto\trg{n_k}} = \mc{T}
	\\
	\hline
	&
	$\forall i \in 1 .. k$ \quad let \trg{w_i} = \asm{new\ r_0} 
	\quad
	\mc{T} = $\trg{n_1}\mapsto\trg{w_1},\cdots,\trg{n_k}\mapsto\trg{w_k}$
\end{tabular}
}

\subsubsection{Function calls}\label{sec:funcalls}
Function calls are calls from outside to within a compiled module.
The \PMA access control policy ensures that these calls can only be calls to entry points.
Therefore, an entry point is created for all methods~\cite{scoo-j}.
Function \fun{methodEP}{\src{{M_t}},\trg{n}} creates the code to be placed at a method entry point for a method with signature \src{M_t} whose implementation is located at address $(\trg{id},\trg{n})$ in the current module with id \trg{id}.
Following is the pseudo-code of method entry points; notation $\trg{n}\mapsto\mtt{code}$ means that \trg{n} is the address where \mtt{code} is located.
\mytabnocap{
\begin{tabular}{c l}
	&
	\fun{methodEP}{\src{t}(\src{t_1},\cdots,\src{t_k})\to\src{t'},\trg{n}} =
	\\
	\hline
	&
	If $\trg{r}(\asm{r_0}) \neq \trg{1}$ or $\trg{r}(\asm{r_5})\neq 3*\trg{\mc{N}_w}$ then abort
	\\
	&
	\fun{loadObjects}{\trg{r}(\asm{r_6}),\trg{r}(\asm{r_7}),\cdots,\trg{r}(\asm{r_{6+k}})}
	\\
	&
	\fun{dynamicTypechecks}{\trg{r}(\asm{r_6}),\comp{\src{t}}}
	\\
	&
	for i = 1 to k
	\\
	&
	\quad \fun{dynamicTypechecks}{\trg{r}(\asm{r_{6+i}}),\comp{\src{t_i}}}
	\\
	$\trg{n'}\mapsto$
	&
	Set \asm{r_5} to \trg{n'+1} and jump to module-local address \trg{n}
	\\
	&
	\fun{updateMaskingTable}{\trg{r}(\asm{r_6}),\comp{\src{t'}}}
	\\
	&
	\fun{maskObjectId}{ 6 }
	\\
	&
	Set \asm{r_5} to \trg{1} and \asm{jmp} to the address (\trg{1},3*\trg{\mc{N}_w})
\end{tabular}
}

Upon jumping to an entry point, a check is made that the jump comes from the \sys module and that the return address is the \fun{forwardReturn}{\cdot} address (i.e., $3*\trg{\mc{N}_w}$ in a module whose id is \trg{1}).
If this is not the case, execution is aborted as some code is trying to bypass \sys.
All data needed for this check is known statically 
Then, masked object ids are loaded via function \fun{loadObjects}{\cdot}.
Only parameters of object type are loaded; if a parameter could not be loaded, the execution is aborted.
Once the objects are loaded, dynamic typechecks are made via function \fun{dynamicTypechecks}{\cdot} (explained below).
These checks are located inside the module where class \src{t} is compiled, so checking that the current object (\asm{r_6}) is of type \comp{\src{t}} is equivalent to checking that the current object implements the current method.
Then the code jumps to the method body located at address \trg{n} setting \asm{r_5} to the address where that code must return.
There, current arguments are placed on the module-local stack alongside other information required by function activation records.
When the method body returns (to address \trg{n'+1}), the masking table is (possibly) updated with the value to be returned (in \asm{r_6}) via function \fun{updateMaskingTable}{\cdot} and internal object ids in registers are masked via function \fun{maskObjectId}{\cdot}.

Let us now provide details about the auxiliary functions.

\fun{dynamicTypechecks}{\cdot} ensures that each parameter inhabits its type.
\src{\lstb{Unit}}-typed values are checked to be \comp{\src{\lst{unit}}} and \src{\lstb{Bool}}-typed ones are checked to be either \comp{\src{\lst{true}}} or \comp{\src{\lst{false}}}~\cite{scoo-j,fstar2js}.
Objects are dynamically typechecked by means of the \fun{testObj}{\cdot} function (as discussed in \Cref{sec:sys}).
If any check fails, the execution is aborted. 
By checking the existence and types of all parameters, securely compiled code is resilient to \Cref{prob:exist,prob:types}.
\mytabnocap{
\begin{tabular}{l}
	\fun{dynamicTypechecks}{\trg{w},\trg{n}}
	\\
	\hline
	if $\trg{n}\equiv\comp{\src{\lstb{Unit}}}$ then 
	\ if $\trg{w} \neq \comp{\src{\lst{unit}}}$ then abort
	\\
	if $\trg{n}\equiv\comp{\src{\lstb{Bool}}}$ then 
	\ if $\trg{w} \neq \comp{\src{\lst{true}}}$ and $\trg{w} \neq \comp{\src{\lst{false}}}$ then abort
	\\
	if \trg{n} \emph{is a class type} and $\trg{w} \neq \comp{\src{\lst{null}}}$ then 
	\\
	\quad if \fun{testObj}{\trg{w},\trg{n}} == \trg{1} then abort
\end{tabular}
}


Function \fun{maskObjectId}{\cdot} inputs the index of the register to mask and loads a cross-module object id there.
\mytabnocap{
\begin{tabular}{c l}
	&
	\fun{maskObjectId}{n} 
	\\
	\hline
	&
	$\trg{r}(\asm{r_n})\mapsto\mc{T}(\trg{r}(\asm{r_n}))$
\end{tabular}
}

\subsubsection{Outcalls}\label{sec:outcalls}
Outcalls are the dual of function calls, i.e., they are calls from within to outside a module.
To allow returning from outcalls, a specific entry point must be created: the \emph{return entry point}~\cite{scoo-j}.

Function \fun{preamble}{\cdot} returns what all code must execute before making an outcall,
function \fun{returnEP}{\cdot} returns the code \trg{m} to be placed at the return entry point.

\mytabnocap{
\begin{tabular}{l}
	\fun{preamble}{} = 
	\\
	\hline
	Let \trg{a} be the address where the external method is; $\trg{a}\equiv (\trg{id},\trg{w})$
	\\
	$\fun{signatureOf}{\trg{a}} = \comp{\src{t}}(\comp{\src{t_1}},\cdots,\comp{\src{t_n}})\to\comp{\src{t'}}$
	\\
	\fun{storeData}{\trg{r}(\asm{r_6}), \comp{\src{t'}}, \trg{r}(\asm{r_5})}
	\\
	for i = 1 to n
	\\
	\quad \fun{updateMaskingTable}{\trg{r}(\asm{r_{6+i}}),\comp{\src{t_i}}}
	\\
	\quad \fun{maskObjectId}{6+i}
	\\
	\fun{resetRegistersExcept}{\asm{r_0},\cdots,\asm{r_{6+n}}}
	\\
	Set \trg{r}(\asm{r_3}) to \trg{id} and \trg{r}(\asm{r_4}) to \trg{w}
	\\
	Call \fun{forwardCall}{\cdot}
\end{tabular}
}
The \emph{preamble} code is executed before jumping to an external method located at address $\trg{a}$, assume this address is communicated via registers \asm{r_0} and \asm{r_1}.
Compiled code will execute the \fun{preamble}{\cdot} before jumping outside to ensure the right checks are made.
\fun{signatureOf}{\cdot} is used to determine the signature of the compiled method related to \trg{a}.
All method signatures are known statically so their signature encoding can be stored in a table mapping addresses to signature encodings; the table is placed in the data section of the module.
This data is communicated to \prot{\cdot} via the required method bindings \trg{RM} returned by \comp{\cdot}.
\fun{preamble}{} then calls to \fun{storeData}{\cdot}, which stores the current object \asm{r_6}, the expected return type \comp{\src{t'}}, and where to resume the execution after the outcall, an address that is assumed to be passed in \asm{r_5}.
For any parameter, the masking table is updated with any possible newly created object via function \fun{updateMaskingTable}{\cdot} and their ids are masked with function \fun{maskObjectId}{\cdot}.
Then, registers that are not used to convey parameters nor the address where to jump are reset to 0 by function \fun{resetRegistersExcept}{\cdot} in order not to leak information (\sys will erase unused registers with index less than 6).
This ensures that \asm{r_5} contains \trg{0}, so \sys will return to the return entry point located at address \trg{0} when forwarding the return after this call.
Finally the code sets \asm{r_3} and \asm{r_4} as expected by \sys, then it jumps to the proxy function for method calls: \fun{forwardCall}{}.

\mytabnocap{
\begin{tabular}{c l}
	&
	\fun{returnEP}{} =
	\\
	\hline
	$\trg{0} \mapsto$
	&
	If $\trg{r}(\asm{r_0}) \neq \trg{1}$ then abort
	\\
	&
	\fun{loadData}{} \qquad\qquad \emph{// access} \trg{w_o}, \comp{\src{t'}}, \trg{n'}
	\\
	&
	\fun{loadObjects}{\trg{r}(\asm{r_6})}
	\\
	&
	\fun{dynamicTypechecks}{\trg{r}(\asm{r_6}),\comp{\src{t'}}}
	\\
	&
	Resume execution from address \trg{n'} with current object \trg{w_o}
\end{tabular}
}
When a return is made, if the module returning is not \sys (i.e., a module with id \trg{1}), then the code aborts, as some code is trying to bypass \sys.
Then, the current object \trg{w_o}, the expected return type \comp{\src{t}} and the address where to resume execution \trg{n'} are loaded from the module-local stack via function \fun{loadData}{\cdot} into registers \asm{r_7} onwards (since these registers are unused).
Finally, the returned value \asm{r_6} is checked to be of the expected type.
Execution aborts if any check fails, otherwise it resumes within the module with the loaded current object.

\subsubsection{The \prot{\cdot} function}\label{sec:protdef}
\prot{\cdot} is formalised in \Cref{tr:prot-def}.

To ensure that the compiler is correct, \prot{\cdot} needs to provide an implementation of \lstb{instanceof}, addressing \Cref{prob:types2} below.
\begin{problem}[Using \lstb{instanceof}]\label{prob:types2}
Consider a module \trg{M} that contains object \src{o} implementing class \src{c}.
Another module \trg{M_1} executes the following code: $\lstb{instanceof}(\src{o}:\src{c})$.
In order to tell if the test succeeds or not, the code of \trg{M_1} must know the class of \src{o}.
However, with the strong encapsulation provided by \PMA, that information resides in the memory of \trg{M}, which is not accessible by \trg{M_1}.
\end{problem}
To correctly implement cross-module \lstb{instanceof}, the class of an object needs to be publicly known, which is a functionality provided by \sys.
When calling \lstb{instanceof} on an object whose type is implemented in another module, it suffices to call \fun{testObj}{\cdot} to know if the test succeeds or not.
To ensure this happens, the required method binding for \src{instanceof} is replaced with a method binding for the same symbols to \fun{testObj}{\cdot}.
Linking to \sys will ensure that those symbols are replaced with the address of \fun{testObj}{\cdot}.

\mytabnocap{
\typerule{\prot{\cdot} definition}{
	\comp{\src{C}} =  (\trg{m_c}+\trg{m_{exit}}; \trg{m_d}; \trg{EM};\trg{EO};\trg{RM}+\trg{RM_i};\trg{RO})
	\\
	\trg{EM}=\src{m_1}:\src{{M_t}_1}\mapsto(\trg{id},\trg{n_1}),\cdots,\src{m_k}:\src{{M_t}_k}\mapsto(\trg{id},\trg{n_k})
	\\
	\trg{EM'} =\ \src{m_1}:\src{{M_t}_1}\mapsto\trg{a_1},\cdots,\src{m_k}:\src{{M_t}_k}\mapsto\trg{a_k}
	\\
	\trg{EO} = \src{o_1}:\src{c_1}\mapsto\trg{n_1''},\cdots,\src{o_j}:\src{c_j}\mapsto\trg{n_j''}
	\\
	\trg{EO'} = \src{o_1}:\src{c_1}\mapsto\mc{T}(\trg{n_1''}),\cdots,\src{o_j}:\src{c_j}\mapsto\mc{T}(\trg{n_j''})
	\\
	\trg{RM_i}= \src{instanceof}:\src{\lstb{Obj}}(\src{\lstb{Obj}},\src{\lstb{Obj}})\src{\lstb{Bool}}\mapsto (\trg{\iota_{inst}},\trg{\sigma_{inst}})
	\\
	\trg{RM_i'}= \src{testObject}:\src{\lstb{Obj}}(\src{\lstb{Obj}},\src{\lstb{Obj}})\src{\lstb{Bool}}\mapsto (\trg{\iota_{inst}},\trg{\sigma_{inst}})
	\\
	\trg{s}=(\trg{id},\trg{n'},\trg{k+1})
	&
	\fun{maskingTable}{\trg{EO}} = \mc{T}
	\\
	\trg{m_m} = \fun{implementationOf}{\mc{T}}
	&
	\trg{m_r} =\ \fun{returnEP}{}
	\\
	\trg{m_{ce}} =\ \fun{extraCode}{}
	&
	\trg{m_{de}} =\ \fun{extraData}{}
	\\
	\forall i\in 1..k
	&
	\trg{a_i} = (\trg{id},i\cdot\trg{\mc{N}_w})
	&
	\trg{{m}_i} =\ \fun{methodEP}{\src{{M_t}_i},\trg{n_i}}
	\\
	\trg{m_{exit}'} = \fun{append}{\trg{m_{exit}}, \fun{preamble}{\cdot}}
	\\
	\trg{m_{code}} = \trg{{m}_1} + \cdots + \trg{{m}_k} + \trg{m_r} + \trg{m_c} + \trg{m_{exit}'} + \trg{m_{ce}}
	\\
	\trg{n'} = |\trg{m_{code}}|
	&
	\trg{m} = \trg{m_{code}} + \trg{{m_d}} + \trg{m_{de}} + \trg{m_{m}}
}{
	\prot{\comp{\src{C}}} =\  \trg{m};\trg{{s}};\trg{EM'};\trg{EO'};\trg{RM}+\trg{RM_i'};\trg{RO}
}{prot-def}
}

\prot{\cdot} calculates the memory layout \trg{s} based on the functionality listed in \trg{EM}.
The memory created by \prot{\cdot} contains the following functionality.
Each method exported to external code is given an entry point with code created with the \fun{methodEP}{\cdot} function.
Function \fun{returnEP}{\cdot} yields memory \trg{m_r} for the return entry point.
The code section of \trg{m} is completed with the \fun{extraCode}{\cdot}, containing the implementation of all the helper functions described above (e.g., \fun{dynamicTypechecks}{\cdot}, \fun{resetFlagsAndRegs}{\cdot} etc.) and with \trg{m_c}, the code generated by \comp{\cdot}.
Additionally, the code section appends the \fun{preamble}{\cdot} procedure at the exit-point code provided by \comp{\cdot} to ensure that it is always executed before exiting the module.
After the code section, \trg{m} contains \fun{extraData}{\cdot}, i.e., all the data needed by the helper functions (e.g., the encodings of method signatures and of types) and by \trg{m_d}, the data generated by \comp{\cdot}.
Finally the data section contains the masking table \trg{m_m}, which is obtained with function \fun{maskingTable}{\cdot}; as memories are infinite, the data section of a module has no boundaries.

What \prot{\cdot} returns exports the same methods and objects as \comp{\cdot}.
The former are bound to entry points and the latter have their ids masked.

\subsection{The Secure Linker: \mylink{\cdot}}\label{sec:mylink}
This section presents function \mylink{\cdot}, which inputs and returns \aim modules; it is formalised as follows:
\begin{align*}
	\mylink{\trg{P},\trg{P'}} =&\ \trg{M_1}\mathbin{\trg{\uplus}}\cdots\mathbin{\trg{\uplus}}\trg{M_n}\mathbin{\trg{\uplus}}\trg{M_1'}\mathbin{\trg{\uplus}}\cdots\mathbin{\trg{\uplus}}\trg{M_m'}
	\\
	\text{where }&\ \trg{P}\equiv\trg{M_1}\mathbin{\trg{\uplus}}\cdots\mathbin{\trg{\uplus}}\trg{M_n}
	\\
	\text{and }&\ \trg{P'}\equiv\trg{M_1'}\mathbin{\trg{\uplus}}\cdots\mathbin{\trg{\uplus}}\trg{M_m'}
\end{align*}
Modules \trg{M_i} are obtained through calls to \compaim{\cdot}, they all are a pair of a compiled \jem class and an instance of \sys.
Informally, the $\mathbin{\trg{\uplus}}$ operator does the following:
\begin{itemize}
\item it performs \aim-level joining of modules (\trg{+}, presented in \Cref{sec:coreaim}) ensuring that only one occurrence of \sys is present in the resulting component;
\item it initialises the resulting \sys table \mc{G} with all exported object bindings for all \trg{M_i}, so that static objects are registered in \mc{G}.
\end{itemize}

%% file: bodies/disc.tex
\section{Discussion}\label{sec:disc}
This section presents how to extend \compaim{\cdot} to a source language supporting object-orientation (\Cref{sec:objor})  and how to support multi-register data (\Cref{sec:crypt}).

\subsection{Supporting Object-Orientation}\label{sec:objor}
\jem can be made object-oriented by adding
support for interfaces, inheritance and dynamic dispatch. We suggest doing
this by adding interfaces to \jem, making class types private to a
component and only using interface types in the types of
cross-component method calls and returns~\cite{javaJr,scoo-j}. Since
different \jem components can implement the same interface, a number
of concerns arise. Nevertheless, we believe the concerns can be
adequately addressed and this section also describes a possible change
to our compiler which we believe is a way to solve the concerns.

\begin{problem}[Module id at the target level]\label{prob:oidformat}
Consider two modules \trg{M_1} and \trg{M_2} containing the compilation of two classes that implements the same interface \src{i}.
Code in another module \trg{M} could input objects of that interface.
However, \trg{M} cannot be sure of where the object id is located unless it indicates its module id.
In fact, all that \trg{M} knowns when receiving an argument is the type \src{i}, but both \trg{M_1} and \trg{M_2} implement it. 
\end{problem}

To address \Cref{prob:oidformat}, cross-module object ids need to state the module id where they reside; formally $\compaim{\src{o}}=\trg{\pi},\trg{id}$.
Unless \aim is extended to merge this information in nonces, object ids span multiple words (\Cref{sec:crypt} describes how to securely compile them).

Often, object-oriented languages implements dynamic dispatch as vtables that are located at the address where an object is compiled.
With \PMA, \trg{M_1} has no access to the vtable of \src{o} if \src{o} is allocated in \trg{M_2}.
The single-module version of \aim~\cite{scoo-j} used entry points as a vtable but this does not hold when multiple modules are considered.
\begin{problem}[Dynamic dispatch]\label{prob:dyndis} 
Consider two modules \trg{M_1} and \trg{M_2} containing the compilation of two classes implementing the same interface \src{Bank}.
Module \trg{M_1} also implements interface \src{Account} while module \trg{M_2} also implements interface \src{Currency}.
Since methods need to be sorted alphabetically based on their namespaces (which include interface names), methods for \src{Bank} in \trg{M_1} will have different entry points than the same methods in \trg{M_2}.

A module \trg{M} interacting with \trg{M_1} and \trg{M_2} however should not need to know the full specification of which component contains which interfaces, as this would defeat the purpose of object orientation.
All that \trg{M} knows is that somewhere outside its memory, interface \src{Bank} is implemented.
When receiving an object from \trg{M_1} or from \trg{M_2} that implements the \src{Bank} interface, \trg{M} needs to be able to calculate where to jump in order to call methods on it.
\end{problem}

To address \Cref{prob:dyndis}, a single method entry point is created at address \trg{\mc{N}_w} instead of an entry point per method implemented in a component.
That entry point serves as a dynamic dispatch entry point.
All cross-module method calls update the calling convention to use \asm{r_6} as a container for the encoding of the method to be called.
Formally, indicate a method encoding as follows $\comp{\src{m}}=\trg{n}$.

When a module receives an object id of the form $\trg{\pi},\trg{id}$, it can call method \src{m} on it by jumping to address $(\trg{id},\trg{0})$ and setting \asm{r_6} to $\comp{\src{m}}$.
Parameters and the current object are then passed via registers \asm{r_7} onwards.

The dynamic dispatch entry point must perform a new check: \asm{r_6} must be a valid method encoding. 
All methods implemented in interfaces are known statically, so a module can save this information in its memory to encode this check.
If this check succeeds, the execution continues by dispatching to the checks for method entry points of \Cref{sec:prot}.

\subsection{Supporting Multi-Register Data}\label{sec:crypt}
In some cases, securely-compiled code needs to communicate data that spans multiple registers.
For example, the source language could be extended to support computation on complex numbers, or the format of object ids could vary (as discussed in \Cref{sec:objor}).
Communicating data that spans multiple registers can be done by extending \aim with cryptographic functions (either in the form of instructions, as presented here, or as module-internal procedures).


First, all modules must have a public and a private key, all modules know each other's public keys but the private one is confidential to the module~\cite{sancus}.
Second, cryptograhic functions for signing and verifying signatures are needed.
\begin{center}
\begin{tabular}{p{\fcol\textwidth} | p{\fcol\textwidth}}
\hline
\hline
\asm{sign\ r_s\ r_e\ r_k\ r_d}
 	& Sign all registers from \asm{r_s} to \asm{r_e} with the (private) key found in register \asm{r_k} and place the result in register \asm{r_d}.
 		\\
\asm{verify\ r_s\ r_e\ r_k\ r_d}
 	& Set the \trg{\ms{ZF}} according to whether the signature found in register \asm{r_d} has been created for all registers from \asm{r_s} to \asm{r_e} with the dual key of that found in register \asm{r_k}. 
\end{tabular}
\end{center}

Consider the case of multi-register object ids of \Cref{sec:objor}; cross-module object ids need to be changed into a triplet, so they span three registers.
The triplet consists of: a masking index (as in \Cref{sec:prot}), a type encoding and a signature of the two; the signature prevents \Cref{prob:oidshuff} below; Formally: $\compaim{\src{o}}=\trg{\pi},\trg{n},\trg{w}$.
\begin{problem}[Object-id shuffling]\label{prob:oidshuff}
Consider two objects \src{o_1} and \src{o_2} residing in two modules \trg{M_1} and \trg{M_2} implementing two different classes.
Without signatures, those objects would have the following cross-module object ids: $\compaim{\src{o_1}}=\trg{\pi_1},\trg{n_1}$ and $\compaim{\src{o_2}}=\trg{\pi_2},\trg{n_2}$.
An attacker can forge new object ids as follows: $\trg{\pi_1},\trg{n_2}$ and $\trg{\pi_2},\trg{n_1}$.
Another module \trg{M_3} receiving the forged objects from the attacker has no way to tell whether they are forged by inspecting them. 
\end{problem}
With the signature in place, the object ids of \src{o_1} and \src{o_2} change as follows: $\compaim{\src{o_1}}=\trg{\pi_1},\trg{n_1},\trg{w_1}$ and $\compaim{\src{o_2}}=\trg{\pi_2},\trg{n_2},\trg{w_2}$.
An attacker can still forge object ids by creating $\trg{\pi_1},\trg{n_2},\trg{w_2}$, but \trg{M_3} can verify the signature part of the triplet, therefore finding out when forgery has taken place.

With this approach, some auxiliary functions need to be changed.
Function \fun{loadObjects}{\cdot} needs to verify that the object ids of all externally-located modules have been signed with the proper key and that the key corresponds to the module that implements the class \trg{n} mentioned in the object id.
As a module implements a single class, the type encoding tells a module which public key to use for the verification.
Function \fun{updateMaskingTable}{\cdot} needs to ensure that communicated cross-module object id are signed.

Concerning \sys, \fun{registerObj}{\cdot} is not needed and the pseudo code inserted by \prot{\cdot} does not call it.
Procedure \fun{testObj}{\cdot} can be implemented locally in a module, since due to the new cross-module object ids the type of an object is part of its id.

%% file: bodies/proof.tex
\section{Full-Abstraction and Modularity of \compaim{\cdot}} \label{sec:proof}
As stated in \Cref{sec:intro}, a compiler is fully-abstract if it preserves and reflects contextual equivalence of source and target components.
This section briefly discusses how contextual equivalence is reflected (\Cref{sec:compcorrdir}) and preserved (\Cref{sec:compsecdir}) for \compaim{\cdot}.
Then it concludes by presenting the proof sketch of full-abstraction and modular full-abstraction of \compaim{\cdot} (\Cref{sec:allproofs}).

\subsection{\compaim{\cdot} Reflects Contextual Equivalence}\label{sec:compcorrdir}
Proving that \compaim{\cdot} reflects contextual equivalence (\Cref{thm:compaimcorr}) is analogous to proving that it is correct and adequate.
\begin{theorem}[\compaim{\cdot} preserves behaviour]\label{thm:compaimcorr}
$\forall\src{\mc{C}_1},\src{\mc{C}_2}.\compaim{\src{\mc{C}_1}}\ceqaim\compaim{\src{\mc{C}_2}}\Rightarrow\src{\mc{C}_1}\ceqjem\src{\mc{C}_2}$.
\end{theorem}

Intuitively, given \Cref{ass:comp-corr}, this holds because neither \prot{\cdot} nor the addition of \sys change the semantics of compiled programs.

\subsection{\compaim{\cdot} Preserves Contextual Equivalence}\label{sec:compsecdir}
Proving that \compaim{\cdot} preserves contextual equivalence (\Cref{thm:compaimsecure}) is analogous to proving that \compaim{\cdot} is secure.
\Cref{thm:compaimsecure} intuitively states that \jem abstractions are preserved in the \aim output \compaim{\cdot} produces.
\begin{theorem}[\compaim{\cdot} is secure]\label{thm:compaimsecure} 
$\forall\src{\mc{C}_1},\src{\mc{C}_2}.\src{\mc{C}_1}\ceqjem\src{\mc{C}_2}\Rightarrow\compaim{\src{\mc{C}_1}}\ceqaim\compaim{\src{\mc{C}_2}}$.
\end{theorem}
For \Cref{thm:compaimsecure} we proceed as follows.
First, we devise a notion of trace equivalence $\teqaim$ for securely-compiled components; this is a slight adaptation of a similar trace semantics for the single-module version of \aim~\cite{llfatr-j}. 
Intuitively the trace semantics describes the behaviour of a set of modules with sets of traces, i.e., concatenation of actions such as call and return. 
Most importantly, trace semantics lets us disregard contexts when reasoning about modules behaviour.
We assume that $\teqaim$ is equivalent to $\ceqaim$, so the two notions can be exchanged. 
\begin{assumption}[Trace semantics coincides with contextual equivalence for securely compiled \jem components]\label{ass:fatracesaim}\label{thm:fatracesaim}
$\forall \src{\mc{C}_1}, \src{\mc{C}_2}. \compaim{\src{\mc{C}_1}}\ceqaim \compaim{\src{\mc{C}_2}} \iff \compaim{\src{\mc{C}_1}}\teqaim \compaim{\src{\mc{C}_2}}$.
\end{assumption}
Then we re-state \Cref{thm:compaimsecure} in contrapositive form: 
	$$\forall\src{\mc{C}_1},\src{\mc{C}_2}.\compaim{\src{\mc{C}_1}}\nteqaim\compaim{\src{\mc{C}_2}}\Rightarrow\src{\mc{C}_1}\nceqjem\src{\mc{C}_2}$$
To achieve $\src{\mc{C}_1}\nceqjem\src{\mc{C}_2}$ we need to show (by negating \Cref{def:ceqjem}) that there exist a context that, wlog, terminates with \src{\mc{C}_1} and diverges with \src{\mc{C}_2}.
This context is said to \emph{differentiate} between \src{\mc{C}_1} and \src{\mc{C}_2}.
For this, we devise an algorithm \algo{\cdot} that can always generate such a differentiating context given two \jem components whose compiled counterparts are trace-inequivalent~\cite{faSemUML,scoo-j,llfatr-j}. 
Since \algo{\cdot} is sketched to be correct (\Cref{thm:algocorr}), we can use it to witness that two \jem components are contextually-inequivalent if their compiled counterparts are trace-inequivalent.
\begin{theorem}[Algorithm correctness]\label{thm:algocorr}
$\forall\src{\mc{C}_1},\src{\mc{C}_2},$ $\trg{\OB{\alpha}\alpha_1}\in\taim{\compaim{\src{\mc{C}_1}}},$ $\trg{\OB{\alpha}\alpha_2}\in\taim{\compaim{\src{\mc{C}_2}}},$ $\trg{\alpha_1}\neq\trg{\alpha_2},$ 
$\algo{\src{\mc{C}_1},\src{\mc{C}_2},\trg{\OB{\alpha}\alpha_1},\trg{\OB{\alpha}\alpha_2}}=\src{\mb{C}}$ such that $\src{\mb{C}}[\src{\mc{C}_1}]\divrjem\niff\src{\mb{C}}[\src{\mc{C}_2}]\divrjem$.
\end{theorem}

\subsection{Full-Abstraction and Modularity}\label{sec:allproofs}
By the theorems of \Cref{sec:compcorrdir,sec:compsecdir}, \compaim{\cdot} is fully-abstract (\Cref{thm:compaimisfa}).
\begin{theorem}[\compaim{\cdot} is fully-abstract]\label{thm:compaimisfa}
$\forall\src{\mc{C}_1},\src{\mc{C}_2}.$ $\src{\mc{C}_1}\ceqjem\src{\mc{C}_2}\iff\compaim{\src{\mc{C}_1}}\ceqaim\compaim{\src{\mc{C}_2}}$.
\end{theorem}

The novel result we are after for \compaim{\cdot} is \emph{modular full-abstraction}, i.e. components can be compiled separately and linked together afterwards without compromising security.
Indicate the linking of two \aim components \trg{P_1} and \trg{P_2} as $\mylink{\trg{P_1},\trg{P_2}}$.
This definition is derived from Ahmed's definition of horizontal compiler compositionality~\cite{amalverifcomp}; it states that we can create source-level components by joining an arbitrary number of them and by compiling them individually and then linking the result.
Additionally, we prove this result for arbitrary target-level components \trg{P} and \trg{P'} that are equivalent to securely-compiled ones.
\trg{P} and \trg{P'} can be seen as hand-optimised versions of \src{\mc{C}_2} and \src{\mc{C}_4} that respect the behaviour imposed by \compaim{\cdot}.
Formally, we have that \Cref{thm:compaimhorizcomp} is a corollary of \Cref{thm:compaimisfa}.
\begin{corollary}[Modular full-abstraction]\label{thm:compaimhorizcomp}
$\forall\src{\mc{C}_1},\src{\mc{C}_2},\src{\mc{C}_3},\src{\mc{C}_4}.$ 
$\forall\trg{P}.\compaim{\src{\mc{C}_2}}\ceqaim\trg{P},$ 
$\forall\trg{P'}.\compaim{\src{\mc{C}_4}}\ceqaim\trg{P'},$ 
$\src{\mc{C}_1};\src{\mc{C}_2}\ceqjem\src{\mc{C}_3};\src{\mc{C}_4}\iff$ $\mylink{\compaim{\src{\mc{C}_1}},\trg{P}}\ceqaim\mylink{\compaim{\src{\mc{C}_3}},\trg{P'}}$.
\end{corollary}

%% file: bodies/rw-conc.tex
\section{Related Work}\label{sec:rw}
Secure compilation through full-abstraction was pioneered by Abadi~\cite{abadiFa} and successfully applied to many different settings~\cite{scoo-j,fstar2js,nonintfree,ahmedCPS,AhmedFa,abadiLayout,Jagadeesan,popldomi}.
Parrow proved which conditions must hold in source and target languages to provide a fully-abstract compiler between the two~\cite{gcFA}.
Gorla and Nestmann concluded that full-abstraction is meaningful when it entails properties like security~\cite{faEHM}.

A large body of research provided secure compilers for a variety of languages with different language features. 
Three works are closely related to the present one.
The first one is the fully-abstract compilation scheme of single-module code to single-module \PMA of Patrignani \etal~\cite{scoo-j}, where linking is not explicitly considered.
The second one is the secure compiler targeting an extension of the PUMP machine (a tag-based assembly-level architecture that enforces micro-policies with each instruction)~\cite{crashsafemachine} by Juglaret and Hritcu~\cite{jannis}.
This work considers a very similar source language, so it incurs in very similar problems to what discussed here.
Intuitively, Juglaret and Hritcu use tags to create \PMA-like modules (with their local stacks) where jumps can only be done at specific addresses.
By relying on a very different architecture, their solution differs significantly from ours.
For example, to address \Cref{prob:callstack}, they rely on linear tags for return addresses.
Their tags also capture type information, so that dynamic typechecks (as needed for \Cref{prob:types}) are made based on tags.
As stated in \Cref{sec:intro}, the main difference between the two works is in the architecture they target: while SGX-like \PMA is readily available, it may be a while before PUMP-like machine hit the market.
The third closest secure compilation result is the (probabilistic) fully-abstract compiler to \ASLR-enhanced target languages, which also does not explicitly consider linking~\cite{abadiLayout,Jagadeesan}.
\ASLR prevents some linking problems (e.g., object guessing) but not all of them (e.g., call stack shortcutting).
We expect that \ASLR-based secure compilation can be made resilient to all linking related attacks by adopting analogous countermeasures to those described in this paper.

Most secure compilation works adopt the fully-abstract compilation notion of Abadi~\cite{abadiFa}.
These works achieve security by relying on type systems for the target language~\cite{ahmedCPS,AhmedFa,fstar2js}, cryptographical primitives~\cite{bugliesi,corin,scoojoin} and the already mentioned \ASLR~\cite{abadiLayout,Jagadeesan} and \PMA~\cite{scoo-j}.
Additionally, certain works provide secure compilers by means of type-preserving compilers~\cite{barthe,balto,embed}, though they require the target language to be well-typed.
Of all these works, only Abadi \etal~\cite{scoojoin} consider multiple modules, but in a distributed setting rather than on a single machine, so that presents different vulnerabilities than the ones considered in this paper.

\section{Conclusion and Future Work}\label{sec:conc}
This paper presented a secure compilation scheme from \jem, an object-based imperative language to \aim, an untyped assembly language enhanced with \PMA.
Because \aim explicitly deals with linking, the secure compiler developed in this paper faces a number of threats that no previous work considers.
This paper formalised the compiler \compaim{\cdot} and explaind how it withstands these new threats.
Finally, it presented how \compaim{\cdot} is fully-abstract and modular, so that the additional threats arising from linking cannot be exploited by a malicious attacker.

\smallskip

The authors foresee a number of future research trajectories for this work: integrating a garbage collector with securely compiled programs, supporting secure compilation for concurrent programs and developing secure compilation schemes for emerging security architectures such as capability machines~\cite{cheri,capsicum}.


\smallskip

{\small
{\bf Acknowledgements.}
		The authors would like to thank Amal Ahmed, C\u{a}t\u{a}lin Hri\c{t}cu, Yannis Juglaret, Raoul Strackx and the anonymous reviewers for useful comments on an earlier version of this paper.
		This work has been supported in part by the Intel Labs University Research Office. 
		This research is also partially funded by the Research Fund KU Leuven, and by the EU FP7 project NESSoS. 
		With the financial support from the Prevention of and Fight against Crime Programme of the European Union (B-CCENTRE).}

%% file: bodies2/apps.tex
\section*{Appendix Introduction}
The appendix contains formalisation of the \jem (\Cref{sec:jemform}) and \aim (\Cref{sec:aimform}) languages, the formalisation of the algorithm \algo{\cdot} (\Cref{sec:algoaim}) and proofs and proof sketches of the paper (\Cref{sec:actualproofs}).

%% file: bodies2/jemlang.tex
\section{\jem}\label{sec:jemform}
The formalisation of \jem borrows extensively from that of Java Jr.~\cite{javaJr} and from the single-component version of it~\cite{scoo-j,tome-secure-compilation}.

\subsection{Static Semantics}\label{sec:semjem}
A rule of the form ${\mc{P}}\vdash{\OB{E}}$ is intended to be ${\mc{P}}\vdash{{E}_1},\cdots,{\mc{P}}\vdash{{E}_n}$ where ${\OB{E}} = {E_1},\cdots,{E_n}$.
Denote a stack of variable bindings with $\src{\OB{S}} ::=\src{\OB{x:t}}$
\begin{center}
\typerule{Type-programs}{
	\src{\mc{P}} = \src{C_1,\cdots,C_n},
	&
	\forall i \in 1 .. n 
	&
	\src{\mc{P}}\vdash\src{C_i}
	\\
	\src{C_i}=\lstb{import}\ \src{\OB{I}};\src{\OB{X}};\ \lstb{class}\ \src{c} \{\src{K}\ \src{\OB{F_t}}\ \src{\OB{M}}\}; \src{\OB{O}}
	&
	\src{\mc{P}}\vdash\src{\OB{I}}
	&
	\src{\mc{P}}\vdash\src{\OB{X}}
}{
	\vdash\src{\mc{P}}:\ms{prg}
}{type-programs}\quad
\typerule{Interface-checking}{
	\src{\mc{P}} = \src{C_1,\cdots,C_n},
	&
	\exists i \in 1..n . 
	\\
	\src{C_i}=\lstb{import}\ \src{\OB{I}};\src{\OB{X}};\ \lstb{class}\ \src{c} \{\src{K}\ \src{\OB{F_t}}\ \src{\OB{M}}\src{\OB{M'}}\}; \src{\OB{O}}
	\\
	\src{\OB{M_t}} = \src{{M_t}_1},\cdots,\src{{M_t}_m}
	&
	\src{\OB{M}} = \src{M_1},\cdots,\src{M_m}
	\\
	\forall j \in 1..m.
	&
	\src{M_j}\vdash\src{{M_t}_j}
	&
	\src{c} \text{ is unique in }\src{\mc{P}}
}{
	\src{\mc{P}}\vdash\lstb{class-decl}\ \src{c} \{\src{\OB{M_t}}\}
}{aux-intf}\quad
\typerule{Method-checking}{
	\src{M} = \lstb{public}\ \src{m}(\src{\OB{x}}):\src{M_t}\ \{\lstb{return}\ \src{E};\}
	&
	\src{M_t} = \src{m}:\src{t}(\src{\OB{t}})\to\src{t'}
}{
	\src{M}\vdash\src{M_t}
}{aux-meth}\quad
\typerule{Extern-checking}{
	\src{\mc{P}}\vdash\src{o}:\src{c}	
}{
	\src{\mc{P}}\vdash\lstb{obj-decl}\ \src{o}:\src{c}
}{aux-extern}\quad
\typerule{Type-class}{
	\src{\mc{P}}\vdash\src{\OB{I}} 
	&
	\src{\mc{P}}\vdash\src{\OB{X}} 
	&
	\src{\mc{P}},\src{c},\src{\OB{F_t}}\vdash\src{K}
	\\
	\src{\mc{P}};\src{c}\vdash\src{\OB{M}} 
	&
	\src{\mc{P}},\src{c},\src{\OB{F_t}},\src{\OB{O}}\vdash\src{\OB{O}} 
	&
	\src{c} \text{ is unique in }\src{\mc{P}}
}{
	\src{\mc{P}}\vdash\lstb{import}\ \src{\OB{I}};\src{\OB{X}};\ \lstb{class}\ \src{c} \{\src{K}\ \src{\OB{F_t}}\ \src{\OB{M}}\}; \src{\OB{O}}
}{type-class}\quad
\typerule{Type-object}{
	\src{o} \text{ is unique in }\src{\mc{P}}
	&
	\src{\OB{F_t}} = \src{f_1}:\src{t_1},\cdots,\src{f_n}:\src{t_n},
	\\
	\src{\OB{F}} = \lstb{private}\ \src{f_1}=\src{v_1},\cdots,\lstb{private}\ \src{f_n}=\src{v_n}
	\\
	\forall i \in 1 ..n . 
	&
	\src{\mc{P}}\vdash\src{v_i}:\src{t_i}
}{
	\src{\mc{P}},\src{c},\src{\OB{F_t}},\src{\OB{O}}\vdash\lstb{object}\ \src{o}:\src{c} \{\src{\OB{F}}\}	
}{type-intf}\quad
\typerule{Type-unit}{}{
	\src{\mc{P}}\vdash\src{\lst{unit}}:\src{\lstb{Unit}}
}{type-unit}\quad
\typerule{Type-bool}{
	\src{v}\equiv\src{\lst{true}} \vee \src{v}\equiv\src{\lst{false}}
}{
	\src{\mc{P}}\vdash\src{v}:\src{\lstb{Bool}}
}{type-bool}\quad
\typerule{Type-int}{
	\src{v}\in\mb{N}
}{
	\src{\mc{P}}\vdash\src{v}:\src{\lstb{Int}}
}{type-int}\quad
\typerule{Type-obj}{
	\src{\mc{P}}\vdash\src{t}
	&
	\src{\mc{P}} = \src{C_1,\cdots,C_n},
	&
	\exists i \in 1..n . 
	\\
	\src{C_i}=\lstb{import}\ \src{\OB{I}};\src{\OB{X}};\ \lstb{class}\ \src{c} \{\src{K}\ \src{\OB{F_t}}\ \src{\OB{M}}\}; \src{\OB{O}}
	\\
	\src{O} = \src{O_1,\cdots,O_m},
	&
	\exists j \in 1..m .
	&
	\src{O_j} = \lstb{object}\ \src{o}:\src{c} \{\src{\OB{F}}\}
}{
	\src{\mc{P}}\vdash\src{v}:\src{t}
}{type-obj}\quad
\typerule{Type-constructor}{
	\src{\OB{F_t}} = \src{f_1}:\src{t_1},\cdots,\src{f_n}:\src{t_n}
	&
	\src{\OB{f}} = \src{f_1},\cdots,\src{f_n}
	&
	\src{\OB{t}} = \src{t_1},\cdots,\src{t_n}
}{
	\src{\mc{P}},\src{c},\src{\OB{F_t}}\vdash\src{c}(\src{\OB{f}}:\src{\OB{t}})\ \{\lstb{this}.\src{\OB{f'}}=\src{\OB{f}}\}
}{type-constr}\quad
\typerule{Type-method}{
	\src{\OB{x}} = \src{x_1},\cdots,\src{x_n}
	&
	\src{M_t} = \src{m} : \src{t}(\src{\OB{t}})\to\src{t'}
	&
	\src{\OB{t}} = \src{t_1},\cdots,\src{t_n}
	\\
	\src{\mc{P}};\src{x_1}:\src{t_1},\cdots,\src{x_n}:\src{t_n};\src{c}\vdash\src{E}:\src{t'};\src{\OB{S'}}
	&
	\src{m}\text{ is unique in }\src{\mc{P}}
	\\
	\forall i \in 1..n
	&
	\src{x_i}\text{ is unique in }\src{\OB{x}}
}{
	\src{\mc{P}};\src{c}\vdash\lstb{public}\ \src{m}(\src{\OB{x}}):\src{M_t}\ \{\lstb{return}\ \src{E};\}
}{type-meth}\quad
\typerule{Type-expr-var}{
	\src{x}:\src{t}\in{\src{\OB{S}}}
}{
	\src{\mc{P}};\src{\OB{S}};\src{c}\vdash\src{x}:\src{t};\src{\emptyset}
}{type-expr-var}\quad
\typerule{Type-expr-field}{
	\src{\mc{P}};\src{\OB{S}};\src{c}\vdash\src{E}:\src{c};\src{\emptyset}
	&
	\src{\mc{P}} = \src{C_1,\cdots,C_n},
	&
	\exists i \in 1..n . 
	\\
	\src{C_i}=\lstb{import}\ \src{\OB{I}};\src{\OB{X}};\ \lstb{class}\ \src{c} \{\src{K}\ \src{\OB{F_t}}\ \src{\OB{M}}\}; \src{\OB{O}}
	&
	\src{f}:\src{t}\in{\src{\OB{F_t}}}
}{
	\src{\mc{P}};\src{\OB{S}};\src{c}\vdash\src{E}.\src{f}:\src{t};\src{\emptyset}
}{type-expr-fld}\quad
\typerule{Type-expr-field-update}{
	\src{\mc{P}};\src{\OB{S}};\src{c}\vdash\src{E}:\src{c};\src{\emptyset}
	&
	\src{\mc{P}};\src{\OB{S}};\src{c}\vdash\src{E'}:\src{t};\src{\emptyset}
	&
	\src{\mc{P}} = \src{C_1,\cdots,C_n},
	\\
	\exists i \in 1..n . 
	&
	\src{C_i}=\lstb{import}\ \src{\OB{I}};\src{\OB{X}};\ \lstb{class}\ \src{c} \{\src{K}\ \src{\OB{F_t}}\ \src{\OB{M}}\}; \src{\OB{O}}
	&
	\src{f}:\src{t}\in\src{\OB{F_t}}
}{
	\src{\mc{P}};\src{\OB{S}};\src{c}\vdash\src{E}.\src{f}=\src{E'}:\src{\lstb{Unit}};\src{\emptyset}
}{type-expr-fldup}\quad
\typerule{Type-expr-method}{
	\src{\mc{P}};\src{\OB{S}};\src{c}\vdash\src{E}:\src{t};\src{\emptyset}
	&
	\src{\mc{P}};\src{\OB{S}};\src{c}\vdash\src{\OB{E}}:\src{\OB{t}};\src{\emptyset}
	&
	\src{\mc{P}} = \src{C_1,\cdots,C_n},
	\\
	\exists i \in 1..n . 
	&
	\src{C_i}=\lstb{import}\ \src{\OB{I}};\src{\OB{X}};\ \lstb{class}\ \src{c'} \{\src{K}\ \src{\OB{F_t}}\ \src{\OB{M}}\}; \src{\OB{O}}
	\\
	\src{\OB{M}} = \src{M_1,\cdots,M_m}
	\\
	\exists i \in 1..m . 
	&
	\src{M_m}=\lstb{public}\ \src{m}(\src{\OB{x}}):\src{M_t}\ \{\lstb{return}\ \src{E''};\}
	\\
	\src{M_t} = \src{m} : \src{t}(\src{\OB{t}})\to\src{t'}
}{
	\src{\mc{P}};\src{\OB{S}};\src{c}\vdash\src{E}.\src{m}(\src{\OB{E}}):\src{t'};\src{\emptyset}
}{type-expr-meth}\quad
\typerule{Type-expr-op}{
	\src{\mc{P}};\src{\OB{S}};\src{c}\vdash\src{E}:\src{t};\src{\emptyset}
	&
	\src{\mc{P}};\src{\OB{S}};\src{c}\vdash\src{E'}:\src{t'};\src{\emptyset}
	&
	\src{\mtt{op}}:\src{t\times t'\to t''}
}{
	\src{\mc{P}};\src{\OB{S}};\src{c}\vdash\src{E\ \mtt{op}\ E'}:\src{t''};\src{\emptyset}
}{type-expr-op}\quad
\typerule{Type-expr-exit}{
	\src{\mc{P}};\src{\OB{S}};\src{c}\vdash\src{E}:\src{t};\src{\emptyset}
}{
	\src{\mc{P}};\src{\OB{S}};\src{c}\vdash\lstb{exit}\ \src{E}:\src{t};\src{\emptyset}
}{type-expr-exit}\quad
\typerule{Type-expr-sequence}{
	\src{\mc{P}};\src{\OB{S}};\src{c}\vdash\src{E}:\src{t};\src{\OB{S'}}
	&
	\src{\mc{P}};\src{\OB{S},\OB{S'}};\src{c}\vdash\src{E};\src{E'}:\src{t'};\src{\OB{S''}}
}{
	\src{\mc{P}};\src{\OB{S}};\src{c}\vdash\src{E};\src{E'}:\src{t'};\src{\OB{S'},\OB{S''}}
}{type-expr-seq}\quad
\typerule{Type-expr-newvar}{
	\src{x}\notin\dom{\src{\OB{S}}}
	&
	\src{\mc{P}};\src{\OB{S}};\src{c},\src{x}:\src{t};\src{c}\vdash\src{E}:\src{t};\src{\emptyset}
}{
	\src{\mc{P}};\src{\OB{S}};\src{c}\vdash\lstb{var}\ \src{x} : \src{t} = \src{E}:\src{\lstb{Unit}};\src{x:t}
}{type-expr-newvar}\quad
\typerule{Type-expr-if}{
	\src{\mc{P}};\src{\OB{S}};\src{c}\vdash\src{E}: \src{\lstb{Bool}};\src{\emptyset}
	&
	\src{\mc{P}};\src{\OB{S}};\src{c}\vdash\src{E'}: \src{t};\src{\emptyset}
	&
	\src{\mc{P}};\src{\OB{S}};\src{c}\vdash\src{E''}: \src{t};\src{\emptyset}
}{
	\src{\mc{P}};\src{\OB{S}};\src{c}\vdash\lstb{if}\ (\src{E})\ \{\src{E'}\}\ \lstb{else}\ \{\src{E''}\} : \src{t};\src{\emptyset}
}{type-expr-if}\quad
\typerule{Type-expr-new}{
	\src{\mc{P}};\src{\OB{S}};\src{c}\vdash\src{E}: \src{\OB{t}};\src{\emptyset}
	&
	\src{\mc{P}} = \src{C_1,\cdots,C_n},
	&
	\exists i \in 1..n . 
	\\
	\src{C_i}=\lstb{import}\ \src{\OB{I}};\src{\OB{X}};\ \lstb{class}\ \src{c} \{\src{K}\ \src{\OB{F_t}}\ \src{\OB{M}}\}; \src{\OB{O}}
	&
	\src{K} = \src{c}(\src{\OB{f}}:\src{\OB{t}})\ \{\lstb{this}.\src{\OB{f'}}=\src{\OB{f}}\}
}{
	\src{\mc{P}};\src{\OB{S}};\src{c}\vdash\lstb{new}\ \src{c}(\src{\OB{X}}) :\src{c};\src{\emptyset}
}{type-expr-new}\quad
\typerule{Type-expr-val}{
	\exists{t}.
	&
	\src{\mc{P}};\src{\OB{S}};\src{c}\vdash\src{v}:\src{t}
}{
	\src{\mc{P}};\src{\OB{S}};\src{c}\vdash\src{v}
}{type-expr-val}\quad
\typerule{Type-expr-this}{
}{
	\src{\mc{P}};\src{\OB{S}};\src{c}\vdash\src{\lst{this}} : \src{c}
}{type-expr-val}\quad
\end{center}

Contexts for the evalutation of expressions are defined as follows: $\src{\mb{E}}::=\src{\mb{E}} ::=\ [\cdot ] \mid \src{\mb{E}}.\src{m}(\src{\OB{X}}) \mid \src{o}.\src{m}(\src{\OB{v}},\src{\mb{E}},\src{\OB{E}}) \mid \src{\mb{E}}.\src{f} \mid \src{\mb{E}}.\src{f}=\src{E}  \mid \src{v}.\src{f}=\src{\mb{E}} \mid \lstb{new}\ \src{c}(\src{\OB{v}},\src{\mb{E}},\src{\OB{X}}) \mid \lstb{if}(\src{\mb{E}})\{ \src{E}\}\lstb{else}\{\src{E}\} \mid \src{\mb{E}};\src{E} \mid \lstb{var}\ \src{x} : \src{t} = \src{\mb{E}} \mid \lstb{return}\ \src{\mb{E}} \mid \src{\mb{E}}\ \src{\mtt{op}\ E}\ \mid \src{v\ \mtt{op}}\ \src{\mb{E}} \mid \lstb{exit}\ \src{\mb{E}}\mid\lstb{instanceof}(\src{\mb{E}},\src{c})$.

Configurations for the dynamic semantics: $\src{k}::= (\src{\mc{P}};\src{\OB{B}}\vdash\src{\mb{E}}[\src{E}])$.

\src{\OB{B}} is a stack of stacks $\src{B}::=\src{\OB{x\mapsto v}}$, lookup and update is always done to the top stack of the stack.

Transitions of the static semantics $\to\subseteq\src{k}\times\src{k}$.

Initial state for a program $\src{\mc{P}} = \src{\mc{P}};\src{\emptyset}\vdash[\src{main}.\src{main}()]$.

Define plugging for contextual equivalence as follows: if $\src{\mb{C}} = \src{\mc{C}}$, then $\src{\mb{C}}[\src{\mc{C'}}] = \src{\mc{C},\mc{C}'}$ if $\src{\mb{C}}\compat\src{\mc{C'}}$.
The initial state of $\src{\mb{C}}[\src{\mc{C'}}]$ is $\src{\mc{C},\mc{C}'};\src{\emptyset}\vdash[\src{main}.\src{main}()]$.
Assume a context always define the \src{main} method and the \src{main} class.

\begin{center}
\typerule{\jem-eval-field-lookup}{
	\src{\mc{P}} = \src{C_1,\cdots,C_n},
	&
	\exists i \in 1..n . 
	&
	\src{C_i}=\lstb{import}\ \src{\OB{I}};\src{\OB{X}};\ \lstb{class}\ \src{c} \{\src{K}\ \src{\OB{F_t}}\ \src{\OB{M}}\}; \src{\OB{O}}
	\\
	\src{\OB{O}} = \src{O_1,\cdots,O_m},
	&
	\exists j \in 1..m .
	&
	\src{O_j} = \lstb{object}\ \src{o}:\src{c} \{\src{\OB{F}}\}
	&
	\lstb{private}\ \src{f} = \src{v}\in{\src{\OB{F}}}
}{
	\src{\mc{P}};\src{\OB{B}}\vdash\src{\mb{E}}[\src{o}.\src{f}] \to \src{\mc{P}};\src{\OB{B}}\vdash\src{\mb{E}}[\src{v}]
}{jemeval-fldlkup}\quad
\typerule{\jem-eval-field-update}{
	\src{\mc{P}} = \src{C_1,\cdots,C_n}
	&
	\exists i \in 1..n . 
	\\
	\src{C_i}=\lstb{import}\ \src{\OB{I}};\src{\OB{X}};\ \lstb{class}\ \src{c} \{\src{K}\ \src{\OB{F_t}}\ \src{\OB{M}}\}; \src{\OB{O}}
	\\
	\src{\OB{O}} = \src{O_1,\cdots,O_m}
	&
	\exists j \in 1..m .
	\\
	\src{O_j} = \lstb{object}\ \src{o}:\src{c} \{\src{\OB{F}}\}
	&
	\src{F} = \src{F_1,\cdots,F_l}
	\\
	\exists h \in 1..l .
	&
	\src{F_h} = \lstb{private}\ \src{f} = \src{v'}
	&
	\src{F_h'} = \lstb{private}\ \src{f} = \src{v}
	\\
	\src{\OB{F'}} = \src{F_1,\cdots,F_{h-1},F_h',F_{h+1},\cdots,F_l}
	&
	\src{O_j'} = \lstb{object}\ \src{o}:\src{c} \{\src{\OB{F'}}\}
	\\
	\src{\OB{O'}} = \src{O_1,\cdots,O_{j-1},O_j',O_{j+1},\cdots,O_m}
	\\
	\src{C_i'}=\lstb{import}\ \src{\OB{I}};\src{\OB{X}};\ \lstb{class}\ \src{c} \{\src{K}\ \src{\OB{F_t}}\ \src{\OB{M}}\}; \src{\OB{O'}}
	\\
	\src{\mc{P'}} = \src{C_1,\cdots,C_{i-1},C_i',C_{i+1},\cdots,C_n}
}{
	\src{\mc{P}};\src{\OB{B}}\vdash\src{\mb{E}}[\src{o}.\src{f} =\src{v}] \to \src{\mc{P'}};\src{\OB{B}}\vdash\src{\mb{E}}[\src{\lst{unit}}]
}{jemeval-fldup}\quad
\typerule{\jem-eval-method-call}{
	\src{\mc{P}} = \src{C_1,\cdots,C_n}
	&
	\exists i \in 1..n . 
	&
	\src{C_i}=\lstb{import}\ \src{\OB{I}};\src{\OB{X}};\ \lstb{class}\ \src{c} \{\src{K}\ \src{\OB{F_t}}\ \src{\OB{M}}\}; \src{\OB{O}}
	\\
	\src{\OB{O}} = \src{O_1,\cdots,O_m}
	&
	\exists j \in 1..m.
	&
	\src{O_j} =  \lstb{object}\ \src{o}:\src{c} \{\src{\OB{F}}\}
	\\
	\src{\OB{M}} = \src{M_1,\cdots, M_k}
	&
	\exists h \in 1 .. k .
	&
	\src{M_k} = \lstb{public}\ \src{m}(\src{\OB{x}}):\src{M_t}\ \{\lstb{return}\ \src{E};\}
}{
	\src{\mc{P}};\src{\OB{B}}\vdash\src{\mb{E}}[\src{o}.\src{m}(\src{\OB{v}})] \to \src{\mc{P}};\src{\emptyset;\OB{B}}\vdash\src{\mb{E}}[\src{E}[\src{o}/\src{\lst{this}}][\src{\OB{v}}/\src{\OB{x}}]]
}{jemeval-mc}\quad
\typerule{\jem-eval-method-return}{}{
	\src{\mc{P}};\src{B;\OB{B}}\vdash\src{\mb{E}}[\lstb{return}\ \src{v}] \to \src{\mc{P}};\src{\OB{B}}\vdash\src{\mb{E}}[\src{v}]
}{jemeval-rt}\quad
\typerule{\jem-eval-new}{
	\src{\mc{P}} = \src{C_1,\cdots,C_n}
	&
	\exists i \in 1..n . 
	&
	\src{C_i}=\lstb{import}\ \src{\OB{I}};\src{\OB{X}};\ \lstb{class}\ \src{c} \{\src{K}\ \src{\OB{F_t}}\ \src{\OB{M}}\}; \src{\OB{O}}
	\\
	\src{\OB{O}} = \src{O_1,\cdots,O_m}
	&
	\src{\OB{F_t}} = \src{\OB{f:t}}
	&
	\src{O} = \lstb{object}\ \src{o}:\src{c} \{\src{\OB{f} = \OB{v}}\}
	&
	\src{o} \text{ is fresh in } \src{\mc{P}}
	\\
	\src{\OB{O}'} = \src{O_1,\cdots,O_m,O}
	&
	\src{C_i'}=\lstb{import}\ \src{\OB{I}};\src{\OB{X}};\ \lstb{class}\ \src{c} \{\src{K}\ \src{\OB{F_t}}\ \src{\OB{M}}\}; \src{\OB{O}'}
	\\
	\src{\mc{P'}} = \src{C_1,\cdots,C_{i-1},C_i',C_{i+1},\cdots,C_n}
}{
	\src{\mc{P}};\src{\OB{B}}\vdash\src{\mb{E}}[\lstb{new}\ \src{c}(\src{\OB{v}})] \to \src{\mc{P'}};\src{\OB{B}}\vdash\src{\mb{E}}[\src{o}]
}{jemeval-new}\quad
\typerule{\jem-eval-if-true}{}{
	\src{\mc{P}};\src{\OB{B}}\vdash\src{\mb{E}}[\lstb{if}\ (\src{\lst{true}})\ \{\src{E_t}\}\ \lstb{else}\ \{\src{E_f}\}] \to \src{\mc{P}};\src{\OB{B}}\vdash\src{\mb{E}}[\src{E_t}]
}{jemeval-ift}\quad
\typerule{\jem-eval-if-false}{}{
	\src{\mc{P}};\src{\OB{B}}\vdash\src{\mb{E}}[\lstb{if}\ (\src{\lst{false}})\ \{\src{E_t}\}\ \lstb{else}\ \{\src{E_f}\}] \to \src{\mc{P}};\src{\OB{B}}\vdash\src{\mb{E}}[\src{E_t}]}{jemeval-iff}\quad
\typerule{\jem-eval-instanceof-true}{
	\exists\src{C} \in \src{\mc{P}}.
	&
	\src{C}=\lstb{import}\ \src{\OB{I}};\src{\OB{X}};\ \lstb{class}\ \src{c} \{\src{K}\ \src{\OB{F_t}}\ \src{\OB{M}}\}; \src{\OB{O}}
	\\
	\exists\src{O}\in\src{\OB{O}}.
	&
	\src{O} = \lstb{object}\ \src{v}:\src{c} \{\src{\OB{F}}\}
}{
	\src{\mc{P}};\src{\OB{B}}\vdash\src{\mb{E}}[\lstb{instanceof}(\src{v}:\src{c})] \to \src{\mc{P}};\src{\OB{B}}\vdash\src{\mb{E}}[\src{\lst{true}}]
}{jemeval-inst}\quad
\typerule{\jem-eval-instanceof-false}{
	\nexists\src{C} \in \src{\mc{P}}.
	&
	\src{C}=\lstb{import}\ \src{\OB{I}};\src{\OB{X}};\ \lstb{class}\ \src{c} \{\src{K}\ \src{\OB{F_t}}\ \src{\OB{M}}\}; \src{\OB{O}}
	\\
	\exists\src{O}\in\src{\OB{O}}.
	&
	\src{O} = \lstb{object}\ \src{v}:\src{c} \{\src{\OB{F}}\}}{
	\src{\mc{P}};\src{\OB{B}}\vdash\src{\mb{E}}[\lstb{instanceof}(\src{v}:\src{c})] \to \src{\mc{P}};\src{\OB{B}}\vdash\src{\mb{E}}[\src{\lst{false}}]
}{jemeval-insf}\quad
\typerule{\jem-eval-local-var}{}{
	\src{\mc{P}};\src{\OB{B}}\vdash\src{\mb{E}}[\lstb{var}\ \src{x:t} = \src{v}] \to \src{\mc{P}};\src{\OB{B}}+(\src{x\mapsto v})\vdash\src{\mb{E}}[\src{v}]
}{jemeval-locv}\quad
\typerule{\jem-eval-lookup-var}{
	\src{\OB{B}}(\src{x}) = \src{v}
}{
	\src{\mc{P}};\src{\OB{B}}\vdash\src{\mb{E}}[\src{x}] \to \src{\mc{P}};\src{\OB{B}}\vdash\src{\mb{E}}[\src{v}]
}{jemeval-lookv}\quad
\typerule{\jem-eval-op}{
	\src{v\ \mtt{op}\ v'}=\src{v''}
}{
	\src{\mc{P}};\src{\OB{B}}\vdash\src{\mb{E}}[\src{v\ \mtt{op}\ v'}] \to \src{\mc{P}};\src{\OB{B}}\vdash\src{\mb{E}}[\src{v''}]
}{jemeval-op}\quad
\typerule{\jem-eval-exit}{}{
	\src{\mc{P}};\src{\OB{B}}\vdash\src{\mb{E}}[\lstb{exit}\ \src{v}] \to \src{\mc{P}};\src{\OB{B}}\vdash\src{v}
}{jemeval-ex}\quad
\typerule{\jem-eval-concatenation}{}{
	\src{\mc{P}};\src{\OB{B}}\vdash\src{\mb{E}}[\src{v};\src{E}] \to \src{\mc{P}};\src{\OB{B}}\vdash\src{\mb{E}}[\src{E}]
}{jemeval-conc}\quad
\end{center}

Define compatibility as follows: 
\begin{center}
\typerule{Component satisfaction}{
	\src{\mc{C}}\Vdash\src{\mc{C'}}
	&
	\src{\mc{C}'}\Vdash\src{\mc{C}}
}{
	\src{\mc{C}}\compat\src{\mc{C'}}
}{}
\end{center}

Define a component to satisfy another as $\src{\mc{C}'}\Vdash\src{\mc{C}}$.
Formally: 
\begin{center}
\typerule{Component satisfaction base}{
}{
	\src{\mc{C}'}\Vdash\src{\emptyset}
}{}\quad
\typerule{Component satisfaction inductive}{
	\src{\mc{C}'}\Vdash\src{C_1}
	&
	\src{\mc{C}'}\Vdash\src{\mc{C}}	
}{
	\src{\mc{C}'}\Vdash\src{C_1};\src{\mc{C}}
}{}\quad
\typerule{Component satisfaction single}{
	\src{\mc{C}} = \src{C_1},\cdots,\src{C_n}
	&
	\src{C} = \lstb{import}\ \src{\OB{I}};\src{\OB{X}};\ \lstb{class}\ \src{c} \{\src{K}\ \src{\OB{F_t}}\ \src{\OB{M}}\}; \src{\OB{O}}
	\\
	\src{\OB{I}} = \src{I_1},\cdots,\src{I_m}
	&
	\src{\OB{X}} = \src{X_1},\cdots,\src{X_k}
	\\
	\forall i \in 1..m, \exists j\in 1..n . 
	&
	\src{I_i} = \lstb{class-decl}\ \src{i} \{\src{\OB{M_t}}\}
	\\
	\src{C_j} = \lstb{import}\ \src{\OB{I}'};\src{\OB{X}'};\ \lstb{class}\ \src{i} \{\src{K'}\ \src{\OB{F_t}'}\ \src{\OB{M}'}\}; \src{\OB{O}'}
	&
	\src{\OB{M}'}\vdash\src{\OB{M_t}}
	\\
	\forall g \in 1..k, \exists h \in 1 .. n.
	&
	\src{X_g} = \lstb{obj-decl}\ \src{o}:\src{c};
	\\
	\src{C_h} = \lstb{import}\ \src{\OB{I}''};\src{\OB{X}''};\ \lstb{class}\ \src{c} \{\src{K''}\ \src{\OB{F_t}''}\ \src{\OB{M}''}\}; \src{\OB{O}''}
	&
	\src{\OB{O}''} = \src{{O}_1''},\cdots,\src{{O}_l''}
	\\
	\exists h \in 1 .. l .
	&
	\src{O_h''}= \lstb{object}\ \src{o}:\src{c} \{\src{\OB{F}}\}	
}{
	\src{\mc{C}}\Vdash\src{C}
}{}\quad
\end{center}
A component satisfies a class if all import declarations in the class are implemented in the component.
So, for all interfaces imported in the class, a class exists in the component with the same name and at least the same required methods.
For all externs imported, an object exists with the same name and the same type.

%% file: bodies2/aimlang.tex
\section{\aim, Formally}\label{sec:aimform}
The formalisation of \aim borrows extensively from that of the single-module version by Patrignani and Clarke~\cite{llfatr-j}.

Define $(\trg{P_1})\compat(\trg{P_2})$ as follows.
\begin{center}
\typerule{program compatibility}{
	\trg{t_1} = \trg{\OB{EM}_1};\trg{\OB{EO}_1};\trg{\OB{RM}_1};\trg{\OB{RO}_1}
	&
	\trg{t_2} = \trg{\OB{EM}_2};\trg{\OB{EO}_2};\trg{\OB{RM}_2};\trg{\OB{RO}_2}
	\\
	\trg{\OB{EM}_1}\Vdash\trg{\OB{RM}_2}
	&
	\trg{\OB{EM}_2}\Vdash\trg{\OB{RM}_1}
	&
	\trg{\OB{EO}_1}\Vdash\trg{\OB{RO}_2}
	&
	\trg{\OB{EO}_2}\Vdash\trg{\OB{RO}_1}
}{
	(\trg{m_1},\trg{\OB{s}_1},\trg{t_1})\compat(\trg{m_2},\trg{\OB{s}_2},\trg{t_2})
}{prcomp}\quad
\typerule{Exported satisfy required -meth- ind}{
	\forall i \in 1 .. n
	&
	\trg{\OB{EM}}\Vdash\trg{{RM}_i}	
}{
	\trg{\OB{EM}}\Vdash\trg{{RM}_1,\cdots,RM_n}
}{em-s-rm-i}\quad
\typerule{Exported satisfy required -meth- base}{
	\exists\trg{a}. 
	&
	\src{m}:\src{M_t}\mapsto\trg{a}\in\trg{\OB{EM}}
}{
	\trg{\OB{EM}}\Vdash\src{m}:\src{M_t}\mapsto\trg{\iota};\trg{\sigma}
}{em-s-rm-b}\quad
\typerule{Exported satisfy required -obj- ind}{
	\forall i \in 1 .. n
	&
	\trg{\OB{EO}}\Vdash\trg{{RO}_i}	
}{
	\trg{\OB{EO}}\Vdash\trg{{RO}_1,\cdots,RO_n}
}{eo-s-ro-i}\quad
\typerule{Exported satisfy required -obj- base}{
	\exists\trg{n}. 
	&
	\src{o}:\src{c}\mapsto\trg{n}\in\trg{\OB{EO}}
}{
	\trg{\OB{EO}}\Vdash\src{o}:\src{c}\mapsto\trg{\sigma}
}{eo-s-ro-b}\quad
\end{center}

Define a substitution of a symbol \trg{\sigma} for a word \trg{w} as $\trg{\eta}=[\trg{w}/\trg{\sigma}]$.

Define $\trg{P_1}+\trg{P_2}$ as follows.
\begin{center}
\typerule{program merging}{
	\trg{m_1}\compat\trg{m_2}
	&
	\trg{\OB{s}_1}\compat\trg{\OB{s}_2}
	&
	\trg{t_1}+\trg{t_2} = \trg{t};\trg{\OB{\eta}}
	&
	\trg{m} = (\trg{m_1}+\trg{m_2})\trg{\OB{\eta}}
}{
	(\trg{m_1},\trg{\OB{s}_1},\trg{t_1})\trg{+}(\trg{m_2},\trg{\OB{s}_2},\trg{t_2}) = (\trg{m},\trg{\OB{s}_1}\trg{\OB{s}_2},\trg{t})
}{prcomp}\quad
\typerule{mem - compat}{
	\dom{\trg{m_1}}\cap\dom{\trg{m_2}}=\emptyset
}{
	\trg{m_1}\compat\trg{m_2}
}{mcomp}\quad
\typerule{desc - compat - ind}{
	\forall i \in 1 .. n
	&
	\trg{\OB{s}}\compat\trg{s_i}	
}{
	\trg{\OB{s}}\compat\trg{s_1,\cdots,s_n}
}{dcom-i}\quad
\typerule{desc - compat - base}{
	\trg{\OB{s}} = \trg{s_1,\cdots,s_n}
	&
	\forall i \in 1 .. n.
	&
	\trg{s_i}\equiv(\trg{id}_i,\trg{n_c}_i,\trg{n}_i)
	&
	\trg{id_i}\neq\trg{id}
}{
	\trg{\OB{s}}\compat(\trg{id},\trg{n_c},\trg{n})
}{dcom-b}\quad
\typerule{mem sub - ind}{
	\trg{m} = \trg{a_1\mapsto w_1\trg{\OB{\eta}},\cdots,a_n\mapsto w_n\trg{\OB{\eta}}}
}{
	\trg{a_1\mapsto w_1,\cdots,a_n\mapsto w_n}\trg{\OB{\eta}} = \trg{m}
}{msi}\quad
\typerule{sub-s-ind}{
	\forall i \in 1..n
	&
	\trg{w_{i}'}=\trg{w_{i-1}'}\trg{{[\trg{w_i}/\trg{\sigma_i}]}}
}{
	\trg{w}\trg{{[\trg{w_1}/\trg{\sigma_1}]}},\cdots,\trg{{[\trg{w_n}/\trg{\sigma_n}]}} = \trg{w_n'}
}{ssw}\quad
\typerule{sub-s-n}{}{
	\trg{n}\trg{{[\trg{w}/\trg{\sigma}]}} = \trg{n}
}{ssn}\quad
\typerule{sub-s-p}{}{
	\trg{\pi}\trg{{[\trg{w}/\trg{\sigma}]}} = \trg{\pi}
}{ssno}\quad
\typerule{sub-s}{
	\trg{\sigma'}=\trg{\sigma}
}{
	\trg{\sigma'}\trg{{[\trg{w}/\trg{\sigma}]}} = \trg{w}
}{sssyt}\quad
\typerule{sub-s-no}{
	\trg{\sigma'}\neq\trg{\sigma}
}{
	\trg{\sigma'}\trg{{[\trg{w}/\trg{\sigma}]}} = \trg{\sigma'}
}{sssyf}\quad
\typerule{mem add}{
	\trg{m_1}\compat\trg{m_2}
}{
	\trg{m_1}+\trg{m_2} = \trg{m_1}\trg{m_2}
}{madd}\quad
\typerule{linking tables merging}{
	\trg{\OB{EM}_1};\trg{\OB{EO}_1};\trg{\OB{RM}_1};\trg{\OB{RO}_1} \compat \trg{\OB{EM}_2};\trg{\OB{EO}_2};\trg{\OB{RM}_2};\trg{\OB{RO}_2}
	\\
	\trg{\trg{\OB{EM}}} = \trg{\trg{\OB{EM}_1}}\trg{\trg{\OB{EM}_2}}
	&
	\trg{\trg{\OB{EO}}} = \trg{\trg{\OB{EO}_1}}\trg{\trg{\OB{EO}_2}}
	&
	\trg{t}= \trg{\trg{\OB{EM}}};\trg{\OB{EO}};
	\\
	\trg{\OB{RM}_1}+\trg{\OB{EM}_2} = \trg{\OB{RM}'}+\trg{\OB{\eta'}}
	&
	\trg{\OB{RM}_2}+\trg{\OB{EM}_1} = \trg{\OB{RM}''}+\trg{\OB{\eta''}}
	\\
	\trg{\OB{RO}_1}+\trg{\OB{EO}_2} = \trg{\OB{RO}'}+\trg{\OB{\eta'''}}
	&
	\trg{\OB{RO}_2}+\trg{\OB{EO}_1} = \trg{\OB{RO}''}+\trg{\OB{\eta''''}}
	\\
	\trg{\OB{RM}} = \trg{\OB{RM}'}\trg{\OB{RM}''}
	&
	\trg{\OB{RO}} = \trg{\OB{RO}'}\trg{\OB{RO}''}
	&
	\trg{\OB{\eta}}=\trg{\OB{\eta'}}\trg{\OB{\eta''}}\trg{\OB{\eta'''}}\trg{\OB{\eta''''}}
}{
	\trg{\OB{EM}_1};\trg{\OB{EO}_1};\trg{\OB{RM}_1};\trg{\OB{RO}_1} + \trg{\OB{EM}_2};\trg{\OB{EO}_2};\trg{\OB{RM}_2};\trg{\OB{RO}_2} = \trg{t};\trg{\OB{\eta}}
}{ltm}\quad
\typerule{rm - merged - em - ind}{
	\trg{\OB{RM}} = \trg{RM_1},\cdots,\trg{RM_n}
	&
	\forall i \in 1 .. n
	&
	\trg{RM_i}+\trg{\OB{EM}} = \trg{RM_i'}+\trg{\eta'_i}\trg{\eta''_i}
	\\
	\trg{\OB{RM}'}=\trg{RM_1'},\cdots,\trg{RM_n'}
	&
	\trg{\OB{\eta'}} = \trg{\eta'_1}\trg{\eta''_1},\cdots,\trg{\eta'_n}\trg{\eta''_n}
}{
	\trg{\OB{RM}}+\trg{\OB{EM}} = \trg{\OB{RM}'}+\trg{\OB{\eta'}}
}{rmmrm-i}\quad
\typerule{rm - merged - em - base}{
	\src{m}:\src{M_t}\mapsto\trg{(id,n)}\in\trg{\OB{EM}}
}{
	\src{m}:\src{M_t}\mapsto\trg{\iota};\trg{\sigma}+\trg{\OB{EM}} = \trg{\emptyset}+\trg{[\trg{n}/\trg{\sigma}]}\trg{[\trg{id}/\trg{\iota}]}
}{rmmem-b-ok}\quad
\typerule{rm - merged - em - base2}{
	\src{m}:\src{M_t}\notin\dom{\trg{\OB{EM}}}
}{
	\src{m}:\src{M_t}\mapsto\trg{\iota};\trg{\sigma}+\trg{\OB{EM}} = \src{m}:\src{M_t}\mapsto\trg{\iota};\trg{\sigma}+\trg{\emptyset}\trg{\emptyset}
}{rmmem-b-no}\quad
\typerule{ro - merged - eo - ind}{
	\trg{\OB{RO}} = \trg{RO_1},\cdots,\trg{RO_n}
	&
	\forall i \in 1 .. n
	&
	\trg{RO_i}+\trg{\OB{EO}} = \trg{RO_i'}+\trg{\eta'_i}
	\\
	\trg{\OB{RO}'}=\trg{RO_1'},\cdots,\trg{RO_n'}
	&
	\trg{\OB{\eta'}} = \trg{\eta'_1},\cdots,\trg{\eta'_n}
}{
	\trg{\OB{RO}}+\trg{\OB{EO}} = \trg{\OB{RO}'}+\trg{\OB{\eta'}}
}{eo-i}\quad
\typerule{ro - merged - eo - base}{
	\src{m}:\src{M_t}\mapsto\trg{n}\in\trg{\OB{EM}}
}{
	\src{m}:\src{M_t}\mapsto\trg{\sigma}+\trg{\OB{EM}} = \trg{\emptyset}+\trg{[\trg{n}/\trg{\sigma}]}
}{rmmem-b-ok-o}\quad
\typerule{ro - merged - eo - base2}{
	\src{m}:\src{M_t}\notin\dom{\trg{\OB{EM}}}
}{
	\src{m}:\src{M_t}\mapsto\trg{\sigma}+\trg{\OB{EM}} = \src{m}:\src{M_t}\mapsto\trg{\sigma}+\trg{\emptyset}
}{rmmem-b-no-o}\quad
\end{center}

These rules extend the single-module \PMA access control policy to account for multiple modules.
Assume $\trg{s}=(\trg{id},\trg{n_c},\trg{n_t})$.
\begin{center}
\small
	\typerule{Aux-unprotected}{
		\trg{a}=(\trg{0},\trg{n})
	}{
		\trg{\OB{s}}\vdash\acp{unprotected}(\trg{a})
	}{aux-unp}\quad
	\typerule{Aux-entrypoint}{
		\trg{a} = (\trg{id},\trg{n'})
		&
		\trg{n'}=\trg{n_b}+ m\cdot \trg{\mc{N}_w}
		&
		m \in \mb{N}
		&
		m<\trg{n_t}
	}{
		\trg{s}\vdash\acp{entryPoint}(\trg{a})
	}{aux-ep}\quad
	\typerule{Aux-methodEntryPoint}{
		\trg{a} = (\trg{id},\trg{n'})
		&
		\trg{n'}=\trg{n_b}+ m\cdot \trg{\mc{N}_w}
		&
		\trg{n'}=\trg{n_b}+ (m)\cdot \trg{\mc{N}_w}
		&
		m\in 1 .. \trg{n_t}-1
	}{
		\trg{s}\vdash\acp{entryPoint}(\trg{a})
	}{aux-mep}\quad
	\typerule{Aux-returnEntryPoint}{
		\trg{a} = (\trg{id},\trg{n'})
		&
		\trg{n'}=\trg{n_b}+ (\trg{n_t}-1)\cdot \trg{\mc{N}_w}
	}{
		\trg{s}\vdash\acp{returnEntryPoint}(\trg{a})
	}{aux-rbep}\quad
	\typerule{Aux-code}{
		\trg{a} = (\trg{id},\trg{n'})
		&
		\trg{n_b}\leq\trg{n'}<\trg{n_b+n_c}
	}{
		\trg{s}\vdash\acp{codeRange}(\trg{a})
	}{aux-c}\quad
	\typerule{Aux-data}{
		\trg{a} = (\trg{id},\trg{n'})
		&
		\trg{n_b+n_c}\leq\trg{n'}
	}{
		\trg{s}\vdash\acp{dataRange}(\trg{a})
	}{aux-}\quad
	\typerule{Aux-range}{
		\trg{a} = (\trg{id},\trg{n'})
		&
		\trg{n'}\in\mb{N}	
	}{
		\trg{s}\vdash\acp{range}(\trg{a})
	}{aux-r}\quad
	\typerule{Aux-internalJump-module}{
		\trg{s}\vdash\acp{codeRange}((\trg{id},\trg{n}),(\trg{id},\trg{n'}))
	}{
		\trg{{s}}\vdash\acp{internalJump}((\trg{id},\trg{n}),(\trg{id},\trg{n'}))
	}{aux-iju}\quad
	\typerule{Aux-writeAllowed-module}{
		\trg{s}\vdash\acp{codeRange}(\trg{a})
		&
		\trg{s}\vdash\acp{dataRange}(\trg{a'})
	}{
		\trg{{s}}\vdash\acp{writeAllowed}(\trg{a},\trg{a'})
	}{aux-wau}\quad
	\typerule{Aux-readAllowed-module}{
		\trg{s}\vdash\acp{codeRange}(\trg{a})
		&
		\trg{s}\vdash\acp{range}(\trg{a'})
	}{
		\trg{{s}}\vdash\acp{readAllowed}(\trg{a},\trg{a'})
	}{aux-rau}\quad
	\typerule{Aux-crossjump}{
		\exists \trg{s}\in\trg{\OB{s}}. \trg{s}\vdash\acp{range}(\trg{a})
		&
		\exists \trg{s'}\in\trg{\OB{s}}. \trg{s'}\vdash\acp{entryPoint}(\trg{a'})
		&
		\trg{s}\neq\trg{s'}
	}{
		\trg{\OB{s}}\vdash\acp{crossJump}(\trg{a},\trg{a'})
	}{aux-cj}\quad
	\typerule{Aux-internalJump}{
		\exists \trg{s}\in\trg{\OB{s}}. \trg{s}\vdash\acp{internalJump}(\trg{a},\trg{a'})
	}{
		\trg{\OB{s}}\vdash\acp{internalJump}(\trg{a},\trg{a'})
	}{aux-ij}\quad
	\typerule{Aux-internalJump2}{
		\trg{\OB{s}}\vdash\acp{unprotected}(\trg{a})
		&
		\trg{\OB{s}}\vdash\acp{unprotected}(\trg{a'})
	}{
		\trg{\OB{s}}\vdash\acp{internalJump}(\trg{a},\trg{a'})
	}{aux-ij2}\quad
	\typerule{Aux-validJump}{
		\trg{\OB{s}}\vdash\acp{internalJump}(\trg{a},\trg{a'})
	}{
		\trg{\OB{s}}\vdash\acp{validJump}(\trg{a},\trg{a'})
	}{aux-vj}\quad
	\typerule{Aux-validJump2}{
		\trg{\OB{s}}\vdash\acp{crossJump}(\trg{a},\trg{a'})
	}{
		\trg{\OB{s}}\vdash\acp{validJump}(\trg{a},\trg{a'})
	}{aux-vj2}\quad
	\typerule{Aux-validJump3}{
		\trg{\OB{s}}\vdash\acp{unprotected}(\trg{a})
		&
		\trg{\OB{s}}\vdash\acp{entryPoint}(\trg{a'})
	}{
		\trg{\OB{s}}\vdash\acp{validJump}(\trg{a},\trg{a'})
	}{aux-vj3}\quad
	\typerule{Aux-validJump4}{
		\trg{\OB{s}}\vdash\acp{protected}(\trg{a})
		&
		\trg{\OB{s}}\vdash\acp{unprotected}(\trg{a'})
	}{
		\trg{\OB{s}}\vdash\acp{validJump}(\trg{a},\trg{a'})
	}{aux-vj4}\quad
	\typerule{Aux-writeAllowed}{
		\exists \trg{s}\in\trg{\OB{s}}. \trg{s}\vdash\acp{writeAllowed}(\trg{a},\trg{a'})
	}{
		\trg{\OB{s}}\vdash\acp{writeAllowed}(\trg{a},\trg{a'})
	}{aux-wa}\quad
	\typerule{Aux-writeAllowed2}{
		\trg{\OB{s}}\vdash\acp{unprotected}(\trg{a})
		&
		\trg{\OB{s}}\vdash\acp{unprotected}(\trg{a'})
	}{
		\trg{\OB{s}}\vdash\acp{writeAllowed}(\trg{a},\trg{a'})
	}{aux-wa2}\quad
	\typerule{Aux-writeAllowed3}{
		\exists \trg{s}\in\trg{\OB{s}}. \trg{s}\vdash\acp{range}(\trg{a})
		&
		\trg{\OB{s}}\vdash\acp{unprotected}(\trg{a'})
	}{
		\trg{\OB{s}}\vdash\acp{writeAllowed}(\trg{a},\trg{a'})
	}{aux-wa3}\quad
	\typerule{Aux-readAllowed}{
		\exists \trg{s}\in\trg{\OB{s}}. \trg{s}\vdash\acp{readAllowed}(\trg{a},\trg{a'})
	}{
		\trg{\OB{s}}\vdash\acp{readAllowed}(\trg{a},\trg{a'})
	}{aux-ra}\quad
	\typerule{Aux-readAllowed2}{
		\trg{\OB{s}}\vdash\acp{unprotected}(\trg{a})
		&
		\trg{\OB{s}}\vdash\acp{unprotected}(\trg{a'})
	}{
		\trg{\OB{s}}\vdash\acp{readAllowed}(\trg{a},\trg{a'})
	}{aux-ra2}\quad
	\typerule{Aux-readAllowed3}{
		\exists \trg{s}\in\trg{\OB{s}}. \trg{s}\vdash\acp{range}(\trg{a})
		&
		\trg{\OB{s}}\vdash\acp{unprotected}(\trg{a'})
	}{
		\trg{\OB{s}}\vdash\acp{readAllowed}(\trg{a},\trg{a'})
	}{aux-ra3}\quad
	\typerule{Aux-readAllowed4}{
		\trg{\OB{s}}\vdash\acp{unprotected}(\trg{a})
		&
		\exists \trg{s}\in\trg{\OB{s}}. \trg{s}\vdash\acp{entryPoint}(\trg{a'})
	}{
		\trg{\OB{s}}\vdash\acp{readAllowed}(\trg{a},\trg{a'})
	}{aux-ra4}\quad
	\typerule{Aux-methodEntryPoint}{
		\exists \trg{s}\in\trg{\OB{s}}. \trg{s}\vdash\acp{methodEntryPoint}(\trg{a'})
	}{
		\trg{\OB{s}}\vdash\acp{methodEntryPoint}(\trg{a})
	}{aux-mepp}\quad
	\typerule{Aux-entryPoint}{
		\exists \trg{s}\in\trg{\OB{s}}. \trg{s}\vdash\acp{entryPoint}(\trg{a'})
	}{
		\trg{\OB{s}}\vdash\acp{entryPoint}(\trg{a})
	}{aux-epp}\quad
	\typerule{Aux-returnEntryPoint}{
		\exists \trg{s}\in\trg{\OB{s}}. \trg{s}\vdash\acp{returnEntryPoint}(\trg{a'})
	}{
		\trg{\OB{s}}\vdash\acp{returnEntryPoint}(\trg{a})
	}{aux-repp}\quad
	\typerule{Aux-protected}{
		\exists \trg{s}\in\trg{\OB{s}}. \trg{s}\vdash\acp{range}(\trg{a'})
	}{
		\trg{\OB{s}}\vdash\acp{protected}(\trg{a})
	}{aux-prot}\quad
	\typerule{Aux-currentModule}{
		\trg{s}\in\trg{\OB{s}}
		&
		\trg{s}=\trg{id},\trg{n_c},\trg{n_t}
		&
		\trg{s}\vdash\acp{range}(\trg{id},\trg{n})
	}{
		(\trg{id},\trg{n})\vdash\acp{currentModule}(\trg{\OB{s}},\trg{s})
	}{aux-cm}\quad
	\typerule{Aux-exitJump}{
		\exists \trg{s}\in\trg{\OB{s}}. \trg{s}\vdash\acp{codeRange}(\trg{a})
		&
		\trg{t}= \trg{\OB{EM}};\trg{\OB{EO}};\trg{\OB{RM}};\trg{\OB{RO}}
		&
		\src{m}:\src{t}(\src{\OB{t}})\to\src{t}\mapsto\trg{w},\trg{w'}\in{\trg{\OB{RM}}}
	}{
		\trg{\OB{s}},\trg{t}\vdash\acp{exitJump}(\trg{a},\trg{w,w'})
	}{aux-ej}\quad
	\typerule{Aux-stuck}{
		\trg{m}(\trg{p})= \trg{w}
		&
		\trg{w}\notin\trg{\mc{I}}
	}{
		\vdash\acp{stuck}(\trg{p'},\trg{r'},\trg{f'},\trg{m'},\trg{\OB{s}},\trg{h}))
	}{aux-stuck}\quad
	\typerule{Aux-validJump-traces}{
		\trg{\OB{s}}\vdash\acp{validJump}(\trg{p},\trg{a})	
	}{
		\trg{\OB{s}},\trg{t}\vdash\acp{validJump}(\trg{p},\trg{a})
	}{aux-vjt}\quad
	\typerule{Aux-validJump-traces2}{
		\trg{\OB{s}},\trg{t}\vdash\acp{exitJump}(\trg{p},\trg{a})
	}{
		\trg{\OB{s}},\trg{t}\vdash\acp{validJump}(\trg{p},\trg{a})
	}{aux-vjt2}\quad
	\typerule{Aux-forwardReturn-EP}{
		\trg{p}\equiv(\trg{1},3*\trg{\mc{N}_w})
	}{
		\trg{\OB{s}},\trg{t}\vdash\acp{forwardReturnEP}(\trg{p})
	}{aux-vjt2}\quad
\end{center}
An exit jump is only used for the trace semantics, so where states span just some modules and the descriptors describe just their layout.
An exit jump is therefore a jump to an address that is not in the domain covered by the modules.
This could be completely another module or the unprotected code.

The rules presented in the paper extend the operational semantics of \aim.
Define a module state as follows: $\trg{\Upsilon}::=(\trg{p},\trg{r},\trg{f},\trg{m},\trg{{s}},\trg{h})$ and a program state as follows: $\trg{\Omega}::=(\trg{p},\trg{r},\trg{f},\trg{m},\trg{\OB{s}},\trg{h})$.
Let $\trg{\Pi}$ denote the set of all symbolic nonces $\trg{\pi}$.
Let $\equiv$ denote syntactic equivalence; equating a number and a symbol therefore results in false.
\begin{center}
\small
	\typerule{Eval-module-common}{
		(\trg{p},\trg{r},\trg{f},\trg{m},\trg{s},\trg{h})\toid(\trg{p'},\trg{r'},\trg{f'},\trg{m'},\trg{s},\trg{h'})
		\\
		\trg{s}=(\trg{id},\trg{n_c},\trg{n_t})
		&
		\trg{p}=(\trg{id},\trg{n})
		&
		\trg{s}\vdash\acp{internalJump}(\trg{p},\trg{p'})
	}{
		(\trg{p},\trg{r},\trg{f},\trg{m},\trg{s},\trg{h})\tolid(\trg{p'},\trg{r'},\trg{f'},\trg{m'},\trg{s},\trg{h'})
	}{evalutops}\quad
	\typerule{Eval-movi}{
		\trg{p}=(\trg{id},\trg{n})
		&
		\trg{m}(\trg{p})=  (\asm{movi\ r_d\ i})
		&
		\trg{r'}=\trg{r}[\asm{r_d}\mapsto\trg{i}]
	}{
		(\trg{p},\trg{r},\trg{f},\trg{m},\trg{s},\trg{h})\toid((\trg{id},\trg{n+1}),\trg{r'},\trg{f},\trg{m},\trg{s},\trg{h})
	}{evalmovi}\quad
	\typerule{Eval-add}{
		\trg{p}=(\trg{id},\trg{n})
		&
		\trg{m}(\trg{p})=  (\asm{add\ r_d\ r_s})
		\\
		v_d=\trg{r}(\asm{r_d}), \text{ if } \trg{r}(\asm{r_d})\in\trg{\Pi} \text{ then } v_d=0
		\\
		v_s=\trg{r}(\asm{r_s}), \text{ if } \trg{r}(\asm{r_s})\in\trg{\Pi} \text{ then } v_s=0
		\\
		v=v_d+v_s
		&
	 	\trg{r'}=\trg{r}[\asm{r_d}\mapsto v]
	 	&
		\trg{f'}=\trg{f}[\trg{\ms{ZF}}\mapsto(v==0)]
	}{
		(\trg{p},\trg{r},\trg{f},\trg{m},\trg{s},\trg{h})\toid((\trg{id},\trg{n+1}),\trg{r'},\trg{f'},\trg{m},\trg{s},\trg{h})
	}{evaladd}\quad
	\typerule{Eval-sub}{
		\trg{p}=(\trg{id},\trg{n})
		&
		\trg{m}(\trg{p})=  (\asm{sub\ r_d\ r_s})
		&
		v_d=\trg{r}(\asm{r_d}), \text{ if } \trg{r}(\asm{r_d})\in\trg{\Pi} \text{ then } v_d=0
		\\
		v_s=\trg{r}(\asm{r_s}), \text{ if } \trg{r}(\asm{r_s})\in\trg{\Pi} \text{ then } v_s=0
		&
		v=v_d-v_s
		&
	 	\trg{r'}=\trg{r}[\asm{r_d}\mapsto \fun{abs}{v}]
	 	\\
		\trg{f'}=\trg{f}[\trg{\ms{ZF}}\mapsto(v==0);\trg{\ms{SF}}\mapsto(v<0)]
	}{
		(\trg{p},\trg{r},\trg{f},\trg{m},\trg{s},\trg{h})\toid((\trg{id},\trg{n+1}),\trg{r'},\trg{f'},\trg{m},\trg{s},\trg{h})
	}{evalsub}\quad
	\typerule{Eval-cmp}{
		\trg{p}=(\trg{id},\trg{n})
		&
		\trg{m}(\trg{p})=  (\asm{cmp\ r_d\ r_s})
	 	&
		\trg{f'}=\trg{f}[\trg{\ms{ZF}}\mapsto(\trg{r}(\asm{r_d}) \equiv \trg{r}(\asm{r_s}))]
	}{
		(\trg{p},\trg{r},\trg{f},\trg{m},\trg{s},\trg{h})\toid((\trg{id},\trg{n+1}),\trg{r},\trg{f'},\trg{m},\trg{s},\trg{h})
	}{evalcmp}\quad
	\typerule{Eval-je-true}{
		\trg{p}=(\trg{id},\trg{n})
		&
		\trg{m}(\trg{p})= (\asm{je\ r_d\ f_i})
		&
		\trg{n'}= \trg{r}(\asm{r_d})
		&
		\trg{f}(\asm{f_i})=\trg{1}
	}{
		(\trg{p},\trg{r},\trg{f},\trg{m},\trg{s},\trg{h})\toid((\trg{id},\trg{n'}),\trg{r},\trg{f},\trg{m},\trg{s},\trg{h})
	}{evaljet}\quad
	\typerule{Eval-je-false}{
		\trg{p}=(\trg{id},\trg{n})
		&
		\trg{m}(\trg{p})= (\asm{je\ r_d\ f_i})
		&
		\trg{f}(\asm{f_i})=\trg{0}
	}{
		(\trg{p},\trg{r},\trg{f},\trg{m},\trg{s},\trg{h})\toid((\trg{id},\trg{n+1}),\trg{r},\trg{f},\trg{m},\trg{s},\trg{h})
	}{evaljef}\quad
	\typerule{Eval-new}{
		\trg{p}=(\trg{id},\trg{n})
		&
		\trg{m}(\trg{p}) = \asm{new\ r_d}
		&
		\trg{r'} = \trg{r}[\asm{r_d}\mapsto\trg{\pi}]
		&
		\trg{\pi}\notin\trg{h}
	}{
		(\trg{p},\trg{r},\trg{f},\trg{m},\trg{{s}},\trg{h})\tolid((\trg{id},\trg{n+1}),\trg{r'},\trg{f},\trg{m},\trg{{s}},\trg{h};\trg{\pi})
	}{evalnew}\quad
	\typerule{Eval-zero}{
		\trg{p}=(\trg{id},\trg{n})
		&
		\trg{m}(\trg{p})= (\asm{zero})
		&
		\forall \trg{i} \in \mb{N}
		&
		\trg{r'}=\trg{r}[\asm{r_i}\mapsto\trg{0}]
	}{
		(\trg{p},\trg{r},\trg{f},\trg{m},\trg{s},\trg{h})\toid((\trg{id},\trg{n+1}),\trg{r'},\trg{f},\trg{m},\trg{s},\trg{h})
	}{evalzero}\quad
	\typerule{Eval-halt}{
		\trg{p}=(\trg{id},\trg{n})
		&
		\trg{m}(\trg{p})= (\asm{halt})
	}{
		(\trg{p},\trg{r},\trg{f},\trg{m},\trg{s},\trg{h})\toid((\trg{0},\trg{-1}),\trg{r},\trg{f},\trg{m},\trg{s},\trg{h})
	}{evalhalt}\quad
	\typerule{Eval-single-to-multiple-modules}{
		(\trg{p},\trg{r},\trg{f},\trg{m},\trg{s},\trg{h}) \tolid (\trg{p'},\trg{r'},\trg{f'},\trg{m'},\trg{s},\trg{h'})
		&
		\trg{p}\vdash\acp{currentModule}(\trg{\OB{s}},\trg{s})
		&
		\trg{p}=(\trg{id},\trg{n})
	}{
		(\trg{p},\trg{r},\trg{f},\trg{m},\trg{\OB{s}},\trg{h}) \tol (\trg{p'},\trg{r'},\trg{f'},\trg{m'},\trg{\OB{s}},\trg{h'})
	}{evalutop} 
	\typerule{Eval-movl}{
		\trg{p}=(\trg{id},\trg{n})
		&
		\trg{m}(\trg{p})=  (\asm{movl\ r_d\ r_s\ r_i})
		\\
		\trg{\OB{s}}\vdash\acp{readAllowed}(\trg{n}, (\trg{r}(\asm{r_s}),\trg{r}(\asm{r_i})) )
		&
		\trg{r'}=\trg{r}[\asm{r_d}\mapsto\trg{m}(\trg{r}(\asm{r_s}),\trg{r}(\asm{r_i}))]
	}{
		(\trg{p},\trg{r},\trg{f},\trg{m},\trg{\OB{s}},\trg{h}) \tol ((\trg{id},\trg{n+1}),\trg{r'},\trg{f},\trg{m},\trg{\OB{s}},\trg{h})
	}{evalmovlc}\quad
	\typerule{Eval-movs}{
		\trg{p}=(\trg{id},\trg{n})
		&
		\trg{m}(\trg{p})=  (\asm{movs\ r_d\ r_s\ r_i})
		\\
		\trg{\OB{s}}\vdash\acp{writeAllowed}(\trg{n}, (\trg{r}(\asm{r_d}),\trg{r}(\asm{r_i})))
		&
		\trg{m'}=\trg{m}[(\trg{r}(\asm{r_d}),\trg{r}(\asm{r_i}))\mapsto\trg{r}(\asm{r_s})]
	}{
		(\trg{p},\trg{r},\trg{f},\trg{m},\trg{\OB{s}},\trg{h}) \tol ((\trg{id},\trg{n+1}),\trg{r},\trg{f},\trg{m'},\trg{\OB{s}},\trg{h})
	}{evalmovsc}\quad
	\typerule{Eval-jmp}{
		\trg{p}=(\trg{id},\trg{n})
		&
		\trg{m}(\trg{p})= (\asm{jmp\ r_d\ r_i})
		&
		\trg{n'}= \trg{r}(\asm{r_d})
		&
		\trg{id'}= \trg{r}(\asm{r_i})
		&
		\trg{r'}=\trg{r}[\asm{r_0} \mapsto \trg{id}]
		\\
		\trg{\OB{s}}\vdash\acp{validJump}(\trg{p},(\trg{id'},\trg{n'}))
	}{
		(\trg{p},\trg{r},\trg{f},\trg{m},\trg{\OB{s}},\trg{h})\tol((\trg{id'},\trg{n'}),\trg{r'},\trg{f},\trg{m},\trg{\OB{s}},\trg{h})
	}{evaljmp}\quad
\end{center}

%% file: bodies2/trsem.tex

\subsubsection{Trace Equivalence for Securely-Compiled Components in \aim}\label{sec:aimtrsem}
This section defines a trace equivalence to reason about securely-compiled components in \aim; it is inspired by the trace semantics for the single-module version of \aim~\cite{llfatr-j}.

Trace equivalence is a simpler tool to reason about equivalence of components than contextual equivalence~\cite{llfatr-j,scoo-j,abadiLayout,Jagadeesan,javaJr}.
Trace equivalence relates two components that exhibit the same trace semantics.
Trace semantics describes the behaviour of a component as a set of traces.
Traces are sequences of actions \trg{\OB{\alpha}} with the following syntax:
\begin{align*}
&\mi{labels}& \trg{\lambda} ::=&\ \trg{\alpha}\mid\ \trg{\tau}
	&
&\mi{actions}& \trg{\alpha}\ ::=&\ \trg{\gamma!} \mid\ \trg{\gamma?} \mid\  \trg{\surd}
\end{align*}\vspace{-2.3em}
\begin{align*}
&\mi{observables}&	\trg{\gamma} ::=&\ \clgen{\ \trg{a}\ \trg{\OB{w}}} \mid\ \rtgen{\ \trg{a}\ \trg{w},\trg{id}}
\end{align*}
Labels $\trg{\lambda}$ generated by the trace semantics can be actions $\trg{\alpha}$ or the silent action $\trg{\tau}$.
The silent action $\trg{\tau}$ is generated by unobservable transitions, thus these labels are not accumulated in traces.
Actions $\trg{\alpha}$ can be an observable, decorated action $\trg{\gamma}$ or a tick $\trg{\surd}$, which indicates termination of the computation.
Observable actions $\trg{\gamma}$ are either a function call to an address $\trg{a}$ with parameters \trg{\OB{w}}, or a return of value $\trg{w}$ to address $\trg{a}$ from the module with id \trg{id}.
Decorations $\trg{?}$ or $\trg{!}$ indicate whether unprotected code performs the action ($\trg{?}$) or receives it ($\trg{!}$).

Define a (compiled) component state with $\trg{\Theta}$, it indicates either that the execution is within a component or that the execution is in an unknown location.
Formally: $\trg{\Theta}::= (\trg{p},\trg{r},\trg{f},\trg{m},\trg{\OB{s}},\trg{h},\trg{t})\mid\blk$.
$\trg{\Theta_0}(\trg{P},\trg{h})$ denotes the initial state of a component $\trg{P}=(\trg{m},\trg{\OB{s}},\trg{t})$ with a nonce oracle $\trg{h}$.
Formally: $\trg{\Theta_0}(\trg{P},\trg{h})=\blk$.
Relation $\xtol{\trg{{\lambda}}}\ \subseteq\trg{\Theta}\times\trg{{\lambda}}\times\trg{\Theta}$ captures how single labels are generated.
The reflexive-transitive closure of $\xtol{\trg{{\lambda}}}$ accumulates labels in traces and filters out silent traces; it is captured by relation $\Xtol{\trg{\OB{\alpha}}}\ \subseteq\trg{\Theta}\times\trg{\OB{\alpha}}\times\trg{\Theta}$. 
The auxiliary functions adopted for the trace semantics are used to tell if an address is an entry point (\acp{methodEntryPoint}), or if it is the entry point related to the \fun{forwardReturn}{\cdot} procedure in \sys (\acp{forwardReturnEP}).
Function \acp{exitJump} detects jumps outside the memory space of a component; function \acp{stuck} tells if a state is stuck.
\begin{figure}
\centering
\small
	\typerule{Trace-call}{
		\trg{\OB{w}} = \trg{w_0},\cdots,\trg{w_{6+k}}
		&
		\trg{r}=[\asm{r_0}\mapsto\trg{w_0},\cdots,\asm{r_{6+k}}\mapsto\trg{w_{6+k}}]
		\\
		\trg{\OB{s}}\vdash\acp{methodEntryPoint}(\trg{p})
		&
		\src{m}:\src{t}(\src{t_1\cdots,t_k})\to\src{t}\mapsto\trg{p} \in \trg{\OB{EM}}
		\\
		\trg{t} = \trg{\OB{EM}};\trg{\OB{EO}};\trg{\OB{RM}};\trg{\OB{RO}}
	}{
		\blk\xtol{\cl{\trg{p}\ \trg{\OB{w}}}}(\trg{p},\trg{r},\trg{f},\trg{m},\trg{\OB{s}},\trg{h},\trg{t})
	}{trace-call}\quad 
	\typerule{Trace-outcall}{
		\trg{m}(\trg{p})= (\asm{jmp\ r_d\ r_i})
		&
		\trg{p'}=(\trg{\iota},\trg{\sigma})
		&
		\trg{r}(\asm{r_d})\mapsto\trg{\sigma}
		&
		\trg{r}(\asm{r_i})\mapsto\trg{\iota}
		\\
		\trg{\OB{w}} = \trg{w_0},\cdots,\trg{w_{6+k}}
		&
		\trg{r}=[\asm{r_0}\mapsto\trg{w_0},\cdots,\asm{r_{6+k}}\mapsto\trg{w_{6+k}}]
		\\
		\trg{\OB{s}},\trg{t}\vdash\acp{exitJump}(\trg{p},\trg{p'})
		&
		\trg{r}(\asm{r_5})=\trg{3*\mc{N}_w}
		\\
		\src{m}:\src{t}(\src{t_1\cdots,t_k})\to\src{t}\mapsto\trg{\iota};\trg{\sigma} \in \trg{\OB{RM}}
		&
		\trg{t} = \trg{\OB{EM}};\trg{\OB{EO}};\trg{\OB{RM}};\trg{\OB{RO}}
	}{
		(\trg{p},\trg{r},\trg{f},\trg{m},\trg{\OB{s}},\trg{h},\trg{t})\xtol{\cb{\trg{p'}\ \trg{\OB{w}}}}\blk
	}{trace-outcall}\quad
	\typerule{Trace-returnback}{
		&
		\trg{r}=[\asm{r_0}\mapsto\trg{id'},\asm{r_6}\mapsto\trg{w}]
		&
		\trg{\OB{s}}\vdash\acp{forwardReturnEP}(\trg{p})
	}{
		\blk\xtol{\rb{\trg{p}\ \trg{w},\trg{id'}}}(\trg{p},\trg{r},\trg{f},\trg{m},\trg{\OB{s}},\trg{h},\trg{t})
	}{trace-retback}\quad
	\typerule{Eval-return}{
		\trg{m}(\trg{p})= (\asm{jmp\ r_d\ r_i})
		&
		\trg{p}=(\trg{id},\trg{n})
		&
		\trg{p'}=(\trg{id'},\trg{n'})
		\\
		\trg{r}(\asm{r_d})\mapsto\trg{n'}
		&
		\trg{r}(\asm{r_i})\mapsto\trg{id'}
		&
		\trg{r}=[\asm{r_6}\mapsto\trg{w}]
		\\
		\trg{\OB{s}},\trg{t}\vdash\acp{exitJump}(\trg{p},\trg{p'})
		&
		\trg{r}(\asm{r_5})\mapsto\trg{0}
	}{
		(\trg{p},\trg{r},\trg{f},\trg{m},\trg{\OB{s}},\trg{h},\trg{t})\xtol{\rt{\trg{p'}\ \trg{w},\trg{id}}} (\unk,\trg{m},\trg{\OB{s}},\trg{h},\trg{t})
	}{trace-return}

	\typerule{Trace-tau}{
		\trg{\Theta}
		\tol
		\trg{\Theta'}
	}{
		\trg{\Theta}
		\xtol{\trg{\tau}}
		\trg{\Theta'}
	}{trace-tau}
	\typerule{Trace-refl}{
	}{
		\trg{\Theta}\Xtol{\epsilon}\trg{\Theta}
	}{trrefl}
	\typerule{Trace-tau}{
		\trg{\Theta}\xtol{\trg{\tau}}\trg{\Theta'}
	}{
		\trg{\Theta}\Xtol{\epsilon}\trg{\Theta'}
	}{trtau}
	\typerule{Trace-trans}{
		\trg{\Theta}\xtol{\trg{{\alpha}}}\trg{\Theta''} 
		\\
		\trg{\Theta''}\Xtol{\trg{\OB{\alpha}'}}\trg{\Theta'}
	}{
		\trg{\Theta}\Xtol{\trg{\OB{\alpha}}\cdot\trg{\OB{\alpha}'}}\trg{\Theta'}
	}{trtrans}
	\typerule{Trace-termination}{
		\trg{\Theta}\tol\trg{\Theta'}
		\\
		\vdash\acp{stuck}(\trg{\Theta'})
	}{
		\trg{\Theta}
		\xtol{\trg{\surd}}
		\trg{\Theta'}
	}{trace-tick}
\end{figure}

Define the trace semantics of a component $\trg{P}$ as: $\taim{\trg{P}}=\{\trg{\OB{\alpha}} \mid \forall \trg{h}.\exists\trg{\Theta}.$ $\trg{\Theta_0}(\trg{P},\trg{h})\Xtol{\trg{\OB{\alpha}}}\trg{\Theta}\}$.
Two modules $\trg{P_1}$ and $\trg{P_2}$ are trace equivalent, denoted with $\trg{P_1}\teqaim\trg{P_2}$ if their trace semantics coincides.
\begin{definition}[Trace equivalence]\label{def:treqaim}
$\trg{P_1}\teqaim\trg{P_2}\triangleq$ $\taim{\trg{P_1}}=\taim{\trg{P_2}}$.
\end{definition}
For securely-compiled components, trace semantics coincides with contextual equivalence (\Cref{thm:fatracesaim}).

\begin{proofsketch}
This proof is analogous to the one for the single-module version of \aim~\cite{llfatr-j}, we just give an informal argument why it holds, though the proof of~\cite{llfatr-j} can be easily adapted to scale for \aim.
Informally, we need to prove that labels capture all the information that is communicated between a securely-compiled component and external code.

\begin{description}
\item[\trg{!}-decorated actions].
The \rt{\cdots} action has a standard structure for when it is valid: \rt{\trg{a}\ \trg{w},\trg{1}}, where \trg{a} has been received on \asm{r_0} and \asm{r_5} beforehand by \sys.

Also \cb{\cdots} has a standard structure: \cb{\trg{id,n}\ (\trg{1},\trg{n},\trg{0},\trg{0},\trg{0},3*\trg{\mc{N}_w},\OB{\trg{w}})}, where \trg{id,n} is in the export tables of the modules.

\Cref{ass:comp-corr} about \comp{\cdot} code ensures that all communication coming from a compiled component happen only via entry points.

Specifically, this communication is regulated by the code added via the \prot{\cdot} function.

In case of a \cb{\cdot} only the registers used to carry a parameter (including the current object) are passed because \fun{resetFlags}{\cdot}, \fun{resetRegisters}{\cdot} and \fun{resetRegistersExcept}{\cdot} ensure that other registers are always \trg{0}.

In case of a \rt{\cdot} \fun{resetFlags}{\cdot}, \fun{resetRegisters}{\cdot} and \fun{resetRegistersExcept}{\cdot} ensure that only the module \trg{id} and the returned value (in \asm{r_6}) are not always \trg{0}.

Some of the data communicated via registers could be omitted since, as stated, it is always the same so it conveys no information.
We keep this data as we chose for the same structure of labels for incoming and outgoing actions.

\item[ \trg{?}-decorated actions].
Nor the code inserted by \prot{\cdot} nor the one generated by \comp{\cdot} uses flags as set by external code, so flags are never used to convey information.

In case of a \cl{\cdot}, due to the assumption on the correctness of \comp{\cdot} and on the definition of \prot{\cdot} only the registers used to carry parameters are used for computation.

In case of a \rb{\cdot}, the only registers used for calculation by \prot{\cdot} are \asm{r_0} and \asm{r_6}, which are the only ones captured by the trace semantics.
\qedhere
\end{description}
\end{proofsketch}

%% file: bodies2/algoinform.tex
\section{Algorithm Formalisation}\label{sec:algoaim}

\subsection{Trace Back-Translation Algorithm \algo{\cdot}}\label{sec:algo}
Informally, the algorithm is used to build a \jem context that can distinguish two \jem components whose compiled counterparts are trace inequivalent.

The algorithm inputs two distinct \aim traces \trg{\OB{\alpha_1}} and \trg{\OB{\alpha_2}} and the two \jem components that generate them \src{\mc{C}_1} and \src{\mc{C}_2}.
In other words $\trg{\OB{\alpha_1}}\in\taim{\compaim{\src{\mc{C}_1}}}$ and $\trg{\OB{\alpha_2}}\in\taim{\compaim{\src{\mc{C}_2}}}$.
The algorithm outputs a single \jem component \src{\mc{C}} that is the context that can differentiate between \src{\mc{C}_1} and \src{\mc{C}_2}.
Formally: $\algo{\src{\mc{C}_1},\src{\mc{C}_2},\trg{\OB{\alpha_1}},\trg{\OB{\alpha_2}}}=\src{\mb{C}}$.
The two traces are the same up to the last action, therefore $\trg{\OB{\alpha_1}}\equiv\trg{\OB{\alpha}\alpha_1!}$ and $\trg{\OB{\alpha_2}}\equiv\trg{\OB{\alpha}\alpha_2!}$ where $\trg{\alpha_1!}\neq\trg{\alpha_2!}$.
Trace \trg{\OB{\alpha}} is called the \emph{common prefix} while the two actions $\trg{\alpha_1!}$ and $\trg{\alpha_2!}$ are the \emph{different actions}, which appear at index $i$.

Intuitively, the algorithm produces code \src{\mb{C}} that replicates all ?-decorated actions in the common prefix \trg{\OB{\alpha}}.
The !-decorated actions are made by the component that fills the hole of \src{\mb{C}} and the algorithm adds code to \src{\mc{C}} to update its internal state.
Then, \src{\mb{C}} performs the differentiation by terminating in case it detects \src{\mc{C}_1} and diverging in the other.

The algorithm is divided in three sub-routines: building the skeleton (\skel{\src{\mc{C}_1},\src{\mc{C}_2}}), emulating the common prefix (\emul{\trg{\OB{\alpha}},\trg{t}}) and distinguishing the different actions (\diff{\trg{\alpha_1!},\trg{\alpha_2!},i}).

\paragraph{$\skel{\src{\mc{C}_1},\src{\mc{C}_2}}$}
implements all classes and objects that \src{\mc{C}_1} and \src{\mc{C}_2} specify in their import declarations.
Additionally, it creates helper functions and objects, e.g., it has tables where all globally-known objects are stored and it has a variable to keep track of the action being emulated.

\paragraph{$\emul{\trg{\OB{\alpha}},\trg{t}}$}
returns method bodies that fill the classes created by \skel{\cdot} based on the different encountered actions.
Intuitively, this algorithm subroutine creates a \jem context that emulates all ?-actions in the traces.

\begin{itemize}
	\item$\cl{\trg{a}\ \trg{\OB{w}}}$
	In this case the context must call method \src{m} compiled at address \trg{a}; \trg{t} tells which \aim addresses correspond to which \jem methods.
	Based on the (known) signature of \src{m}, the parameters \trg{\OB{w}} are emulated to their \jem counterpart.
	
	Primitively-typed parameters are emulated to their \jem counterpart.
	For example when a \src{\lstb{Bool}} is expected and \comp{\src{\lst{true}}} is received, the parameter is emulated by \src{\lst{true}}.
	Object-typed parameters are stored in globally-accessible tables based on the encoding of their type and of the id they have in \aim.
	These tables have getters and setters to retrieve objects based on their \aim id.
	
	\item$\rb{\trg{a}\ \trg{w}}$
	In this case the produced code must return the value emulating \trg{w}, as discussed in the call case.

	\item$\cb{\trg{a}\ \trg{\OB{w}}}$ and $\rt{\trg{a}\ \trg{w}}$
	In these cases, the internal state of \src{\mb{C}} is updated.
	For example, \src{\mb{C}} keeps track of the index of the action that is emulating and of all allocated objects.
	In these cases, the index is incremented by 1 and all newly allocated objects received via \trg{\OB{w}} or \trg{w} are added to the tables.
\end{itemize}

There are cases in which a ?-decorated action cannot be emulated, e.g., a return when no method call was made or a call to a method that expects a \src{\lstb{Unit}}-typed parameter and instead receives \comp{\src{\lst{true}}}.
These are actions that try to violate \jem abstractions; in this cases the algorithm must \fail\ and generate a differentiating component that does nothing (\Cref{thm:maintheorem}).
\begin{lemma}[Termination is emulation failure]\label{thm:maintheorem}
$\forall\src{\mc{C}},\comp{\src{\mc{C}}}=(\trg{m_c},\trg{m_d},\trg{t}). \trg{\OB{\alpha}\surd}\in\taim{\compaim{\src{\mc{C}}}} \iff \emul{\trg{\OB{\alpha}},\trg{t}}=\fail$.
\end{lemma}
\Cref{thm:maintheorem} is the most important result of this paper, as proving it reveals which issues arise in a secure compiler.
This \namecref{thm:maintheorem} states that any traces that causes a termination tick cannot be replicated at the source level (i.e., the algorithm \fail s).
An action that cannot be replicated in \jem is an action that tries to violate the abstractions of \jem.
Termination ticks are only generated by the execution of \asm{halt}; only the checks inserted by \prot{\cdot} or by \sys insert that instruction.
Traces with a \trg{\surd} are traces that trigger a check; so they are violating a \jem abstraction in compiled code.
If this is true for all traces with a \trg{\surd} (the $\Rightarrow$ direction), and also all traces that incur violations generate a \trg{\surd} (the $\Leftarrow$ direction) then the algorithm is correct.
Most importantly, this proof indicates where to place the checks at entry points and what to check.
The fact that it is not possible to emulate something (e.g., a \src{\lstb{Bool}}-typed parameter that is not \comp{\src{\lst{true}}} nor \comp{\src{\lst{false}}}) indicates what checks to add (in this case, that \src{\lstb{Bool}}-typed parameters must be either \comp{\src{\lst{true}}} or \comp{\src{\lst{false}}}, i.e., function \fun{dynamicTypechecks}{\cdot}).

\paragraph{$\diff{\trg{\alpha_1!},\trg{\alpha_2!},i}$} 
returns two different code fragments that will detect whether the produced context is interacting with \src{\mc{C}_1} or \src{\mc{C}_2}.
For example, consider the two actions to be $\rt{\trg{a}\ \comp{\src{\lst{true}}}}$ and $\rt{\trg{a}\ \comp{\src{\lst{false}}}}$.
After the last emulated action, the context needs to check the returned value and terminate in case it is \src{\lst{true}} and diverge in the other.

All cases when two !-decorated actions are different are considered in \diff{\cdot}; they are omitted for space reasons.

%% file: bodies2/algoform.tex

\subsection{The Algorithm, Formally \algo{\cdot}}\label{sec:algoform}
\skel{\src{\mc{C}_1},\src{\mc{C}_2}} returns a component \src{\mc{C}}.

\emul{\trg{\OB{\alpha}},\trg{t}} returns a list of expressions bound to the method where they must be added.

\diff{\trg{\alpha_1},\trg{\alpha_2},i,M,V, \OB{\src{m}}, \OB{\src{t}}} returns a pair of expressions bound to the methods where they must be added.

Define code additions \src{A} as follows: $\src{A} ::= \src{E}@\src{m} \mid \src{E}@\src{m}@i$.

$\src{\mc{C}} + {\src{E}@\src{m}}$ returns a new component where the body of method \src{m} has been extended with \src{E}.

$\src{\mc{C}} + {\src{E}@\src{m}@i}$ is analogous but the extension \src{E} is placed in the body of method \src{m} inside an \lst{if} statement whose guard is \lst{oc.isStep( i )}.
The necessity of this will be clear after having seen the rest of the functions.

The algorithm can be defined as follows:
\begin{center}
\small
\typerule{\algo{\cdot}}{
	\compaim{\src{\mc{C}_1}} = (\trg{m_1},\trg{\OB{s}},\trg{t})
	&
	\compaim{\src{\mc{C}_2}} = (\trg{m_2},\trg{\OB{s}},\trg{t})
	\\
	\skel{\src{\mc{C}_1},\src{\mc{C}_2}} = \src{\mc{C}}
	\\
	\emul{\trg{\OB{\alpha}},\trg{t}} = \OB{\src{A}},M,V, \OB{\src{m}}, \OB{\src{t}},i 
	\\
	\diff{\trg{\alpha_1},\trg{\alpha_2},i,M,V, \OB{\src{m}}, \OB{\src{t}}} = \src{A_1},\src{A_2} 
}{
	\algo{\src{\mc{C}_1},\src{\mc{C}_2},\trg{\OB{\alpha}\alpha_1},\trg{\OB{\alpha}\alpha_2}} = \src{\mc{C}} +\OB{\src{A}} + \src{A_1} + \src{A_2}
}{algo}
\end{center}

\subsection{\skel{\src{\mc{C}_1},\src{\mc{C}_2}}}\label{sec:skel}
\begin{center}
\small
\typerule{Collect Interfaces - base}{
	\src{\OB{M}} = \lstb{public}\ \src{m_1}(\src{\OB{x_1}}):\src{{M_t}_1}\ \{\lstb{return}\ \src{E_1};\},\cdots,\lstb{public}\ \src{m_k}(\src{\OB{x_k}}):\src{{M_t}_k}\ \{\lstb{return}\ \src{E_k};\}
	\\
	\src{I} = \lstb{class-decl}\ \src{c} \{\src{{M_t}_1},\cdots,\src{{M_t}_k}\}
}{
	\fun{interfaces}{\lstb{import}\ \src{\OB{I}};\src{\OB{X}};\ \lstb{class}\ \src{c} \{\src{K}\ \src{\OB{F_t}}\ \src{\OB{M}}\}; \src{\OB{O}}} = \src{I}
}{collectinterfacesbase}\quad
\typerule{Collect Interfaces - inductive}{
	\forall i \in 1 .. k 
	&
	\fun{interfaces}{\src{C_i}}= \src{I_i}
}{
	\fun{interfaces}{\src{C_1,\cdots,C_k}} = \src{I_1,\cdots,I_k}
}{collectinterfacesind}\quad
\typerule{Collect externs - base}{
	\src{\OB{O}} = \lstb{object}\ \src{o_1}:\src{c} \{\src{\OB{F}_1}\},\cdots,\lstb{object}\ \src{o_k}:\src{c} \{\src{\OB{F}_k}\}
	\\
	\src{\OB{X}'} = \lstb{obj-decl}\ \src{o_1}:\src{c},\cdots,\lstb{obj-decl}\ \src{o_k}:\src{c};
}{
	\fun{externs}{\lstb{import}\ \src{\OB{I}};\src{\OB{X}};\ \lstb{class}\ \src{c} \{\src{K}\ \src{\OB{F_t}}\ \src{\OB{M}}\}; \src{\OB{O}}} = \src{\OB{X}'}
}{collectexternsbase}\quad
\typerule{Collect externs - inductive}{
	\forall i \in 1 .. k 
	&
	\fun{externs}{\src{C_i}} = \src{X_i}
}{
	\fun{externs}{\src{C_1,\cdots,C_k}} = \src{\OB{X_1}},\cdots,\src{\OB{X_k}}
}{collectexternsind}\quad
\typerule{stub-class}{
	\src{C} = \lstb{import}\ \src{\OB{I}};\src{\OB{X}};\ \lstb{class}\ \src{i} \{\src{K}\ \src{\emptyset}\ \src{\OB{M}}\}; \src{\OB{O}};\src{O}
	&
	\src{K} = \src{i}(\src{\emptyset})\ \{ \}
	\\
	\fun{staticFor}{\src{i}} = \src{O} 
	&
	\src{i}\mapsto\src{\OB{O}} \in \OB{\src{i}\mapsto\src{\OB{O}}}
	\\
	\src{M_t} = \src{{M_t}_1},\cdots,\src{{M_t}_k}
	&
	\forall i \in 1..k.
	&
	\fun{stub-method}{\src{{M_t}_i}} = \src{M_i} 
	\\
	\src{\OB{M}}=\src{M_1},\cdots,\src{M_k}
}{
	\fun{stub-class}{\lstb{class-decl}\ \src{i} \{\src{\OB{M_t}}\}, \src{\OB{I}};\src{\OB{X}},\OB{\src{i}\mapsto\src{\OB{O}}}} = \src{C}
}{stub-class}\quad
\typerule{stub-class inductive}{
	\forall i \in 1 .. k
	&
	\fun{stub-class}{\src{I_i}} = \src{C_i}
}{
	\fun{stub-class}{\src{I_1},\cdots,\src{I_k}, \src{\OB{I}};\src{\OB{X}},\OB{\src{i}\mapsto\src{\OB{O}}}} = \src{C_1},\cdots,\src{C_k}
}{stubclassind}\quad
\typerule{stub-obj}{
	\lstb{object}\ \src{o}:\src{t} \{\src{\emptyset}\}	= \src{O} 
}{
	\fun{stub-obj}{ \lstb{obj-decl}\ \src{o}:\src{t};	} = \src{c}\mapsto\src{O}
}{stubobj}\quad
\typerule{stub-obj ind}{
	\forall i \in 1 .. k
	&
	\fun{stub-obj}{\src{X_i}} = \src{c_i}\mapsto\src{O_i}
}{
	\fun{stub-obj}{\src{X_1},\cdots,\src{X_k}} = \src{c_1}\mapsto\src{O_1}\uplus\cdots\uplus\src{c_k}\mapsto\src{O_k}
}{stubobjind}\quad
\typerule{stub-method-unit}{
	\src{M} = \lstb{public}\ \src{m}(\src{\OB{x}}):\src{t}(\src{\OB{t}})\to\src{\lstb{Unit}}\ \{\lstb{return}\ \src{\lst{unit}};\}	
}{
	\fun{stub-method}{\src{m} : \src{t}(\src{\OB{t}})\to\src{\lstb{Unit}}} = \src{M}
}{stub-meth-unit}\quad
\typerule{stub-method-bool}{
	\src{M} = \lstb{public}\ \src{m}(\src{\OB{x}}):\src{t}(\src{\OB{t}})\to\src{\lstb{Bool}}\ \{\lstb{return}\ \src{\lst{true}};\}	
}{
	\fun{stub-method}{\src{m} : \src{t}(\src{\OB{t}})\to\src{\lstb{Bool}}} = \src{M}
}{stub-meth-bool}\quad
\typerule{stub-method-int}{
	\src{M} = \lstb{public}\ \src{m}(\src{\OB{x}}):\src{t}(\src{\OB{t}})\to\src{\lstb{Int}}\ \{\lstb{return}\ \src{0};\}	
}{
	\fun{stub-method}{\src{m} : \src{t}(\src{\OB{t}})\to\src{\lstb{Int}}} = \src{M}
}{stub-meth-int}\quad
\typerule{stub-method-obj}{
	\src{M} = \lstb{public}\ \src{m}(\src{\OB{x}}):\src{t}(\src{\OB{t}})\to\src{c}\ \{\lstb{return}\ \src{\lst{null}};\}	
}{
	\fun{stub-method}{\src{m} : \src{t}(\src{\OB{t}})\to\src{c}} = \src{M}
}{stub-meth-obj}\quad
\typerule{staticfor}{}{
	\fun{staticFor}{\src{i}} = \lstb{object}\ \src{static-for-i}:\src{i} \{\src{\emptyset}\}
}{staticfor}\quad
\typerule{skel}{
	\fun{interfaces}{\src{C_1},\cdots,\src{C_k}} = \src{\OB{I}}
	&
	\fun{externs}{\src{C_1},\cdots,\src{C_k}} = \src{\OB{E}}
	\\
	\forall i \in 1 .. k.
	&
	\src{C_i} = \lstb{import}\ \src{\OB{I}_i};\src{\OB{X}_i};\ \lstb{class}\ \src{c_i} \{\src{K_i}\ \src{\OB{F_t}_i}\ \src{\OB{M}_i}\}; \src{\OB{O}_i}
	\\
	\src{\OB{I}'} = \src{\OB{I}_1},\cdots,\src{\OB{I}_k}
	&
	\src{\OB{X}'} = \src{\OB{X}_1},\cdots,\src{\OB{X}_k}
	&
	\fun{helpers}{} = \src{C}
	\\
	\fun{stub-obj}{\src{\OB{X}'}} = \OB{\src{c}\mapsto\src{\OB{O}}} 
	&
	\fun{stub-class}{\src{\OB{I}'},\src{\OB{I}},\src{\OB{X}},\OB{\src{c}\mapsto\src{\OB{O}}}} = \src{\mc{C}}
}{
	\skel{\src{C_1},\cdots,\src{C_k},\src{C_1'},\cdots,\src{C_k'}} = \src{\mc{C}};\src{C}
}{skel}\quad
\typerule{helpers}{
	\src{C} = \lstb{import}\ \src{\emptyset};\src{\emptyset};\ \lstb{class}\ \src{Helper} \{\src{K}\ \src{\OB{F_t}}\ \src{\OB{M}}\}; \src{\OB{O}}
	\\
	\fun{helper-fieldtypes}{}=\src{\OB{F_t}}
	\\
	\fun{helper-objects}{}=\src{\OB{O}}
	&
	\fun{helper-methods}{}=\src{\OB{M}}
	\\
	\fun{helper-constructor}{}=\src{K}
}{
	\fun{helpers}{} = \src{C}
}{hel}\quad
\typerule{helpers-fieldtypes}{
	\src{\OB{F_t}} = \src{step}:\src{\lstb{Int}}
}{
	\fun{helper-fieldtypes}{}=\src{\OB{F_t}}
}{hl-ft}\quad
\typerule{helpers-objects}{
	\lstb{object}\ \lst{oc}:\src{Helper} \{\src{step = 0}\}		
}{
	\fun{helper-objects}{}=\src{O}
}{hl-o}\quad
\typerule{helpers-methods}{
	\lstb{public}\ \src{isStep}(\src{x}):\src{M_t}\ \{\lstb{return}\ \lstb{if}\ (\lst{this.step} == \src{x})\ \{\src{\lst{true}}\}\ \lstb{else}\ \{\src{\lst{false}}\};\}	
	\\
	\src{M_t} =  \lst{isStep} : \src{Helper}(\src{\lstb{Int}})\to\src{\lstb{Bool}}
	\\
	\lstb{public}\ \lst{incStep}():\src{M_t'}\ \{\lstb{return}\ \lst{this.step + 1; unit};\}	
	\\
	\src{M_t'} =  \lst{incrStep} : \src{Helper}()\to\src{\lstb{Unit}}
	\\
	\lstb{public}\ \lst{diverge}():\src{M_t''}\ \{\lstb{return}\ \lst{this.diverge()};\}	
	\\
	\src{M_t''} =  \lst{diverge} : \src{Helper}()\to\src{\lstb{Unit}}
	\\
	\lstb{public}\ \lst{main}():\src{M_t'''}\ \{\lstb{return 0};\}	
	\\
	\src{M_t'''} =  \lst{main} : \src{Helper}()\to\src{\lstb{Int}}
}{
	\fun{helper-methods}{}=\src{\OB{M}}
}{hl-m}\quad
\typerule{helpers-constructor}{
	\src{K} = \src{Helper}(\src{step:\lstb{Int}})\ \{\lstb{this}.\src{step}=\src{0}\}
}{
	\fun{helper-constructor}{}=\src{K}
}{hl-k}\quad
\end{center}
Note that the two components input by \skel{} must have the same imports and they must define the same classes, otherwise they are trivial to differentiate.
Therefore, to build the skeleton in \Cref{tr:skel} we use only one component.
Assume the names \src{Helper} and \src{oc} are fresh (a simple substitution can ensure this).

For all class \src{c} and all interface \src{i} defined in \src{\mc{C}_1} and \src{\mc{C}_2}, the \src{Helper} class also contains a list.
Call this class \lst{listof-t} for a class or interface type \src{t}.
Each element of the list has a \src{t} field that points to an object of type \src{t} and an \src{\lstb{Int}} field that contains the \aim encoding of object's id.
The list implements method \lst{getByName( n )} that inputs a name \lst{n} and returns the object with that name in the list.
Additionally, each list has an \lst{append} method.
\src{Helper} also has a method \lst{createNew-t( n )} that creates a new object of type \lst{t} and adds it to the list \lst{listof-t} with name \lst{n}; this method returns the new object.
These lists are populated with all known static objects (i.e., the objects declared in \src{\mc{C}_1} and \src{\mc{C}_2}).

\src{Helper} also has generic method \lst{addObject( o, n )} to insert object \lst{o} with name \lst{n} will call \lstb{instanceof}\lst{(o, \src{c})} for all possible class and interface types implemented by \src{\mc{C}_1} and \src{\mc{C}_2}.
Based on the type \src{c}, the method calls the \lst{append} on list \lst{listof-}\src{c}.

\src{Helper} also has generic method \lst{getByName( n )} will call \lst{getByName( n )} on all lists until it finds the object to return.

\subsection{\emul{\trg{\OB{\alpha}},\trg{t}}}\label{sec:emul}
The \emul{\cdot} function \fail s any time any of its sub-parts \fail s.
If this happens, it returns an empty method body for the \lst{main} method.

Hoe to read the \emul{\cdot} judgment: $\emul{\trg{\alpha}, M_0,V,\OB{\trg{id}}, \OB{\src{m}}, \OB{\src{t}}, i,\trg{t}} = \src{E}@\src{m}, V', \OB{\trg{id}'}, i', \OB{\src{m}'}, \OB{\src{t'}}$
\emph{under method environment $M_0$, under knowledge of teh allocated objects $V$, with the stack of caller ids \OB{\trg{id}}, with the stack of called methods \OB{\src{m}} and the stack of return types \OB{\src{t}}, at step $i$, with bindings \trg{t}, action \trg{\alpha} produces code \src{E} to be placed inside method \trg{m}, it updates the knowledge of objects with $V'$, it increments the step to $i'$, it updates the id stack to \OB{\trg{id}'}, it updates the methods stack to \OB{\trg{m}'} and the expected return types to \OB{\trg{t}'}.}

The algorithm uses $M$ to store method bindings, so $M::=\OB{\trg{a}\mapsto\src{m}:\src{M_t}}$ and $V$ to store object bindings, so $V::=\OB{\trg{w}:\src{t}}$.
\begin{center}
\small
\typerule{methodKnowledge}{
	M = \trg{a_1}\mapsto\src{m_1}:\src{M_t}_1,\cdots,\trg{a_n}\mapsto\src{m_n}:\src{M_t}_n
}{
	\fun{methodKnowledge}{\src{m_1}:\src{M_t}_1\mapsto \trg{a_1},\cdots,\src{m_n}:\src{M_t}_n\mapsto \trg{a_n}} = M
}{mk}\quad
\typerule{objectsKnowledge}{
	V = \trg{w_1}:\src{t_1},\cdots,\trg{w_n}:\src{t_n}
}{
	\fun{objectsKnowledge}{\src{o_1}:\src{t_1}\mapsto \trg{w_1},\cdots,\src{o_n}:\src{t_n}\mapsto \trg{w_n}} = V
}{ok}\quad
\end{center}

\fun{emulate}{\cdot} uses a helper function \fun{nonce-to-int} that translates nonces to integers.

Following is the definition of \fun{emulate}{\cdot}.
\begin{center}
\small
\typerule{Emulate}{
	\trg{t} = \trg{\OB{EM}};\trg{\OB{EO}};\trg{\OB{RM}};\trg{\OB{RO}}
	&
	\fun{methodKnowledge}{\trg{\OB{EM}},\trg{\OB{RM}}} = M_0 
	\\
	\fun{objectsKnowledge}{\trg{\OB{EO}},\trg{\OB{RO}}} = V_0
	&
	\src{m_0} = \lst{main}
	\\
	\emul{\trg{\alpha_1}, M_0,V_0, \emptyset, \src{m_0}, \emptyset, 0,\trg{t}} = \src{A_1}, V_1, \OB{\trg{id}_1}, i_1, \OB{\src{m}_1}, \OB{\src{t}_1}
	\\
	\forall h \in 2 .. n,
	&
	\emul{\trg{\alpha_h}, M_0,V_0+V_{h-1}, \OB{\trg{id}_{h-1}}, \OB{\src{m}_{h-1}}, \OB{\src{t}_{h-1}}, i_{h-1},\trg{t}} = \src{A_h}, V'_h, \OB{\trg{id}_h}, i_h, \OB{\src{m}_h}, \OB{\src{t}_h}
	\\
	V' = V'_1 + \cdots + V'_n
}{
	\emul{\trg{\alpha_1},\cdots\trg{\alpha_n},\trg{t}} = \src{A_1};\cdots;\src{A_n},M_0, V', \OB{\src{m}_n}, \OB{\src{t}_n},i_{n-1}
}{emulate}\quad
\typerule{Emulate- fail}{
	\text{ a \fail~ happens}
}{
	\emul{\trg{\alpha_1},\cdots\trg{\alpha_n}} = \src{\emptyset}@\lst{main}
}{em-f}\quad
\typerule{Emulate - call}{
	\trg{a'} = (\trg{1},2*\trg{\mc{N}_w})
	&
	\trg{a} = \asm{r_3},\asm{r_4}
	\\
	M(\trg{a})=\src{m}:\src{t}(\src{{t_1};\cdots;{t_m}})\to\src{t'}
	&
	\forall i \in 1..m
	\\
	\src{E} = \lst{if (oc.isStep( i )) then oc.incrStep();} + \src{E_1} +\cdots+\src{E_m}+\src{E'}
	\\
	\src{E'} = \lst{ var \src{o} = \src{E^o} ;var retvar = } \src{o.m} ( \lst{arg-1}, \cdots, \lst{arg-m})
	\\
	\emul{\trg{w_6}:\src{t},V} = \src{E^o};\emptyset
	&
	\emul{\trg{w_{6+i}}:\src{t_i},V} = \src{E_i^v}, {V'_i}
	\\
	\src{E_i} = \lst{var arg-i = }\src{E_i^v}
	&
	V' = V'_1 + \cdots + V'_n
}{
	\emul{\trg{\cl{\trg{a'}\ \trg{w_0,\cdots,w_n}}}, M_0,V,\OB{\trg{id}}, \OB{\src{m}}, \OB{\src{t}}, i,\trg{t}} = \src{E}@\src{m}, V', \trg{w_0}\OB{\trg{id}}, i+1, \src{m}\OB{\src{m}}, \src{t'}\OB{\src{t}}
}{emulcall}\quad
\typerule{Emulate - call - fail}{
	(\trg{w_0},\trg{w_5}) \text{ is not an executable address in }\trg{t}
}{
	\emul{\trg{\cl{\trg{a}\ \trg{w_1,\cdots,v_n}}}, M_0,V,\OB{\trg{id}}, \OB{\src{m}}, \OB{\src{t}}, i,\trg{t}} = \fail
}{emulcall-f}\quad
\typerule{Emulate - call -fail2}{
	\trg{a'} = (\trg{1},2*\trg{\mc{N}_w})
	&
	\trg{a} = \asm{r_3},\asm{r_4}
	&
	\trg{a} \notin M
}{
	\emul{\trg{\cl{\trg{a'}\ \trg{w_0,\cdots,w_n}}}, M_0,V,\OB{\trg{id}}, \OB{\src{m}}, \OB{\src{t}}, i,\trg{t}} = \fail
}{emulcall-f2}\quad
\typerule{Emulate - callback}{
	M(\trg{a})=\src{m}:\src{t}(\src{{t_1};\cdots;{t_m}})\to\src{t'}
	&
	\src{E} = \lst{oc.incrStep();}+\src{E_1 } + \cdots + \src{E_m}
	\\
	V' = V'_1 + \cdots + V'_m
	&
	\text{Let the name of the } i\text{-th parameter be }\lst{x}_i
	\\
	\forall i \in 1..m .\src{E_i}, V'_i = \fun{possiblyAdd}{\trg{w_{6+i}},\src{t_i}, \lst{x}_i}
}{
	\emul{\trg{\cb{\trg{a}\ \trg{w_0,\cdots,w_n}}}, M_0,V,\OB{\trg{id}}, \OB{\src{m}}, \OB{\src{t}}, i,\trg{t}} = \src{E}@\src{m}, V', \trg{w_0}\OB{\trg{id}}, i+1, \src{m}\OB{\src{m}}, \src{t'}\OB{\src{t}}
}{emulcb}\quad
\typerule{possiblyAdd}{
	\fun{isInternal}{\src{t}}
	&
	\trg{w}\notin\dom{V}
	\\
	\src{E} = \lst{var arg-i = addObject( x, \fun{nonce-to-int}{\trg{w}})}
	&
	V = \trg{w}:\src{t}
}{
	\fun{possiblyAdd}{\trg{w},\src{t}, \lst{x}} = \src{E},V
}{possadd}\quad
\typerule{Emulate - returnback}{
	\src{E} = \lst{if (oc.isStep( i ) ) then oc.incrStep(); var ret = } \src{E^v} \lst{; return ret;}
	\\
	\emul{\trg{w}:\src{t},V} = \src{E^v}, {V'}
	&
	\trg{w'} = \trg{id}
	&
	\trg{a} = (\trg{1},3*\trg{\mc{N}_w})
}{
	\emul{\trg{\rb{\trg{a}\ \trg{w},\trg{w'}}}, M_0,V, \trg{id}\OB{\trg{id}}, \src{m}\OB{\src{m}}, \src{t}\OB{\src{t}}, i,\trg{t}} = \src{E}@\src{m}, V', \OB{\trg{id}}, i+1, \OB{\src{m}}, \OB{\src{t}}
}{emulretback}\quad
\typerule{Emulate - returnback - fail}{
	\trg{w'} \neq\trg{id}
	&
	\trg{a} = (\trg{1},3*\trg{\mc{N}_w})
}{
	\emul{\trg{\rb{\trg{a}\ \trg{w},\trg{w'}}}, M_0,V, \trg{id}\OB{\trg{id}}, \src{m}\OB{\src{m}}, \src{t}\OB{\src{t}}, i,\trg{t}} = \fail
}{emulretback-f}\quad
\typerule{Emulate - returnback - fail2}{
	\trg{a} \neq(\trg{1},3*\trg{\mc{N}_w})
}{
	\emul{\trg{\rb{\trg{a}\ \trg{w},\trg{w'}}}, M_0,V, \OB{\trg{id}}, \OB{\src{m}}, \src{t}\OB{\src{t}}, i,\trg{t}} = \fail
}{emulretback-f2}\quad
\typerule{Emulate - return}{
	\fun{isInternal}{\src{t}}
	&
	\trg{w}\notin\dom{V}
	&
	V'= \trg{w}:\src{t}
	\\
	\src{E} = \lst{oc.incrStep(); var arg-r = addObject( retvar, \fun{nonce-to-int}{\trg{w}})}
}{
	\emul{\trg{\rt{\trg{a}\ \trg{w},\trg{w'}}}, M_0,V, \trg{id}\OB{\trg{id}}, \src{m}\OB{\src{m}}, \src{t}\OB{\src{t}}, i,\trg{t}} = \src{E}@\src{m}, V', i+1, \OB{\trg{id}}, \OB{\src{m}}, \OB{\src{t}}
}{emulret}\quad
\typerule{Emulate - return2}{
	\lnot\fun{isInternal}{\src{t}}
	&
	\src{E} = \lst{oc.incrStep();}
	&
	V' = \emptyset
}{
	\emul{\trg{\rt{\trg{a}\ \trg{w},\trg{w'}}}, M_0,V, \trg{id}\OB{\trg{id}}, \src{m}\OB{\src{m}}, \src{t}\OB{\src{t}}, i,\trg{t}} = \src{E}@\src{m}, V', i+1, \OB{\trg{id}}, \OB{\src{m}}, \OB{\src{t}}
}{emulret2}\quad
\typerule{Emulate - addresses}{
	M(\trg{a}) = \src{m}:\src{M_t}
}{
	\emul{\trg{a},M} = \src{m}:\src{M_t}
}{emuladdrs}\quad
\typerule{Emulate - addresses - fail}{
	\trg{a}\notin\dom{M}
}{
	\emul{\trg{a},M} = \fail
}{emuladdrs-f}\quad
\typerule{Emulate - values - unit}{
	\src{t}\equiv\src{\lstb{Unit}}
	&
	\trg{w}\equiv\comp{\src{\lst{unit}}}
	&
	\src{E} = \src{\lst{unit}}
	&
	V' = \emptyset
}{
	\emul{\trg{w}:\src{t},V} = \src{E}, {V'}
}{emulvals-unit}\quad
\typerule{Emulate - values - unit fail}{
	\src{t}\equiv\src{\lstb{Unit}}
	&
	\trg{w}\nequiv\comp{\src{\lst{unit}}}
}{
	\emul{\trg{w}:\src{t},V} = \fail
}{emulvals-f}\quad
\typerule{Emulate - values - bool true}{
	\src{t}\equiv\src{\lstb{Bool}}
	&
	\trg{w}\equiv\comp{\src{\lst{true}}}
	&
	\src{E} = \src{\lst{true}}
	&
	V'=\emptyset
}{
	\emul{\trg{w}:\src{t},V} = \src{E}, {V'}
}{emulvals-bt}\quad
\typerule{Emulate - values - bool false}{
	\src{t}\equiv\src{\lstb{Bool}}
	&
	\trg{w}\equiv\comp{\src{\lst{false}}}
	&
	\src{E} = \src{\lst{false}}
	&
	V' = \emptyset
}{
	\emul{\trg{w}:\src{t},V} = \src{E}, {V'}
}{emulvals-bf}\quad
\typerule{Emulate - values- bool fail}{
	\src{t}\equiv\src{\lstb{Bool}}
	&
	\trg{w}\nequiv\comp{\src{\lst{true}}}
	&
	\trg{w}\nequiv\comp{\src{\lst{false}}}
}{
	\emul{\trg{w}:\src{t},V} = \fail
}{emulvals-bfail}\quad
\typerule{Emulate - values - int}{
	\src{t}\equiv\src{\lstb{Int}}
	&
	\src{E} = \fun{integerFor}{\trg{w}}
	&
	V' = \emptyset
}{
	\emul{\trg{w}:\src{t},V} = \src{E}, {V'}
}{emulvals-int}\quad
\typerule{Emulate - values - null}{
	\src{t}<:\src{\lstb{Obj}}
	&
	\trg{w}\equiv\comp{\src{\lst{null}}}
	&
	\src{E} = \src{\lst{null}}
	&
	V' = \emptyset 
}{
	\emul{\trg{w}:\src{t},V} = \src{E}, {V'}
}{emulvals-null}\quad
\typerule{Emulate - values - external old}{
	\src{t}<:\src{\lstb{Obj}}
	&
	\fun{isExternal}{\src{t}}
	&
	\trg{w}\in V
	\\
	\src{E} = \lst{oc.getByName( \fun{nonce-to-int}{\trg{w_i}} )}
	&
	V' = \emptyset
}{
	\emul{\trg{w}:\src{t},V} = \src{E}, {V'}
}{emulvals-externold}\quad
\typerule{Emulate - values - external old - fail}{
	\src{t}<:\src{\lstb{Obj}}
	&
	\fun{isExternal}{\src{t}}
	&
	\trg{w}:\src{t'}\in V
	&
	\src{t}\not<:\src{t'}
}{
	\emul{\trg{w}:\src{t},V} = \fail
}{emulvals-externold-f}\quad
\typerule{Emulate - values - external new}{
	\src{t}<:\src{\lstb{Obj}}
	&
	\fun{isExternal}{\src{t}}
	&
	\trg{n}\notin\dom{V }
	\\
	\src{E}=\lst{oc.createNew-t( \trg{n} );}
	&
	V' = \trg{n}:\src{t}
}{
	\emul{\trg{n}:\src{t},V} = \src{E}, {V'}
}{emulvals-externnew}\quad
\typerule{Emulate - values - intern}{
	\src{t}<:\src{\lstb{Obj}}
	&
	\fun{isInternal}{\src{t}}
	&
	\trg{w}\in \dom{V}
	\\
	\src{E} = \lst{oc.getByName( \fun{nonce-to-int}{\trg{w}} )}
	&
	V' = \emptyset
}{
	\emul{\trg{w}:\src{t},V} = \src{E}, {V'}
}{emulvals-intern}\quad
\typerule{Emulate - values - intern - fail}{
	\src{t}<:\src{\lstb{Obj}}
	&
	\fun{isInternal}{\src{t}}
	&
	\trg{w}\notin V
}{
	\emul{\trg{w}:\src{t},V} = \fail
}{emulvals-intern-f}\quad
\typerule{Emulate - obj}{
	\src{t}\equiv\src{\lstb{Obj}}
	&
	\trg{w}\in \dom{V}
	&
	\src{E} = \lst{oc.getByName( \fun{nonce-to-int}{\trg{w}} )}
	&
	V' = \emptyset
}{
	\emul{\trg{w}:\src{t},V} = \src{E}, {V'}
}{emobj}\quad
\typerule{Emulate - obj -ex}{
	\src{t}\equiv\src{\lstb{Obj}}
	&
	\trg{n}\notin\dom{V}
	&
	\src{E}=\lst{oc.createNew-\src{t'}( \trg{n} );}
	&
	V' = \trg{n}:\src{t'}
	\\
	\text{If }\trg{n}\text{ is used again, let it be with type }\src{t'}\neq\src{\lstb{Obj}}\text{ otherwise let }\src{t'}\text{ be any external type}
}{
	\emul{\trg{n}:\src{t},V} = \src{E}, {V'}
}{emobj2}\quad
\typerule{Emulate - obj -ex fail}{
	\src{t}\equiv\src{\lstb{Obj}}
	&
	\trg{\pi}\notin \dom{V}
}{
	\emul{\trg{\pi}:\src{t},V} = \fail
}{emobj2-f}\quad
\typerule{Emulate-type - unit}{
	\trg{w}=\comp{\src{\lstb{Unit}}}
}{
	\emul{\trg{w}} =\src{\lstb{Unit}}
}{emul-ty-u}\quad
\typerule{Emulate-type - bool}{
	\trg{w}=\comp{\src{\lstb{Bool}}}
}{
	\emul{\trg{w}} =\src{\lstb{Bool}}
}{emul-ty-b}\quad
\typerule{Emulate-type - Int}{
	\trg{w}=\comp{\src{\lstb{Int}}}
}{
	\emul{\trg{w}} =\src{\lstb{Int}}
}{emul-ty-i}\quad
\typerule{Emulate-type - class}{
	\trg{w}=\comp{\src{c}}
}{
	\emul{\trg{w}} =\src{c}
}{emul-ty-u}\quad
\typerule{intfor - number}{}{
	\fun{integerFor}{\trg{n}}=\src{n}
}{intfor}\quad
\typerule{intfor - nonces}{}{
	\fun{integerFor}{\trg{\pi}}=\src{0}
}{intforp}\quad
\typerule{intfor -symbol}{}{
	\fun{integerFor}{\trg{\sigma}}=\fail
}{intfors}\quad
\typerule{isinternal}{
	\text{Type }\src{t}\text{ is a class in }\src{\mc{C}_1}\text{ or }\src{\mc{C}_2}
}{
	\fun{isInternal}{\src{t}}
}{isinter}\quad
\typerule{isExternal}{
	\text{Type }\src{t}\text{ is an interface in }\src{\mc{C}_1}\text{ or }\src{\mc{C}_2}
}{
	\fun{isExternal}{\src{t}}
}{isext}\quad
\typerule{isPrimitive}{
	\text{Type }\src{t}\text{ is }\src{\lstb{Unit}}, \src{\lstb{Bool}} \text{ or }\src{\lstb{Int}}
}{
	\fun{isPrimitive}{\src{t}}
}{isprim}\quad
\typerule{emulate-paper}{
	\fun{linkof}{\src{\mc{C}}} = \trg{t} 
}{
	\emul{\trg{\OB{\alpha}},\src{\mc{C}}} = \emul{\trg{\OB{\alpha}},\trg{t}}
}{em-p}\quad
\typerule{linkof-base}{
	\comp{\src{C}} = (\trg{m_c},\trg{m_d},\trg{t})
}{
	\fun{linkof}{\src{C}} = \trg{t}
}{lkof-base}\quad
\typerule{linkof-ind}{
	\forall i \in 1 .. n .
	&
	\fun{linkof}{\src{C_i}} = \trg{t_i}
}{
	\fun{linkof}{\src{C_1,\cdots,C_n}} = \trg{t_1}\uplus\cdots\uplus\trg{t_n}
}{lkof-ind}\quad
\end{center}

Since whole \aim programs are never executed with free symbols in their memory, \Cref{tr:intfors} is never used.

It is not a problem to accept inputs where the return address is not executable; in those cases there will be an automatic \trg{\surd}.
The emulation rules need to keep track of that, as they do with \Cref{tr:emulcall-f}.

Emulation of actions to \sys.
\begin{center}
\typerule{Emulate - call - fwcall -w1}{
	\trg{a}\equiv(\trg{1},2*\trg{\mc{N}_w})
	&
	\trg{w_0}==\trg{id}
}{
	\emul{\trg{\cl{\trg{a}\ \trg{w_0,\cdots,w_n}}}, M_0, V, \trg{id}\OB{\trg{id}}, \OB{\src{m}}, \OB{\src{t}}, i,\trg{t}} = \fail
}{emulcall-fwcall-w1}\quad
\typerule{Emulate - call - fwcall -w2}{
	\trg{a}\equiv(\trg{1},2*\trg{\mc{N}_w})
	&
	\trg{w_3} == \trg{1}
}{
	\emul{\trg{\cl{\trg{a}\ \trg{w_0,\cdots,w_n}}}, M_0, V, \OB{\trg{id}}, \OB{\src{m}}, \OB{\src{t}}, i,\trg{t}} = \fail
}{emulcall-fwcall-w2}\quad
\typerule{Emulate - call - fwret - w1}{
	\trg{a}\equiv(\trg{1},3*\trg{\mc{N}_w})
	&
	\trg{id}\neq\trg{w_0}
}{
	\emul{\trg{\cl{\trg{a}\ \trg{w_0,\cdots,w_n}}}, M_0, V, \OB{\trg{id}}, \OB{\src{m}}, \OB{\src{t}}, i,\trg{t}} = \fail
}{emulcall-fwret-w1}\quad
\typerule{Emulate - call - fwret - w2}{
	\trg{a}\equiv(\trg{1},3*\trg{\mc{N}_w})
}{
	\emul{\trg{\cl{\trg{a}\ \trg{w_0,\cdots,w_n}}}, M_0, V, \emptyset, \OB{\src{m}}, \OB{\src{t}}, i,\trg{t}} = \fail
}{emulcall-fwret-w2}\quad
\typerule{Emulate - call - testobj}{
	\trg{a}\equiv(\trg{1},\trg{0})
	&
	\src{t'}=\src{\lstb{Bool}}
	\\
	\src{E} = \lst{if (oc.isStep( i ) ) then oc.incrStep(); var retvar = instanceof(} \src{E_w} \lst{,} \src{E_t} \lst{)}
	\\
	\emul{\trg{w_7}:\src{E_t},V} = \src{E_w}, {V'}
	&
	\emul{\trg{w_8}} = \src{E_t}, \emptyset
}{
	\emul{\trg{\cl{\trg{a}\ \trg{w_0,\cdots,w_n}}}, M_0, V, \OB{\trg{id}}, \src{m}\OB{\src{m}}, \OB{\src{t}}, i,\trg{t}} = \src{E}@\src{m}, V', \trg{w_0}\OB{\trg{id}}, i+1, \src{m}\OB{\src{m}}, \src{t'}\OB{\src{t}}
}{emulcall-testobj}\
\typerule{Emulate - call - testobj - w}{
	\trg{a}\equiv(\trg{1},\trg{0})
	&
	\trg{w_7}\notin\dom{V}
}{
	\emul{\trg{\cl{\trg{a}\ \trg{w_0,\cdots,w_n}}}, M_0, V, \OB{\trg{id}}, \OB{\src{m}}, \OB{\src{t}}, i,\trg{t}} = \fail
}{emulcall-testobj-w}\quad
\typerule{Emulate - call - regobj}{
	\trg{a}\equiv(\trg{1},\trg{\mc{N}_w})
	&
	\src{t'}=\src{\lstb{Unit}}
	\\
	\src{E} = \lst{if (oc.isStep( i ) ) then oc.incrStep(); var retvar = unit; }
}{
	\emul{\trg{\cl{\trg{a}\ \trg{w_0,\cdots,w_n}}}, M_0, V, \OB{\trg{id}}, \src{m}\OB{\src{m}}, \OB{\src{t}}, i,\trg{t}} = \src{E}@\src{m}, \emptyset, \trg{w_0}\OB{\trg{id}}, i+1, \src{m}\OB{\src{m}}, \src{t'}\OB{\src{t}}
}{emulcall-regobj}\quad
\typerule{Emulate - call - regobj - w}{
	\trg{a}\equiv(\trg{1},\trg{\mc{N}_w})
	&
	\trg{w_7}\in\dom{V}
}{
	\emul{\trg{\cl{\trg{a}\ \trg{w_0,\cdots,w_n}}}, M_0, V, \OB{\trg{id}}, \OB{\src{m}}, \OB{\src{t}}, i,\trg{t}} = \fail
}{emulcall-regobj-w}\quad
\end{center}

\subsection{\diff{\trg{\alpha_1},\trg{\alpha_2},i,M,V, \OB{\src{m}}, \OB{\src{t}}}}\label{sec:diff}
All cases have their dual, where the actions are flipped.
The arguments are: the different action for \src{C_1}, the different action for \src{C_2}, the index preceding those actions, the knowledge of method bindings, the knowledge of object bindings, the stack of current methods and the stack of expected return types.
In the following, indicate an empty action with $\emptyset$.

\begin{center}
\small
\typerule{Diff-length-tick}{
}{
	\diff{\trg{\surd},\emptyset,i,M,V, \src{m}\OB{\src{m}}, \OB{\src{t}}} ={\src{\emptyset}@\src{m}}@i,{\src{\emptyset}@\src{m}}@i
}{diff-lt} 
\quad
\typerule{Diff-length-call}{
	\src{E} = \lst{ exit( 1 ); }
	&
	\src{m}:\src{M_t} = \emul{\trg{a},M}
}{
	\diff{\cb{\trg{a}\ \trg{\OB{w}}},\emptyset,i,M,V, \OB{\src{m}}, \OB{\src{t}}} ={\src{E}@\src{m}}@i,{\src{E}@\src{m}}@i
}{diff-ct}
\quad
\typerule{Diff-length-return}{
	\src{E} = \lst{ exit( 1 ); }
}{
	\diff{\rt{\trg{a}\ \trg{w,w'}},\emptyset,i,M,V, \src{m}\OB{\src{m}}, \OB{\src{t}}} ={\src{E}@\src{m}}@i,{\src{E}@\src{m}}@i
}{diff-rtt}
\quad
\typerule{Diff-rets-primitive}{
	\fun{isPrimitive}{\src{t}}
	&
	\emul{\trg{w}:\src{t},V} = \src{E^r},V'
	\\
	\src{E} = \lst{if ( retvar == \src{E^r} ) \{ exit( 1 ); \} else \{ oc.diverge(); \} }
}{
	\diff{\rt{\trg{a}\ \trg{w,w'}},\rt{\trg{a}\ \trg{w_1,w_1'}},i,M,V, \src{m}\OB{\src{m}}, \src{t}\OB{\src{t}}} ={\src{E}@\src{m}}@i,{\src{E}@\src{m}}@i
}{diff-rs}
\quad
\typerule{Diff-rets-addr}{
	\text{the code of \comp{\cdot} only uses what was passed in \asm{r_0} and \asm{r_5} as addresses to return}
	\\
	\text{so this case never arises because the addresses where to return}
	\\
	\text{ are always the same (they come from the same \trg{?}-decorated actions)}
}{
	\diff{\rt{\trg{a'}\ \trg{w,w'}},\rt{\trg{a}\ \trg{w,w'}},i,M,V, \src{m}\OB{\src{m}}, \src{t}\OB{\src{t}}} =\emptyset
}{diff-ra}
\quad
\quad
\typerule{Diff-rets-internal}{
	\fun{isInternal}{\src{t}}
	&
	\emul{\trg{w}:\src{t},V} = \src{E^r},V' 
	\\
	\src{E} = \lst{if ( retvar == \src{E^r} ) \{ exit( 1 ); \} else \{ oc.diverge(); \}}
}{
	\diff{\rt{\trg{a}\ \trg{w,w'}},\rt{\trg{a}\ \trg{w_1,w_1'}},i,M,V, \src{m}\OB{\src{m}}, \src{t}\OB{\src{t}}} ={\src{E}@\src{m}}@i,{\src{E}@\src{m}}@i
}{diff-rsi}
\quad
\typerule{Diff-rets-external}{
	\fun{isExternal}{\src{t}}
	&
	\emul{\trg{w}:\src{t},V} = \src{E^r},V' 
	\\
	\src{E} = \lst{if ( retvar == \src{E^r} ) \{ exit(1); \} else \{ oc.diverge(); \} }
}{
	\diff{\rt{\trg{a}\ \trg{w,w'}},\rt{\trg{a}\ \trg{w_1,w_1'}},i,M,V, \src{m}\OB{\src{m}}, \src{t}\OB{\src{t}}} ={\src{E}@\src{m}}@i,{\src{E}@\src{m}}@i
}{diff-rse}
\quad
\typerule{Diff-calls-methods}{
	\src{E_1} = \lst{ exit(1); }
	&
	\emul{\trg{a},M} = \src{m_1}:\src{M_t}_1 
	\\
	\src{E_2}= \lst{ oc.diverge(); }
	&
	\emul{\trg{a'},M} = \src{m_2}:\src{M_t}_2  
}{
	\diff{\cb{\trg{a}\ \trg{\OB{w}}},\cb{\trg{a'}\ \trg{\OB{w}}},i,M,V, \OB{\src{m}}, \OB{\src{t}}} ={\src{E_1}@\src{m_1}}@i,{\src{E_2}@\src{m_2}}@i
}{diff-cs}
\quad
\typerule{Diff-callee}{
	\trg{\OB{w}} = \trg{w_1},\cdots,\trg{w_n}
	&
	\trg{\OB{w}'} = \trg{w_1'},\cdots,\trg{w_n'}
	&
	\trg{w_5\neq\trg{w_5'}}
	\\
	\emul{\trg{a},M} = \src{m}:\src{t}(\src{{t}_1,\cdots,\src{t_n}})\to\src{t'}  
	&
	\emul{\trg{w_5}:\src{t},V} = \src{E^r},V' 
	\\
	\src{E} = \lst{if ( this == \src{E^r} ) \{ exit( 1 ); \} else \{ oc.diverge(); \}}
}{
	\diff{\cb{\trg{a}\ \trg{\OB{w}}},\cb{\trg{a}\ \trg{\OB{w}'}},i,M,V, \OB{\src{m}}, \OB{\src{t}}}={\src{E}@\src{m}}@i,{\src{E}@\src{m}}@i
}{diff-callee}
\quad
\typerule{Diff-calls2-parameters - primitive}{
	\trg{\OB{w}} = \trg{w_1},\cdots,\trg{w_n}
	&
	\trg{\OB{w}'} = \trg{w_1'},\cdots,\trg{w_n'}
	&
	\exists i \in 6 .. n .\trg{w_i\neq\trg{w_i'}}
	\\
	\emul{\trg{a},M} = \src{m}:\src{t}(\src{{t}_1,\cdots,\src{t_n}})\to\src{t'}  
	&
	\fun{isPrimitive}{\src{t_i}}
	&
	\emul{\trg{w_i}:\src{t_i},V} = \src{E^r},V'
	\\
	\src{E} = \lst{if ( retvar == \src{E^r} ) \{ exit( 1 ); \} else \{ oc.diverge(); \} }
}{
	\diff{\cb{\trg{a}\ \trg{\OB{w}}},\cb{\trg{a}\ \trg{\OB{w'}}},i} = \src{E}@\src{m}@i,\src{E}@\src{m}@i
}{diff-cs2-p}
\quad
\typerule{Diff-calls2-parameters - internal}{
	\trg{\OB{w}} = \trg{w_1},\cdots,\trg{w_n}
	&
	\trg{\OB{w}'} = \trg{w_1'},\cdots,\trg{w_n'}
	&
	\exists i \in 6 ..n .\trg{w_i\neq\trg{w_i'}}
	\\
	\emul{\trg{a},M} = \src{m}:\src{t}(\src{{t}_1,\cdots,\src{t_n}})\to\src{t'}  
	&
	\fun{isInternal}{\src{t_i}}
	&
	\emul{\trg{w_i}:\src{t_i},V} = \src{E^r},V'
	\\
	\text{Let the name of the } i\text{-th parameter be }\lst{x}_i
	\\
	\src{E} = \lst{if ( x$_i$ == \src{E^r} ) \{ exit(1); \} else \{ oc.diverge(); \} }
}{
	\diff{\cb{\trg{a}\ \trg{\OB{w}}},\cb{\trg{a}\ \trg{\OB{w'}}},i} = \src{E}@\src{m}@i,\src{E}@\src{m}@i
}{diff-cs2-i}
\quad
\typerule{Diff-calls2-parameters - external}{
	\trg{\OB{w}} = \trg{w_1},\cdots,\trg{w_n}
	&
	\trg{\OB{w}'} = \trg{w_1'},\cdots,\trg{w_n'}
	&
	\exists i \in 6 .. n.\trg{w_i\neq\trg{w_i'}}
	\\
	\emul{\trg{a},M} = \src{m}:\src{t}(\src{{t}_1,\cdots,\src{t_n}})\to\src{t'} 
	&
	\fun{isExternal}{\src{t_i}}
	\\
	\text{Let the name of the } i\text{-th parameter be }\lst{x}_i
	&
	\emul{\trg{w_i}:\src{t_i},V} = \src{E^r},V'
	\\
	\src{E} = \lst{if ( x$_i$ == \src{E^r} ) \{ exit(1); \} else \{ oc.diverge(); \} }
}{
	\diff{\cb{\trg{a}\ \trg{\OB{w}}},\cb{\trg{a}\ \trg{\OB{w'}}},i} = \src{E}@\src{m}@i,\src{E}@\src{m}@i
}{diff-cs2-e}
\quad
\typerule{Diff-r0--r4}{
	\trg{\OB{w}} = \trg{w_1},\cdots,\trg{w_n}
	&
	\trg{\OB{w}'} = \trg{w_1'},\cdots,\trg{w_n'}
	&
	\exists i \in 0,1,2,4.\trg{w_i\neq\trg{w_i'}}
	\\
	\text{\asm{r_0} to \asm{r_2} are always reset to \trg{0} so this case does not exist}
	\\
	\text{\asm{r_4} contains the returnback entry point address, which is always \trg{0}}
}{
	\diff{\cb{\trg{a}\ \trg{\OB{w}}},\cb{\trg{a}\ \trg{\OB{w}'}},i,M,V, \OB{\src{m}}, \OB{\src{t}}}=\emptyset
}{diff-r0}
\quad
\quad
\typerule{Diff-acts}{
	\src{E_1} = \lst{ exit( 1 ); }
	&
	\emul{\trg{a},M} = \src{m_1}:\src{M_t} 
	&
	\src{E_2} = \lst{ oc.diverge(); }
}{
	\diff{\cb{\trg{a}\ \trg{\OB{w}}},\rt{\trg{a'}\ \trg{w,w'}},i,M,V, \src{m_2}\OB{\src{m}}, \OB{\src{t}}} ={\src{E_1}@\src{m_1}}@i,{\src{E_2}@\src{m_2}}@i
}{diff-ac}
\quad
\typerule{Diff-acts2}{
	\emul{\trg{a},M} = \src{m_1}:\src{M_t} 
	&
	\src{E_1} = \lst{ oc.diverge(); }
}{
	\diff{\cb{\trg{a}\ \trg{\OB{w}}},\trg{\surd},i,M,V, \OB{\src{m}}, \OB{\src{t}}} ={\src{E_1}@\src{m_1}}@i,{\src{E_1}@\src{m_1}}@i
}{diff-ac2}
\quad
\typerule{Diff-acts3}{
	\src{E_1} = \lst{ oc.diverge(); }
}{
	\diff{\rt{\trg{a'}\ \trg{w,w'}},\trg{\surd},i,M,V, \src{m_1}\OB{\src{m}}, \OB{\src{t}}} ={\src{E_1}@\src{m_1}}@i,{\src{E_1}@\src{m_1}}@i
}{diff-ac3}
\quad
\end{center}
Since the code is placed inside the code generated by the last action $i$, the code generated by \diff{\cdot} does not need to be wrapped in an \lst{if ( oc.isStep( i ))}.

%% file: bodies2/actualproofs.tex
\section{Proofs of \Cref{sec:proof}}\label{sec:actualproofs}
This section contains only proofs and proof sketches that are not found in the paper.

\myparagraph{Well-behaved Contextual Equivalence}
A form of contextual equivalence considering a restricted number of contexts is \emph{well-behaved contextual equivalence}, denoted as $\trg{M_1}\ceqwt\trg{M_2}$~\cite{exprPA}.
Well-behaved contextual equivalence is defined for \aim modules with respect to \jem contexts.
Well-behaved contextual equivalence is analogous to contextual equivalence except that it considers a subset of \aim contexts, namely those that behave like \jem ones.

The notion of well-behaved contextual equivalence is used for compiler correctness specification, \Cref{ass:comp-corr} can be restated as \Cref{ass:compcorr}.
\begin{assumption}[\comp{\cdot} is correct]\label{ass:compcorr}
$\forall\src{\mc{C}_1},\src{\mc{C}_2}.\src{\mc{C}_1}\ceqjem\src{\mc{C}_2}\iff\comp{\src{\mc{C}_1}}\ceqwt\comp{\src{\mc{C}_2}}$.
\end{assumption}

Proof sketch of \Cref{thm:compaimcorr}
\begin{proofsketch}
The \prot{\cdot} phase of \compaim{\cdot} does not modify how expressions are compiled (i.e., the workings of \comp{\cdot}) except for how \lstb{instanceof} is compiled and how external functions are invoked.

Additionally, the methodEP code is prepended to all method calls and the returnEP code is prepended to all returns.

These procedures contain checks that can abort, though correct code never trigger them. 

As presented in \Cref{prob:types2}, \lstb{instanceof} is compiled to \fun{testObj}{\cdot}.

We can see that this is correct by simply inspecting its code.

The other modification done to \comp{\cdot} code is in the way external functions are called.

The chaining of \fun{preamble}{\cdot} ensures that securely-compiled modules are exited via \sys.

The \fun{forwardCall}{\cdot} function respects the registers calling convention of \Cref{ass:comp-cc}.

Moreover, it only insists that returnbacks are done to the \fun{forwardReturn}{\cdot} entry point.

For well-behaving contexts, that entry point just acts as a proxy for the real return entry point, so no semantics change is introduced there.

Moreover, \fun{forwardReturn}{\cdot} also respects \Cref{ass:comp-cc}.

Well-behaving contexts will not call any entry point in \sys except for \fun{testObject}{\cdot} and \fun{registerObject}{\cdot}.

The former is a correct implementation of \lstb{instanceof}, as already discussed.

The second one will never abort for well-behaving contexts since only new objects will be registered.

Therefore, the additional functionalities of \sys do not alter the semantics of correct programs.

Because of this, \compaim{\cdot} satisfies \Cref{ass:compcorr}.

So we have: $\forall\src{\mc{C}_1},\src{\mc{C}_2}.\src{\mc{C}_1}\ceqjem\src{\mc{C}_2}\iff\compaim{\src{\mc{C}_1}}\ceqwt\compaim{\src{\mc{C}_2}}$.

\jem-like-behaving contexts in \aim are a subset of all the \aim contexts, we have $\compaim{\src{\mc{C}_1}}\ceqjem\compaim{\src{\mc{C}
_2}}\Rightarrow\compaim{\src{\mc{C}_1}}\ceqwt\compaim{\src{\mc{C}_2}}$.

Because we changed the calling convention, we need to show that all of these well-behaved contexts, which used the \comp{\cdot} calling convention, can be made to adopt the \compaim{\cdot} one.

This can be done by adding the following instructions before a jump to a module -- intuitively at the single exit point of those contexts:
\begin{itemize}
	\item assume the address where to jump is stored in \asm{r_n}, \asm{r_i} where $n$ and $i$ are any register with index from 0 to 4 (it can't be other indices as they are needed by the \comp{\cdot} calling convention)
	\item store \asm{r_n} into \asm{r_4} and \asm{r_i} into \asm{r_3}
	\item load \trg{1} into \asm{r_i} and 2*\trg{\mc{N}_w} into \asm{r_n}
	\item \asm{jmp\ r_n\ r_i}
\end{itemize}

We can therefore chain the implications above to conclude the thesis.
\end{proofsketch}

\begin{lemma}[Method emulation is correct]\label{lem:methemulcorr}
$\emul{\trg{a},M} = \src{m}:\src{M_t} \iff \comp{\src{m}}=\trg{a}$.
\end{lemma}
\begin{proofsketch}
This can be seen by \Cref{tr:emuladdrs,tr:emuladdrs-f} and the constitution of $M_0$ in \Cref{tr:mk}.
\end{proofsketch}
\begin{lemma}[Value emulation is correct]\label{lem:valemulcorr}
$\emul{\trg{w}:\src{t},V} = \src{E}, V' \iff \comp{\src{E}}=\trg{w}$.
\end{lemma}
\begin{proofsketch}
This can be seen by \Cref{tr:emulvals-unit,tr:emulvals-bt,tr:emulvals-bf,tr:emulvals-int,tr:emulvals-null,tr:emulvals-externold,tr:emulvals-f,tr:emulvals-bfail,tr:emulvals-externnew,tr:emulvals-intern,tr:emulvals-intern-f} and by the constitution of $V_0$ in \Cref{tr:ok} and by the update on $V$ in all \emul{\cdot} rules in \Cref{sec:emul}.
\end{proofsketch}

Proof of \Cref{thm:maintheorem}. 

As a helper, here is a table of all additionally-defined functions and whether they can abort (i.e., call \asm{halt}) or not.

\begin{tabular}{ | l | c |}
\hline
testObj
	& yes
		\\\hline
registerObj
	& yes
		\\\hline
forwardCall
	& yes, twice
		\\\hline
forwardReturn
	& yes, twice
		\\\hline
classOf
	& no
		\\\hline
isinternal
	& no
		\\\hline
updatemaskingtable
	& no (but it calls registerObj)
		\\\hline
loadobject
	& no (but it calls lookup in \mc{T})
		\\\hline
lookup in \mc{T} 
	& yes
		\\\hline
maskingTable
	& no
		\\\hline
dynamictypechecks
	& yes, twice plus it calls testObj
		\\\hline
maskObjectid
	& no
		\\\hline
signatureof
	&no
		\\\hline
loadData
	& no
		\\\hline
storedata
	& no
		\\\hline
extraCode
	& no
		\\\hline
extraData
	& no
		\\\hline
append
	& no
		\\\hline
replace 
	& no
		\\\hline
method ep
	& yes
		\\\hline
return ep
	& yes
		\\\hline
\end{tabular}
\begin{proofsketch}
This proof is split in 2 cases, one for each direction of the co-implication.

By \Cref{ass:comp-res} we know that $\surd$ is only produced by triggered checks.
\begin{description}
	\item[$\forall\src{\mc{C}},\comp{\src{\mc{C}}}=(\trg{m_c},\trg{m_d},\trg{t}). \trg{\OB{\alpha}\surd}\in\taim{\compaim{\src{\mc{C}}}} \Rightarrow \emul{\trg{\OB{\alpha}},\trg{t}}=\fail$.]

	By inspecting the pseudocode of entry points created by \prot{\cdot}, we can see what function aborted, causing a \trg{\surd} in the trace:
	\begin{description}
		\item[method entry point]:
		\begin{enumerate}
			\item[0m] it can abort if the jump does not come from \sys with return address being the forwrd return entry point
			\item[1m]\fun{dynamicTypechecks}{} can abort on ill-typed or unregistered parameters;
			\item[2m]\fun{loadObjects}{} can abort on fake or non-existing parameters (actually, it's lookup in \mc{T} that aborts);

		\end{enumerate}
		\item[forwardReturn entry point]:
		\begin{enumerate}
			\item[1r]\label{1r} it can abort if a return is made when no outcall was done (empty \mc{S});
			\item[2r]\label{2r} it can abort if the return is not made by the module that was called (mismatching \asm{r_0});
		\end{enumerate}
		\item[return entry point]:
		\begin{enumerate}
			\item[3r]\label{4r}\fun{dynamicTypechecks}{} can abort on an ill-typed return value;
			\item[4r] it can abort if the module jumping here is not \sys.
		\end{enumerate}
		\item[violation of the \PMA access control policy]:
		\begin{enumerate}
			\item[5r] a return to a previously-communicated address that is not executable
		\end{enumerate}
	\end{description}
	We can then analyse the action preceding the \trg{\surd}.

	The proof proceeds by induction on the size of \trg{\OB{\alpha}}.
	\begin{description}
		\item[$\trg{\OB{\alpha}}=\emptyset$] Trivial.
		\item[$\trg{\OB{\alpha}} = \trg{\alpha_1},\cdots,\trg{\alpha_n}$] where $\trg{n}\%2=1$ so $\trg{\alpha_n}=\trg{\gamma_n?}$.
		The proof proceeds by case analysis on \trg{\gamma_n?}.
		\begin{description}
			\item[\cl{\trg{a}\ \trg{w_0,\cdots,w_m}}]:

			Let us consider first the case when \trg{a} is a method entry point, considering the case of a successful forward by \sys.

			\trg{w_0} is the caller id placed by PMA, if it is not \sys, the execution aborts as in point 0m above, so \Cref{tr:emulcall-f2} \fail s.

			\trg{w_1}-\trg{w_5} are unused registers.

			\trg{w_6} is the current object, if that is not an encoding of an object that implements class \src{t} (Point 2m above), \Cref{tr:emulvals-intern-f} \fail s.

			\trg{w_i} with $i\in7..l$, if that is not a valid encoding of a value of type \src{t_i} (Points 1m and 2m above), then \Cref{tr:emulvals-f,tr:emulvals-bfail,tr:emulvals-externold-f,tr:emulvals-intern-f,tr:emobj2-f} \fail. 
			These rules cover all possible cases of ill-typed parameters for \src{\lstb{Unit}} (\Cref{tr:emulvals-f}), \src{\lstb{Bool}} (\Cref{tr:emulvals-bfail}), \src{\lstb{Obj}} (\Cref{tr:emobj2-f}), or class type (\Cref{tr:emulvals-externold-f,tr:emulvals-intern-f}).

			\trg{w_j} with $l<j<m$, those registers are not used.

			\medskip

			Let us consider when \trg{a} is an entry point of \sys; considerations for words that are not parameters (so \trg{w_1} to \trg{w_6} and \trg{9} and beyond) are analogous to those above

			\begin{itemize}
				\item \trg{a} is \fun{testObj}{}

				\trg{w_7} can be a non-existing object id, in which case \Cref{tr:emulcall-testobj-w} \fail s.


				There are no checks that the class encoding passed via \trg{w_8} is correct.
				When called from secure components, this (and the next) function are always called with correct class encodings.
				If the external code puts something else inside, it may try to cheat, but we'll find out because we test with correct class encodings.
				Otherwise, it could put anything inside, an object of a type that we don't know and don't use.
				In this case \sys would not be able to know this class.
				So we cannot filter based on known class encodings.

				\item \trg{a} is \fun{registerObj}{}

				\trg{w_7} can be an existing object id, in which case \Cref{tr:emulcall-regobj-w} \fail s.


				\item \trg{a} is \fun{forwardCall}{}

				\trg{w_0} can be the id of the module that last made a call (so no return was made), so \Cref{tr:emulcall-fwcall-w1} \fail s.

				\trg{w_3} can be the \sys module id, so someone is trying to forward a call to \sys via \sys, so \Cref{tr:emulcall-fwcall-w2} \fail s.

				\trg{w_5} is the address to return to, if that (plus the id \trg{w_0}) is not an executable address (Point 5r above), \Cref{tr:emulcall-f} \fail s.

				\item \trg{a} is \fun{forwardReturn}{}

				If the stack \mc{S} is empty, someone is returning without any call being made, so \Cref{tr:emulcall-fwret-w2} \fail s.

				If \trg{w_0} is not the module id of the last module that made a call, \Cref{tr:emulcall-fwret-w1} \fail s.
			\end{itemize}

			\item[\rb{\trg{a}\ \trg{w,w'}}]:

			If \trg{a} is not the address of \fun{forwardReturn}{}, \Cref{tr:emulretback-f2} \fail s.

			So \trg{a} let's consider \trg{a} to be the address of \fun{forwardReturn}{}.

			If the action's index $\trg{h}=1$, then no returns can be made (Point 1r above) and \Cref{tr:emulcall-fwret-w2} \fail s.

			Let the last called method have return type \src{t} and come from module with id \trg{id}.

			If \trg{w} is not a valid encoding of a value of type \src{t}, see previous case (Point 3r above).

			If \trg{w'} is not \trg{id}, \Cref{tr:emulretback-f} \fail s (Point 2r above).
		\end{description}
		Having covered all cases when a \trg{\surd} can be generated in traces, this case holds.
	\end{description}
	\item[$\forall\src{\mc{C}},\comp{\src{\mc{C}}}=(\trg{m_c},\trg{m_d},\trg{t}). \trg{\OB{\alpha}\surd}\in\taim{\compaim{\src{\mc{C}}}} \Leftarrow \emul{\trg{\OB{\alpha}},\trg{t}}=\fail$.]
	The proof proceeds by analysing all cases that generate a \fail (in \Cref{sec:emul}).
	\begin{description}
		\item[\Cref{tr:emulretback-f}]
		In this case the check inserted by \fun{checkShortcutting}{} ensures the \trg{\surd} is generated.
		\item[\Cref{tr:emulretback-f2}]
		In this case the check inserted by \fun{checkLocalStack}{} ensures the \trg{\surd} is generated.
		\item[\Cref{tr:emuladdrs-f}]
		The PMA architecture ensures the only addresses one can jump to is a valid entry point.
		\item[\Cref{tr:emulvals-f}]
		In this case the check inserted by \fun{dynamicTypechecks}{} ensures the \trg{\surd} is generated.
		\item[\Cref{tr:emulvals-bfail}]
		In this case the check inserted by \fun{dynamicTypechecks}{} ensures the \trg{\surd} is generated.
		\item[\Cref{tr:emulvals-externold-f}]
		In this case if the object id is a non-existant id, then \fun{loadObjects}{} will call lookup on \mc{T} that will abort, ensuring the \trg{\surd} is generated
		Otherwise, if the object id is ill-typed, \fun{dynamicTypechecks}{} calls \fun{testObject}{}, the check there ensures the \trg{\surd} is generated.
		\item[\Cref{tr:emulvals-intern-f}]
		In this case if the object id is a non-existant id, then \fun{loadObjects}{} will call lookup on \mc{T} that will abort, ensuring the \trg{\surd} is generated
		Otherwise, if the object id is ill-typed, \fun{dynamicTypechecks}{} calls \fun{testObject}{}, the check there ensures the \trg{\surd} is generated.
		\item[\Cref{tr:emobj2-f}]
		In this case if the object id is a non-existant id, then \fun{loadObjects}{} will call lookup on \mc{T} that will abort, ensuring the \trg{\surd} is generated
		Otherwise, if the object id is ill-typed, \fun{dynamicTypechecks}{} calls \fun{testObject}{}, the check there ensures the \trg{\surd} is generated.
		\item[\Cref{tr:intfors}]
		Because we are only considered memories where no symbols are present, nor passed, this is never triggered.
		\item[\Cref{tr:emulcall-f}]
		In this case, the violation of the \PMA access control policy ensures the \trg{\surd} is generated.
		\item[\Cref{tr:emulcall-f}] 
		In this case, the return jump to (\asm{r_0},\asm{r_5}) will be halted by the \PMA.
		\item[\Cref{tr:emulcall-fwcall-w1}]
		In this case the check by \fun{forwardCall}{} that the calling id is not the last id is triggered.
		\item[\Cref{tr:emulcall-fwcall-w2}]
		In this case the check by \fun{forwardCall}{} that the forwarding module id is not \trg{1} is triggered.
		\item[\Cref{tr:emulcall-fwret-w1}] 
		In this case the check by \fun{forwardReturn}{} that the returning id is not the last caller's one is triggered
		\item[\Cref{tr:emulcall-fwret-w2}] 
		In this case the check by \fun{forwardReturn}{} that the stack is empty is triggered.
		\item[\Cref{tr:emulcall-testobj-w}] 
		In this case the check by \fun{testObj}{} that the object to be registered is not fresh is triggered.
		\item[\Cref{tr:emulcall-regobj-w}] 
		In this case the check by \fun{registerObj}{} that the object to be tested exists is triggered.
	\end{description}
	Having considered all the cases when a \fail\ is generated, this case holds.
\end{description}
\end{proofsketch}

Proof of \Cref{thm:algocorr}.
\begin{proofsketch}
We need to prove that:
\begin{description}
\item[\skel{\src{C_1},\src{C_2}}] is correct and it produces well-typed code.
	
	This can be verified by inspecting the \skel{\cdot} definition in \Cref{sec:skel}.
\item[\emul{\trg{\OB{\alpha}},\trg{t}}] is correct, it produces well-typed code that correctly emulates all actions in \trg{\OB{\alpha}}.

	This proof proceeds by induction on the length of \trg{\OB{\alpha}}.
	The base case in void, so only the inductive case is presented.
	
	Assume $\trg{\OB{\alpha}} = \trg{\alpha_1},\cdots,\trg{\alpha_{k-1},\trg{\alpha_k}}$.
	Denote $\trg{\alpha_1},\cdots,\trg{\alpha_{k-1}}$ with \trg{\OB{\alpha}'}.
	The inductive hypothesis tells us that \emul{\trg{\OB{\alpha'}},\trg{t}} is correct, it produces well-typed code that correctly emulates all actions in \trg{\OB{\alpha'}}

	We need to prove that the emulation of \trg{\alpha_k} is correct.
	Formally, this is captured by: $\emul{\trg{\alpha_k}, M_k,V_k+V_{k-1}, \OB{\trg{id}_k}, \OB{\src{m}_k}, \OB{\src{t}_k, i_k}} = \src{A_k}, V'_k, \OB{\trg{id}_k}, i_k, \OB{\src{m}_k}, \OB{\src{t}_k}$.

	The proof proceeds by case analysis on \trg{\alpha_k}.
	\begin{description}
	\item[\trg{\cl{\trg{a}\ \trg{w_1,\cdots,v_n}}}] 

	\Cref{tr:emulcall} ensures that the call is emulated by a well-typed method call to the correct method.
	The correctness of the method is given by \Cref{lem:methemulcorr} while the correctness of the arguments and the current object is given by \Cref{lem:valemulcorr}.
	\item[\trg{\cb{\trg{a}\ \trg{w_1,\cdots,v_n}}}]

	\Cref{tr:emulcb} ensures that the state of \lst{oc} is updated correctly.
	\item[\trg{\rb{\trg{a}\ \trg{w},\trg{w'}}}]
	
	\Cref{tr:emulretback} ensures that this is emulated by returning the right value.
	The correctness of returned value is given by \Cref{lem:valemulcorr}.

	\item[\trg{\rt{\trg{a}\ \trg{w},\trg{w'}}}]

	\Cref{tr:emulret} ensures that the state of \lst{oc} is updated correctly.
	\item[\trg{\surd}] in this case the code is empty, as produced by \Cref{tr:em-f}.
	This is because the preceding trace must have generated a \fail, as guaranteed by \Cref{thm:maintheorem}.
	\end{description}

\item[\diff{\trg{\alpha_1},\trg{\alpha_2},i,M,V, \OB{\src{m}}, \OB{\src{t}}}] is correct, it produces well-typed code and it considers all possible cases where \trg{\alpha_1} and \trg{\alpha_2} can be different.

	This can be verified by inspecting the \diff{\cdot} definition in \Cref{sec:diff}.

	Possible cases of different actions:
	\begin{description}
	\item[Different action]:
		\begin{description}
		\item[call-return] \Cref{tr:diff-ac}
		\item[call-tick] \Cref{tr:diff-ac2}
		\item[return-tick] \Cref{tr:diff-ac3}
		\item[call-nothing] \Cref{tr:diff-ct}
		\item[return-nothing]  \Cref{tr:diff-rtt}
		\item[tick-nothing] \Cref{tr:diff-lt}
		\end{description}
	\item[Same action, different parameter]:
		\begin{description}
		\item[outcall]: \trg{\cb{a\ w_1,\cdots,w_n}}
			\begin{description}
			\item[different address \trg{a}] \Cref{tr:diff-cs}
			\item[different word \asm{w_0}] This is always \trg{1}
			\item[different word \asm{w_1} to \asm{w_4}] \Cref{tr:diff-r0}
			\item[different word \asm{w_5}] This is always 3*\trg{\mc{N}_w}
			\item[different word \asm{w_6}] \Cref{tr:diff-callee}
			\item[different word \asm{w_i}] with $i \in 7 .. n$  \Cref{tr:diff-cs2-p,tr:diff-cs2-i,tr:diff-cs2-e}
			\end{description}
		\item[return]: \trg{\rt{a\ w,w'}}
			\begin{description}
			\item[different address \trg{a}]  \Cref{tr:diff-ra}
			\item[different id \asm{id}] This is always \trg{1}
			\item[different word \asm{w}]  \Cref{tr:diff-rs,tr:diff-rsi,tr:diff-rse}
			\end{description}
		\end{description}
	\end{description}
\item[$\src{\mc{C}}\compat\src{\mc{C}_1}$ and $\src{\mc{C}}\compat\src{\mc{C}_2}$.]

	$\src{\mc{C}_1}$ and $\src{\mc{C}_2}$ have the same \lstb{import} requirements, otherwise it's trivial to differentiate them.
	\Cref{tr:emulate} and the definition of \skel{\cdot} ensure that all required interfaces are implemented in \src{\mc{C}} as well as all objects, so this points holds.
\end{description}
\end{proofsketch}

Proof of \Cref{thm:compaimsecure}.
\begin{proof}
Trace semantics deals with \emph{sets} of traces, while the algorithm inputs \emph{single} traces.
Moreover, these single traces must be the same up to a !-decorated action.
The two different single traces are obtained as follows.
Since $\compaim{\src{\mc{C}_1}}\nteqaim\compaim{\src{\mc{C}_2}}$, we have that $\taim{\compaim{\src{\mc{C}_1}}}\neq\taim{\compaim{\src{\mc{C}_2}}}$.
So there exists a trace \trg{\OB{\alpha}} that belongs to either only $\taim{\compaim{\src{\mc{C}_1}}}$ or only to $\taim{\compaim{\src{\mc{C}_2}}}$ but not to both. 
Assume wlog that $\trg{\OB{\alpha}}\in\taim{\compaim{\src{\mc{C}_1}}}$.
The trace \trg{\OB{\alpha}} can be split in two parts \trg{\OB{\alpha_s}} and \trg{\alpha_d} such that \trg{\OB{\alpha}}=\trg{\OB{\alpha_s}}\trg{\alpha_d} and such that \trg{\OB{\alpha_s}} is the longest prefix of all traces of \trg{\OB{\alpha_2}}.
Let $i$ be the index of \trg{\alpha_d} in \trg{\OB{\alpha}}.
So, there exists a trace $\trg{\OB{\alpha'}}\in\taim{\compaim{\src{\mc{C}_2}}}$ that can be split in two parts \trg{\OB{\alpha_s}} and \trg{\alpha_d'} such that \trg{\OB{\alpha'}}=\trg{\OB{\alpha_s}}\trg{\alpha_d'} and $\trg{\alpha_d}\neq\trg{\alpha_d'}$.
Trace \trg{\OB{\alpha'}} always exists, it could be an empty trace, it could be composed by an empty \trg{\OB{\alpha_s}} and, possibly, by an empty \trg{\alpha_d'}.

\Cref{thm:algocorr} tells us that we can use \algo{\cdot} to get the component that differentiates between \src{\mc{C}_1} and \src{\mc{C}_2}.
\end{proof}

Define $\src{\mc{C}}[\src{\mc{C}_1}]\divrjem\niff\src{\mc{C}}[\src{\mc{C}_2}]\divrjem$ as $\src{\mc{C}}[\src{\mc{C}_1}]\divrjem\wedge\src{\mc{C}}[\src{\mc{C}_2}]\termjem\text{ or }\src{\mc{C}}[\src{\mc{C}_2}]\divrjem\wedge\src{\mc{C}}[\src{\mc{C}_1}]\termjem$.

Proof of \Cref{thm:compaimisfa}.
\begin{proof}
The equivalence is split into two subgoals (directions).
Apply \Cref{thm:compaimcorr} to fulfill direction $\Leftarrow$.
The direction $\Rightarrow$ is restated using \Cref{thm:fatracesaim} into: $\forall\src{\mc{C}_1},\src{\mc{C}_2},\src{\mc{C}_1}\ceqjem\src{\mc{C}_2}\Rightarrow\compaim{\src{\mc{C}_1}}\teqaim\compaim{\src{\mc{C}_2}}$. 
The contrapositive of this direction is: $\forall\src{\mc{C}_1},\src{\mc{C}_2},\compaim{\src{\mc{C}_1}}\nteqaim\compaim{\src{\mc{C}_2}}\Rightarrow\src{\mc{C}_1}\nceqjem\src{\mc{C}_2}$.
Apply \Cref{thm:compaimsecure} to prove this direction.
\end{proof}

\begin{corollary}[Modular full-abstraction for compiled components]\label{thm:compaimhorizcomp-compiled}
$\forall\src{\mc{C}_1},\src{\mc{C}_2},\src{\mc{C}_3},\src{\mc{C}_4}.$ $\src{\mc{C}_1;\mc{C}_2}\ceqjem\src{\mc{C}_3;\mc{C}_4}\iff$ $\mylink{\compaim{\src{\mc{C}_1}},\compaim{\src{\mc{C}_2}}}\ceqaim\mylink{\compaim{\src{\mc{C}_3}},\compaim{\src{\mc{C}_4}}}$.
\end{corollary}
\begin{proof}
\Cref{thm:compaimhorizcomp} follows from \Cref{thm:compaimisfa} and modularity of \compaim{\cdot}.

Let $\src{\mc{C}_1} = \src{C_1^1},\cdots,\src{C_1^n}$, 
$\src{\mc{C}_2} = \src{C_2^1},\cdots,\src{C_2^m}$,
$\src{\mc{C}_3} = \src{C_3^1},\cdots,\src{C_3^l}$, 
$\src{\mc{C}_4} = \src{C_4^1},\cdots,\src{C_4^k}$.

By \Cref{def:ceqjem} we have that $n+m$ = $l+k$.

Let $\src{\mc{C}} = \src{C_1^1},\cdots,\src{C_1^n}$, $\src{C_2^1},\cdots,\src{C_2^m}$ and 
$\src{\mc{C}'} = \src{C_3^1},\cdots,\src{C_3^l}$, $\src{C_4^1},\cdots,\src{C_4^k}$.

\Cref{thm:compaimisfa} yields that, $\src{\mc{C}}\ceqjem\src{\mc{C}'} \iff \compaim{\src{\mc{C}}}\ceqaim\compaim{\src{\mc{C}'}}$.

By definition of \compaim{\cdot}, we have $\mylink{\compaim{\src{C_1^1}},\cdots,\compaim{\src{C_1^n}},\compaim{\src{C_2^1}},\cdots,\compaim{\src{C_2^m}}}\ceqaim$ $\mylink{\compaim{\src{C_3^1}},\cdots,\compaim{\src{C_3^l}},\compaim{\src{C_4^1}},\cdots,\compaim{\src{C_4^k}}}$, so the thesis holds.
\end{proof}

\begin{theorem}[Contextual equivalence is congruent with linking]\label{thm:link-pres-ceq}
$\forall \trg{P_1},\trg{P_2},\trg{P_3}.$ $\trg{P_2}\ceqaim\trg{P_3},$ $\mylink{\trg{P_1},\trg{P_2}}\ceqaim\mylink{\trg{P_1},\trg{P_3}}.$
\end{theorem}
\begin{proofsketch}
$\trg{P_2}\ceqaim\trg{P_3}$ tells us that \trg{P_2} and \trg{P_3} have the same modules.

If there is any overlapping module (except \sys) between \trg{P_1} and \trg{P_2} or \trg{P_3}, the result will be an empty memory.

If there is no overlap, then linking simply merges the modules, which trivially preserves $\ceqaim$.
\end{proofsketch}

Proof of \Cref{thm:compaimhorizcomp}.
\begin{proof}
By \Cref{thm:compaimhorizcomp-compiled}, we have that $\mylink{\compaim{\src{\mc{C}_1}},\compaim{\src{\mc{C}_2}}}\ceqaim\mylink{\compaim{\src{\mc{C}_3}},\compaim{\src{\mc{C}_4}}}$.

By \Cref{thm:link-pres-ceq} we have that $\mylink{\compaim{\src{\mc{C}_1}},\compaim{\src{\mc{C}_2}}}\ceqaim\mylink{\compaim{\src{\mc{C}_1}},\trg{P}}$.

By \Cref{thm:link-pres-ceq} we have that $\mylink{\compaim{\src{\mc{C}_3}},\compaim{\src{\mc{C}_4}}}\ceqaim\mylink{\compaim{\src{\mc{C}_3}},\trg{P'}}$.

By transitivity of contextual equivalence, the thesis holds.
\end{proof}

Proof of \Cref{thm:compaimsecure} 
\begin{proof}
$\compaim{\src{\mc{C}_1}}\nteqaim\compaim{\src{\mc{C}_2}}$ means that $\taim{\compaim{\src{\mc{C}_1}}}\neq\taim{\compaim{\src{\mc{C}_2}}}$, so there exists a trace $\trg{\OB{\alpha}\alpha_1}\in\taim{\compaim{\src{\mc{C}_1}}}$ but $\trg{\OB{\alpha}\alpha_1}\notin\taim{\compaim{\src{\mc{C}_2}}}$.

Let \trg{\OB{\alpha}} be the longest odd-length prefix of traces from \taim{\compaim{\src{\mc{C}_1}}} and \taim{\compaim{\src{\mc{C}_2}}} such that $\trg{\OB{\alpha}\alpha_1}\in\taim{\compaim{\src{\mc{C}_1}}}$ and $\trg{\OB{\alpha}\alpha_2}\in\taim{\compaim{\src{\mc{C}_2}}}$ and $\trg{\alpha_1}\neq\trg{\alpha_2}$.
The index of $\trg{\alpha_1}$ and $\trg{\alpha_2}$ is even, so they are generated by \compaim{\src{\mc{C}_1}} and \compaim{\src{\mc{C}_2}} respectively.
\Cref{thm:algocorr} tells us that \algo{\src{\mc{C}_1},\src{\mc{C}_2},\trg{\OB{\alpha}\alpha_1},\trg{\OB{\alpha}\alpha_2}} yields the context that differentiates between \src{\mc{C}_1} and \src{\mc{C}_2}.
\end{proof}

%% file: scmp-tr-full.bbl
\begin{thebibliography}{10}

\bibitem{abadiFa}
Mart\'{\i}n Abadi.
\newblock Protection in programming-language translations.
\newblock In {\em Secure Internet programming}, pages 19--34. Springer-Verlag,
  London, UK, 1999.

\bibitem{scoojoin}
Mart\'{\i}n Abadi, C{\'e}dric Fournet, and Georges Gonthier.
\newblock Secure implementation of channel abstractions.
\newblock {\em Information and Computation}, 174(1):37--83.

\bibitem{oracle}
Mart\'{\i}n Abadi, J{\'e}r{\'e}my Planul, and Gordon Plotkin.
\newblock Layout randomization and nondeterminism.
\newblock {\em Electr. Notes Theo. Comp. Sci.}, 298:29--50, 2013.

\bibitem{abadiLayout}
Mart\'{\i}n Abadi and Gordon~D. Plotkin.
\newblock On protection by layout randomization.
\newblock {\em ACM TISSEC}, 15(2):8:1--8:29, July 2012.

\bibitem{amalverifcomp}
Amal Ahmed.
\newblock {Verified Compilers for a Multi-Language World}.
\newblock In {\em SNAPL 2015}, volume~32 of {\em (LIPIcs)}, pages 15--31, 2015.

\bibitem{AhmedFa}
Amal Ahmed and Matthias Blume.
\newblock Typed closure conversion preserves observational equivalence.
\newblock {\em SIGPLAN Not.}, 43(9):157--168, September 2008.

\bibitem{ahmedCPS}
Amal Ahmed and Matthias Blume.
\newblock An equivalence-preserving {CPS} translation via multi-language
  semantics.
\newblock {\em SIGPLAN Not.}, 46(9):431--444, 2011.

\bibitem{intelattestation}
Ittai Anati, Shay Gueron, S~Johnson, and V~Scarlata.
\newblock {Innovative Technology for {CPU} Based Attestation and Sealing}.
\newblock In {\em HASP'13}, volume~13, 2013.

\bibitem{crashsafemachine}
Arthur Azevedo~de Amorim, Nathan Collins, Andr{\'e} DeHon, Delphine Demange,
  C\u{a}t\u{a}lin Hri\c{t}cu, David Pichardie, Benjamin~C. Pierce, Randy
  Pollack, and Andrew Tolmach.
\newblock A verified information-flow architecture.
\newblock {\em SIGPLAN Not.}, 49(1):165--178.

\bibitem{balto}
Ioannis~G. Baltopoulos and Andrew~D. Gordon.
\newblock Secure compilation of a multi-tier web language.
\newblock In {\em TLDI '09}, pages 27--38. ACM, 2009.

\bibitem{barthe}
Gilles Barthe, Tamara Rezk, Alejandro Russo, and Andrei Sabelfeld.
\newblock Security of multithreaded programs by compilation.
\newblock {\em ACM TISSEC}, 2010.

\bibitem{bistcc}
Nick Benton and Chung-Kil Hur.
\newblock Biorthogonality, step-indexing and compiler correctness.
\newblock {\em SIGPLAN Not.}, 44(9):97--108, August 2009.

\bibitem{nonintfree}
William~J. Bowman and Amal Ahmed.
\newblock Noninterference for free.
\newblock In {\em ICFP '15}, New York, NY, USA, 2015. ACM.

\bibitem{bugliesi}
Michele Bugliesi and Marco Giunti.
\newblock Secure implementations of typed channel abstractions.
\newblock In {\em POPL '07}, pages 251--262, 2007.

\bibitem{corin}
Ricardo Corin, Pierre-Malo Deni{\'e}lou, C{\'e}dric Fournet, Karthikeyan
  Bhargavan, and James Leifer.
\newblock A secure compiler for session abstractions.
\newblock {\em Journal of Computer Security}, 16(5):573--636, 2008.

\bibitem{intelsgxisa}
Victor Costan and Srinivas Devadas.
\newblock {Intel SGX Explained}.

\bibitem{faSemUML}
Frank~S. de~Boer, Marcello~M. Bonsangue, Martin Steffen, and Erika
  \'{A}brah\'{a}m.
\newblock A fully abstract semantics for {UML} components.
\newblock In {\em FMCO'04}.

\bibitem{popldomi}
Dominique Devriese, Marco Patrignani, and Frank Piessens.
\newblock Fully-abstract compilation by approximate back-translation.
\newblock In {\em {POPL} 2016}, pages 164--177, 2016.

\bibitem{dolevyao}
D.~Dolev and A.~C. Yao.
\newblock On the security of public key protocols.
\newblock In {\em SFCS}, pages 350--357, Washington, DC, USA, 1981. IEEE
  Computer Society.

\bibitem{fstar2js}
Cedric Fournet, Nikhil Swamy, Juan Chen, Pierre-Evariste Dagand, Pierre-Yves
  Strub, and Benjamin Livshits.
\newblock Fully abstract compilation to {JavaScript}.
\newblock {\em SIGPLAN Not.}, 48(1):371--384, January 2013.

\bibitem{faEHM}
Daniele Gorla and Uwe Nestman.
\newblock Full abstraction for expressiveness: History, myths and facts.
\newblock {\em Math Struct Comp Science}, 2014.

\bibitem{rdrand}
Mike Hamburg, Paul Kocher, and Mark Marson.
\newblock {Analysis of Intel's Ivy Bridge} digital random number generator.
\newblock {\em Cryptography Research, Inc.}, 2012.

\bibitem{KRLMLa}
Chung-Kil Hur and Derek Dreyer.
\newblock A {Kripke} logical relation between {ML} and assembly.
\newblock {\em SIGPLAN Not.}, 46(1):133--146, January 2011.

\bibitem{Jagadeesan}
Radha Jagadeesan, Corin Pitcher, Julian Rathke, and James Riely.
\newblock Local memory via layout randomization.
\newblock In {\em CSF '11}, pages 161--174. IEEE, 2011.

\bibitem{javaJr}
Alan Jeffrey and Julian Rathke.
\newblock {Java Jr.}: fully abstract trace semantics for a core {J}ava
  language.
\newblock In {\em ESOP'05}, pages 423--438, 2005.

\bibitem{jannis}
Yanis Juglaret and Catalin Hritcu.
\newblock {Secure compilation using micro-policies (Extended Abstract)}.
\newblock In {\em FCS 2015}, 2015.

\bibitem{catalin}
Yannis Juglaret, C\u{a}t\u{a}lin Hri\c{t}cu, Arthur {Azevedo de Amorim}, and
  Benjamin~C. Pierce.
\newblock Beyond good and evil: Formalizing the security guarantees of
  compartmentalizing compilation.
\newblock In {\em 29th IEEE Symposium on Computer Security Foundations (CSF)}.
  IEEE Computer Society Press, July 2016.
\newblock To appear.

\bibitem{leroy}
Xavier Leroy.
\newblock A formally verified compiler back-end.
\newblock {\em J. Autom. Reason.}, 43(4):363--446, December 2009.

\bibitem{trustVisor}
Jonathan~M. McCune, Yanlin Li, Ning Qu, Zongwei Zhou, Anupam Datta, Virgil
  Gligor, and Adrian Perrig.
\newblock Trustvisor: Efficient {TCB} reduction and attestation.
\newblock In {\em SP '10}, pages 143--158. IEEE, 2010.

\bibitem{intel}
Frank McKeen~et. al.
\newblock Innovative instructions and software model for isolated execution.
\newblock In {\em HASP '13}, 2013.

\bibitem{sancus}
Job Noorman~et. al.
\newblock Sancus: Low-cost trustworthy extensible networked devices with a
  zero-software {T}rusted {C}omputing {B}ase.
\newblock In {\em {USENIX}'13'}.

\bibitem{gcFA}
Joachim Parrow.
\newblock General conditions for full abstraction.
\newblock {\em Math. Struc. Comp. Sci.}

\bibitem{exprPA}
Joachim Parrow.
\newblock Expressiveness of process algebras.
\newblock {\em Electronic Notes in Theoretical Computer Science}, 209(0):173 --
  186, 2008.

\bibitem{tome-secure-compilation}
{Marco} {Patrignani}.
\newblock {\em {The Tome of Secure Compilation: Fully Abstract Compilation to
  Protected Modules Architectures}}.
\newblock PhD thesis, KU Leuven, Leuven, Belgium, May 2015.

\bibitem{scoo-j}
Marco Patrignani, Pieter Agten, Raoul Strackx, Bart Jacobs, Dave Clarke, and
  Frank Piessens.
\newblock {Secure Compilation to Protected Module Architectures}.
\newblock {\em ACM Trans. Program. Lang. Syst.}, 37(2):6:1--6:50, April 2015.

\bibitem{llfatr-j}
Marco Patrignani and Dave Clarke.
\newblock Fully abstract trace semantics for protected module architectures.
\newblock {\em Elsevier COMLAN}, 42(0):22 -- 45, 2015.

\bibitem{lcfConsidered}
Gordon~D. Plotkin.
\newblock {LCF} considered as a programming language.
\newblock {\em Theoretical Computer Science}, 5:223--255, 1977.

\bibitem{nizza}
Lenin Singaravelu, Calton Pu, Hermann H\"{a}rtig, and Christian Helmuth.
\newblock Reducing {TCB} complexity for security-sensitive applications: three
  case studies.
\newblock {\em SIGOPS Oper. Syst. Rev.}, 40(4):161--174, April 2006.

\bibitem{fidesArch}
Raoul Strackx and Frank Piessens.
\newblock Fides: Selectively hardening software application components against
  kernel-level or process-level malware.
\newblock In {\em CCS 2012}.

\bibitem{Sumii:2004:BDS:964001.964015}
Eijiro Sumii and Benjamin~C. Pierce.
\newblock A bisimulation for dynamic sealing.
\newblock In {\em Principles of Programming Languages}, pages 161--172. ACM,
  2004.

\bibitem{embed}
Nikhil Swamy, Cedric Fournet, Aseem Rastogi, Karthikeyan Bhargavan, Juan Chen,
  Pierre-Yves Strub, and Gavin Bierman.
\newblock Gradual typing embedded securely in javascript.
\newblock {\em SIGPLAN Not.}, 49(1):425--437, January 2014.

\bibitem{capsicum}
Robert N.~M. Watson, Jonathan Anderson, Ben Laurie, and Kris Kennaway.
\newblock Capsicum: Practical capabilities for {UNIX}.
\newblock In {\em USENIX}, pages 29--46, 2010.

\bibitem{cheri}
Jonathan Woodruff~et. al.
\newblock The {CHERI} capability model: Revisiting {RISC} in an age of risk.
\newblock In {\em ISCA '14}, pages 457--468, 2014.

\end{thebibliography}
